\newif\ifnotes
\newif\ifcr
\notesfalse
\crfalse

\newif\ifllncs
 \llncsfalse

\ifnotes
\newcommand{\omri}[1]{$\ll$\textsf{\color{blue} Omri: { #1}}$\gg$}

\else
\newcommand{\omri}[1]{}
\fi

\ifllncs
\documentclass[11pt]{llncs}
\else
\documentclass[a4paper,11pt]{article}
\usepackage{amsthm}
\usepackage{fullpage}
\fi
\usepackage{footmisc} 
\usepackage[hypertexnames=false,colorlinks=true,citecolor=Maroon]{hyperref}
\usepackage{newtxtext}
\usepackage{amsmath,amssymb,enumerate,algorithm,algorithmic,latexsym}
\usepackage{tikz}
\usetikzlibrary{fpu}
\usepackage{breakcites}
\usepackage{float}
\newfloat{algorithm}{H}{lop}
\usepackage{color}
\usepackage{colortbl}
\definecolor{Maroon}{cmyk}{0, 0.87, 0.68, 0.32}
\usepackage{url}
\usepackage{graphicx}
\usepackage{xspace}
\usepackage{subfigure}
\usepackage{amsfonts}
\usepackage{caption}
\usepackage{enumitem}
\usepackage{soul}
\usepackage[normalem]{ulem}
\usepackage{url}
\usepackage{graphicx}
\usepackage{bookmark}
\numberwithin{algorithm}{section}
\usepackage{mathtools}
\usepackage{dsfont}
\usepackage{cleveref}
\renewcommand{\paragraph}[1]{\vspace{1.5mm}\noindent \textbf{#1}}

\newcommand{\Nat}{\mathbb{N}}
\newcommand{\bbZ}{\mathbb{Z}}
\newcommand{\bbC}{\mathbb{C}}
\newcommand{\bbR}{\mathbb{R}}
\newcommand{\bbE}{\mathbb{E}}
\newcommand{\Var}{\text{Var}}
\newcommand{\poly}{\mathsf{poly}}

\newcommand{\OI}{\textbf{OI}}
\newcommand{\OOI}{\mathcal{O}_{\OI}}
\newcommand{\BQP}{\textbf{BQP}}
\newcommand{\NP}{\textbf{NP}}
\newcommand{\QCMA}{\textbf{QCMA}}

\newcommand{\BQPOI}{\textbf{BQP}^{\text{OI}}}
\newcommand{\GCVP}{\mathsf{GapCVP}}
\newcommand{\LWE}{\mathsf{LWE}}

\newcommand{\GI}{\mathsf{GI}}
\newcommand{\SZK}{\textbf{SZK}}
\newcommand{\SISD}{\mathsf{SISD}}
\newcommand{\SD}{\mathsf{SD}}
\newcommand{\Basis}{\mathbf{B}}
\newcommand{\Lattice}{\mathcal{L}}
\newcommand{\sVector}{\mathbf{s}}
\newcommand{\aVector}{\mathbf{a}}
\newcommand{\bVector}{\mathbf{b}}
\newcommand{\tVector}{\mathbf{t}}
\newcommand{\vVector}{\mathbf{v}}
\newcommand{\pVector}{\mathbf{p}}
\newcommand{\eVector}{\mathbf{e}}
\newcommand{\uVector}{\mathbf{u}}
\newcommand{\YES}{\prod_{\text{YES}}}
\newcommand{\NO}{\prod_{\text{NO}}}
\newcommand{\dist}{\text{dist}}
\newcommand{\QFT}{\mathsf{QFT}}

\newcommand{\ket}[1]{|{#1}\rangle}
\newcommand{\bra}[1]{\langle{#1}|}

\newcommand{\textabbrevstyle}[1]{\mbox{#1}}
\newcommand{\textabbrevstylebol}[1]{\mbox{\textbf{#1}}}
\newcommand{\newtextabbrev}[1]{\expandafter\newcommand\csname #1\endcsname{\textabbrevstyle{#1}\xspace}}
\newcommand{\newtextabbrevbol}[1]{\expandafter\newcommand\csname #1\endcsname{\textabbrevstylebol{#1}\xspace}}
\newcommand{\renewtextabbrevbol}[1]{\expandafter\renewcommand\csname
#1\endcsname{\textabbrevstylebol{#1}\xspace}}

\newtextabbrevbol{EXP}
\newtextabbrevbol{Dtime}
\newtextabbrev{PRG}
\newtextabbrev{PRF}
\newtextabbrev{OWF}
\newtextabbrev{DPT}
\newtextabbrev{PPT}
\newtextabbrev{QPT}
\newtextabbrev{EOWF}
\newtextabbrev{GEOWF}
\newtextabbrev{ZAP}
\newtextabbrevbol{Ntime}
\newtextabbrevbol{coAM}
\newtextabbrevbol{MIP}
\newtextabbrevbol{NEXP}
\newtextabbrevbol{BPP}
\newtextabbrevbol{coNP}
\newtextabbrevbol{coMA}
\newtextabbrevbol{TFNP}
\newtextabbrev{NIWI}
\newtextabbrevbol{AM}
\newtextabbrev{SNARG}
\newtextabbrev{SNARK}
\newtextabbrev{coRP}
\newtextabbrev{NIUA}
\newtextabbrev{UA}
\newtextabbrev{ZK}
\newtextabbrev{WZK}
\newtextabbrev{WH}
\newtextabbrev{AI}
\newtextabbrev{ECRH}
\newtextabbrev{PIR}
\newtextabbrev{WIPOK}
\newtextabbrev{CDS}
\newtextabbrev{POK}
\newtextabbrev{FE}
\newtextabbrev{IO}
\newtextabbrev{XIO}
\newtextabbrev{SXIO}
\newtextabbrev{SKFE}
\newtextabbrev{PKFE}
\newtextabbrev{PPRF}
\newtextabbrev{QLWE}
\newtextabbrev{PKE}

\ifllncs
\else
\newtheorem{definition}{Definition}[section]
\newtheorem{fact}{Fact}[section]
\newtheorem{lemma}{Lemma}[section]
\newtheorem{corollary}{Corollary}[section]
\newtheorem{theorem}{Theorem}[section]
\newtheorem{claim}{Claim}[section]

\theoremstyle{remark}

\newtheorem{remark}{Remark}[section]

\fi

\newenvironment{boxfig}[2]{\begin{figure}[#1]\fbox{\begin{minipage}{\linewidth}
                        \vspace{0.2em}
                        \makebox[0.025\linewidth]{}
                        \begin{minipage}{0.95\linewidth}
            {{
                        #2 }}
                        \end{minipage}
                        \vspace{0.2em}
                        \end{minipage}}}{\end{figure}}

\DeclarePairedDelimiter{\ceil}{\lceil}{\rceil}

\newcommand{\norm}[1]{\ensuremath{\left\lVert #1 \right\rVert}} %
\renewcommand{\paragraph}[1]{\smallskip\noindent{\bf #1}}

\title{Quantum Algorithms in a Superposition of Spacetimes}

\ifllncs
\else

\author{Omri Shmueli\thanks{Tel Aviv University, \texttt{omrishmueli@mail.tau.ac.il}. Supported in part by the European Research Council (ERC) under the European Union’s Horizon Europe research and innovation programme (grant agreement No. 101042417, acronym SPP), and by the Clore Israel Foundation.}}

\date{\vspace{-5ex}}

\date{}

\fi

\begin{document}

\maketitle

\vspace{-1.1cm}

\begin{abstract}
Quantum computers are expected to revolutionize our ability to process information. The advancement from classical to quantum computing is a product of our advancement from classical to quantum physics -- the more our understanding of the universe grows, so does our ability to use it for computation. A natural question that arises is, what will physics allow in the future? Can more advanced theories of physics increase our computational power, beyond quantum computing?

An active field of research in physics studies theoretical phenomena outside the scope of explainable quantum mechanics, that form when attempting to combine Quantum Mechanics (QM) with General Relativity (GR) into a unified theory of Quantum Gravity (QG). One expected phenomenon in QG is quantum uncertainty in the structure of spacetime: QM and GR together may imply that spacetime can be in superposition of curvatures. Under GR, gravitational time dilation asserts that the geometry of spacetime determines the order of events occurring. Accordingly, a superposition of curvatures is known to present the possibility of a superposition of event orderings.

In the literature of quantum information theory, the most natural model for a superposition of event orders, is a superposition of unitary evolution orders: For unitaries $U_0$, $U_1$ and state $\ket{\psi}$, execution in a superposition of orders aims to generate the normalization of the state $U_1 U_0 \cdot \ket{\psi} + U_0 U_1 \cdot \ket{\psi}$. Costa (Quantum 2022) proves that a model with such power cannot produce a valid process matrix (a generalization of standard quantum processes), and leaves the analysis of any relaxation of the model as an open question.

\paragraph{In this work} we show a first example of a natural computational model based on QG, that provides an exponential speedup over standard quantum computation (under standard hardness assumptions).
Formally, we define a relaxation of the previous model of superposition of unitary orders, where we do allow the generation of the state $U_1 U_0 \cdot \ket{\psi} + U_0 U_1 \cdot \ket{\psi}$, but with computational complexity inversely proportional to the state's norm (unlike the previous model, considering unconditional generation). We further provide physical intuition and assumptions behind our model.

We show that a quantum computer with the ability to create a superposition of unitary orders is able to solve in polynomial time two fundamental problems in computer science: The Graph Isomorphism Problem ($\GI$) and the Gap Closest Vector Problem ($\GCVP$), with gap $O\left( n\sqrt{n} \right)$. These problems are believed by experts to be hard to solve for a regular quantum computer. Interestingly, our model does not seem overpowered, and we found no obvious way to solve entire complexity classes that are considered hard in computer science, like the classes $\NP$ and $\SZK$. As part of our work we develop a new parameterization technique for the invertibility of classical circuits, which raises new questions in the analysis of probability distributions and induces a new parameter for hardness inside the complexity class $\SZK$. These new techniques are independent of any non-classical computation.
\end{abstract}

\ifllncs
\pagestyle{plain}
\else

\thispagestyle{empty}
\newpage
\tableofcontents
\newpage
\thispagestyle{empty}

\fi

\pagenumbering{arabic}

\section{Introduction} \label{section:introduction}
Theoretical physics is in a quest to draw a blueprint for the universe. The constructed different theories, which are sets of rules postulating how the universe works in certain scenarios, are our tool for predicting what will be the next step of a physical process. The two most accurate theories to date are Quantum Mechanics (QM) and General Relativity (GR). QM (along with Quantum Field Theory) describes the interaction between three out of the four fundamental forces in nature: Strong nuclear force, weak nuclear force and electromagnetism. GR describes the fourth force -- gravity. Both QM and GR are immensely successful theories, in understanding vastly different physical phenomena.

Formulating an experimentally-verified theory of physics that generalizes both QM and GR, is one of the grand challenges of physics. Such a theory, usually referred to as a theory of Quantum Gravity, will model the interaction between quantum effects and gravitational force; an interaction which is currently ignored in most analyses. The two most popular candidates as theories of quantum gravity are String Theory and Loop Quantum Gravity, however, numerous other proposals exist \cite{de2022frontiers}.

Moving from physics to computer science, we observe an interconnectedness. On the one hand, theories of physics dictate the potential computational power in humanity's hand. If a theory of physics allows for a machine with new information processing capabilities, this sometimes implies algorithmic breakthroughs: For example, quantum mechanics allows for a quantum computer, which in turn is needed to run Shor's algorithm \cite{shor1999polynomial}. Shor's algorithm factors integers in polynomial time, whereas our fastest (known) factoring algorithms on any classical computer run in sub-exponential time. In practical terms, the meaning of these timescales is the difference between possible and impossible.

On the other hand, computational insights shed light back on physics. To see this, consider two physical theories $A$ and $B$, where $B$ claims to be a more elaborate description of the universe and generalize theory $A$. In that case, $B$ can not only explain all physical processes that $A$ can, but can accurately describe processes which $A$ cannot. There is a subtlety here; While we can encounter a physical process that theory $A$ \emph{currently} does not know how to explain, it does not necessarily mean that there is no explanation -- it might exist, and we just haven't found it yet. Proving that there is no explanation, however, is a task of different magnitude. Turning to computer science for help, we can construct two computational models corresponding to the two theories (i.e., a strongest computer $C_A$ that can be constructed according to theory $A$ and a computer $C_B$ that can be constructed according to theory $B$). Now, if we find a computational problem $P$ that $C_B$ can solve but $C_A$ cannot, this gives strong evidence that theory $B$ can explain physical phenomena that $A$ cannot, namely, the solving of $P$!
Given the correspondence between algorithms and theoretical physics, a natural question which arises is the following: Just like the progression from classical to quantum mechanics allowed computational transcendence, how can a theory of quantum gravity transform our theory of computation?

\paragraph{Superposition of Spacetime Geometries.}
While a confirmed theory of quantum gravity remains elusive, there are some physical phenomena expected to take place within \emph{any} valid quantum gravity theory, in some form or another. Arguably the most basic expected phenomenon unique to quantum gravity, is the possibility of quantum uncertainty in the structure spacetime itself \cite{ashtekar1996large}. More elaborately, according to QM, all dynamic quantities are subject to quantum superposition. GR tells us that the geometry of spacetime is not fixed, and a straightforward combination implies the possibility for a superposition of spacetime geometries \cite{ashtekar1996large, giacomini2019quantum, christodoulou2019possibility, giacomini2020einstein, giacomini2021spacetime, giacomini2022quantum}.

\paragraph{Structure of the Introduction.}
The goal of this paper is to suggest and define a computational model for a quantum computer with the added ability of generating a superposition of spacetime geometries, and to analyze its power.
The rest of the Introduction is as follows.
In Section \ref{subsection:intro_gravity_superposition_entanglement} we provide a brief overview of the physics literature on a superposition of spacetime geometries, and then explain how it serves as inspiration for our computational model. We also provide the physical assumptions our model makes.
Section \ref{subsection:intro_gravity_superposition_entanglement} of the introduction deals with physics, it is self contained and can be skipped for the readers that are interested only in the complexity theoretical aspects of this work.
In Section \ref{subsection:intro_model} we present our computational model and its connection to both standard quantum computation and previous generalized models of quantum information processing. In Section \ref{subsection:intro_results} we present our complexity theoretical results for the model.

\subsection{Gravity, Superposition and Entanglement} \label{subsection:intro_gravity_superposition_entanglement}
Before we explain what a quantum spacetime may be able to do, we refer to one ability of a \emph{classical} spacetime. GR postulates the effects of gravitational time dilation, where the structure of spacetime determines the pace at which time flows. More precisely, gravitational time dilation is the effect of time moving slower in a region of space with greater gravitational potential: For some location $s \in \bbR^{3}$ in space, the time in $s$ slows down proportionally to how close $s$ is to a massive object, and the mass of that object. According to this logic, a quantum spacetime may enable a superposition of spacetime curvatures, which in turn induces a superposition of time dilation effects, on the same set of events.

\paragraph{Gravitational Decoherence Versus Gravitationally-Induced Entanglement.}
Let us be more concrete, by using (with small variations) a thought experiment by Zych, Costa, Pikovski and Brukner \cite{zych2019bell} (depicted in Figure 1 from \cite{zych2019bell}). In the experiment there are two \emph{spatially isolated} systems; one contains a gravitationally significant object $M$, and the second has two clocked unitary circuits $U_0$, $U_1$. More precisely,
\begin{itemize}
    \item
    In system $S_{M}$ we can generate a uniform superposition of two possible locations $L_0$, $L_1$ of a mass $M$. Note that we are using the fact that quantum theory does not prohibit superpositions of massive objects.

    \item 
    In system $S_{U}$ there is a quantum register $R$ in some state $\ket{\psi}$, two atomic clocks $C_0$, $C_1$, and a central processing unit $C_{c}$ that gets signals from the clocks. $C_{c}$ always executes $U_0$ and $U_1$ on $R$, and the order in which it executes the two unitary circuits depends on the signal from which clock, $C_0$ or $C_1$, arrives first. For the sake of order, system $S_{U}$ is constructed such that clock $C_1$ takes a bit more time (i.e., needs to make more clock ticks), so that when no disturbances are applied to system $S_{U}$, the order in which $C_{c}$ executes the unitaries is always $U_0$ first and then $U_1$. Such a system should also be realizable by quantum mechanics.

    \item 
    The last detail is the only one that uses effects from GR, specifically, gravitational time dilation. The system $S_{M}$ is located with respect to the second system $S_{U}$ such that, even though the systems are spatially disjoint in coordinates, if $M$ is at location $L_0$, which is closer to clock $C_1$, then $C_1$ ticks slower and $C_0$ concludes first, and if $M$ is at location $L_1$, which is closer to $C_0$, then $C_0$ ticks slower and $C_1$ concludes first.
\end{itemize}

Imagine we (1) execute $S_{M}$ and generate a uniform superposition of the two possible locations for $M$, followed by (2) executing $S_{U}$, the clocked execution of the unitaries $U_0$, $U_1$. We can analyze what happens in each of the branches of the superposition, separately, assuming as if spacetime is classical and there is no superposition of the locations of $M$. If $M$ is at $L_0$, the state in $R$ at the end of the execution is $U_{1} U_{0} \ket{\psi}$, and if $M$ is at $L_1$ then the state in $R$ at the end of execution is $U_{0} U_{1} \ket{\psi}$. The question becomes one in quantum gravity, when we ask what happens not in each branch of the superposition separately, but as a full joint system, in superposition. Experimental physics does not know the answer to this question.

The above question may translate to more basic questions about gravity, superposition and entanglement. Specifically, the first question is whether a massive object can be maintained in a quantum superposition, and the second question is whether entanglement can be created through gravity alone. To elaborate, on one side, the gravitational decoherence hypothesis \cite{diosi1989models, penrose1996gravity, diosi2014gravitation, diosi1987universal} says that a quantum superposition of a massive object will fundamentally decohere, turning into a probabilistic mixture of classical states. On the other end, gravitationally-induced entanglement (GIE) \cite{bose2017spin, marletto2017gravitationally} is hypothesizing that not only a massive object can be in a superposition, but that the gravitational field can be used to mediate entanglement.

In the context of the thought experiment from \cite{zych2019bell}, gravitational decoherence says that when executing the system $S_{M}$, there will not be a coherent superposition for the locations of $M$, so with probability $\frac{1}{2}$ the mass will be at $L_0$ and with the remaining probability at $L_1$. The end result of the execution of the two systems, in that case, is the probabilistic uniform mixture (or in other words, uniform mixed state) between
$$
\ket{L_{0}}_{S_{M}} \otimes U_{1} U_{0} \ket{\psi}_{S_{U}}
\; , \;
\ket{L_{1}}_{S_{M}} \otimes U_{0} U_{1} \ket{\psi}_{S_{U}}
\enspace ,
$$
where for $b \in \{ 0, 1 \}$, $\ket{L_{b}}$ is the state where the mass $M$ is at location $L_{b}$, and we assume that the states of the different locations are orthogonal (that is, $\bra{L_{0}}\ket{L_{1}} = 0$). If GIE is true, then the per-branch analysis from before holds, only that we get the quantum state. Assuming GIE, at the end of the execution of the joint system, the quantum state is the pure state
\begin{equation} \label{equation:gie_state}
    \frac{1}{\sqrt{2}} \cdot \ket{L_{0}}_{S_{M}} \otimes U_{1} U_{0} \ket{\psi}_{S_{U}}
    +
    \frac{1}{\sqrt{2}} \cdot \ket{L_{1}}_{S_{M}} \otimes U_{0} U_{1} \ket{\psi}_{S_{U}}
    \enspace .
\end{equation}

As mentioned earlier, using the effects of gravitational time dilation, gravity can be used to create correlations between two spatially isolated systems\footnote{Recall that this is exactly what happens in the above experiment -- depending on the location of $M$ in the system $S_{M}$, this determines a different spacetime geometry in the system $S_{U}$, which in turn cause the unitaries $U_0$, $U_1$ to be executed in one order or the other. This in particular creates a correlation between the states of the two systems.}. In the scope of this paper we call this phenomenon Gravitational Correlation. In the case of executing the experiment under gravitational decoherence, we can think of the result as a Classical Gravitational Correlation. In the case of executing the experiment under GIE, we consider the result as a Quantum Gravitational Correlation. Up to date and despite considerable efforts, no experiment exists that either proves or refutes gravitational decoherence or gravitationally-induced entanglement (that is, a reproducible experiment with statistical significance).

\paragraph{A complexity theoretical perspective.}
Both classical and quantum gravitational correlations are efficiently simulatable in the quantum circuit model. This is a basic exercise in quantum information processing: We require the information of the mass $M$ (but not an actual massive body). By using controlled versions of the unitary circuits $U_{0}$, $U_{1}$, we can create the same state from the GIE execution \ref{equation:gie_state}. Then, if we want the state from the case of gravitational decoherence, we only need to measure the left register (containing the system $S_{M}$) in the computational basis.

Note that the two options of either a classical or quantum gravitational correlation (corresponding to gravitational decoherence or GIE, respectively), do not cover the entire field of possibilities for the outcome of the experiment. There is a possibility of quantum gravitational correlation \emph{without entanglement} (quantum correlations are known to be possible without entanglement, e.g., quantum discord \cite{ollivier2001quantum}). That is, a massive object can be maintained in a coherent superposition, as opposed to gravitational decoherence, but gravity alone cannot create entanglement, as opposed to GIE.
Indeed, some preliminary results suggest that even if spacetime is quantum, this does not necessarily implies it has to mediate full entanglement \cite{belenchia2018quantum, belenchia2019information, bera2019quantifying, danielson2022gravitationally, sugiyama2023quantum}. 

\paragraph{This work - quantum gravitational correlation without entanglement.}
The question we focus on in this research is, what are the computational implications of the the last scenario? Assuming quantum gravitational correlations without (full) entanglement, what happens when the system $S_{U}$ executes, while the system $S_{M}$ finished executing and is in the state $\frac{1}{\sqrt{2}} \cdot \ket{L_{0}}_{ S_{M} } + \frac{1}{\sqrt{2}} \cdot \ket{L_{1}}_{ S_{M} }$ ? Our computational model makes the following assumptions.
First, if quantum gravitational correlations are possible but do not force entanglement, when executing the experiment, this should create a joint state that resembles the un-normalized state
$$
\left(
\ket{L_{0}} + \ket{L_{1}}
\right)_{ S_{M} }
\otimes 
\left(
U_{1} U_{0} \ket{\psi} + U_{0} U_{1} \ket{\psi}
\right)_{ S_{U} }
\enspace .
$$
We observe that the norm of the state $U_{1} U_{0} \ket{\psi} + U_{0} U_{1} \ket{\psi}$ fluctuates as a function of all components $U_{0}$, $U_{1}$ and $\ket{\psi}$. This means that (similarly to calculations stemming from general relativity in quantum gravity) such transformation will not be linear, due to the need for a state re-normalization.
Next, we imagine that in case that the physical universe allows to execute transformations $\left( U_{0}, U_{1}, \ket{\psi} \right) \rightarrow U_{1} U_{0} \ket{\psi} + U_{0} U_{1} \ket{\psi}$, it is numerically intuitive that decoherence will be inversely proportional to the norm of the state, and that energy could be invested to maintain coherence, which should also be inversely proportional to the norm.
Finally, while we assume the ability to invest energy to maintain coherence of a superposition of spacetimes, we do not assume that matter-wave interference between different spacetimes to be necessarily identical to quantum interference within a single spacetime\footnote{Some theoretical results show evidence that standard interference dynamics is not necessarily the case when considering a superposition of spacetimes \cite{foo2021schrodinger, foo2022quantum, foo2023quantum}.}.
Later when we define our computational model, these assumptions will be more formally expressed.

\subsection{Our Computational Model - Computable Order Interference} \label{subsection:intro_model}
We now move from informal discussions in physics to a formal representation of our model.
In the literature of quantum information theory, the phenomenon of uncertainty of the background spacetime is referred to as an indefinite causal structure (ICS) \cite{hardy2007towards, hardy2005probability, myrvold2009quantum, oreshkov2012quantum, chiribella2013quantum, leifer2013towards, brukner2014quantum}. As the structure of spacetime determines the order of events happening (but not necessarily the set of events happening), ICS studies the possibility of quantum uncertainty of event occurrence order. In quantum information theory, an event is the time-evolution of a physical system, or equivalently, the execution of a quantum circuit on a quantum register. Following these baselines, the most natural interpretation of ICS is a quantum superposition of unitary execution orders \cite{oreshkov2012quantum, chiribella2013quantum}. Given two $n$-qubit unitary circuits $U_0$, $U_1$ and an $n$-qubit state $\ket{\psi}$, execution in ICS can be imagined as a superposition between two scenarios -- one scenario where $U_0$ is executed followed by an execution of $U_1$, generating the state $U_1 U_0 \cdot \ket{\psi}$, and a second scenario with flipped execution order, generating $U_0 U_1 \cdot \ket{\psi}$. Executing $U_0$, $U_1$ on $\ket{\psi}$ with indefinite causal structure should attempt to create the state $U_1 U_0 \cdot \ket{\psi} + U_0 U_1 \cdot \ket{\psi}$.

There is a crucial difference between a \emph{pure} superposition of unitary orders, which takes the above form $U_1 U_0 \cdot \ket{\psi} + U_0 U_1 \cdot \ket{\psi}$, and a mixed superposition, entangled with an auxiliary register, which takes the form $\ket{0} U_1 U_0 \cdot \ket{\psi} + \ket{1} U_0 U_1 \cdot \ket{\psi}$. First, the latter can be easily generated in the standard quantum circuit model, using controlled versions of $U_0$, $U_1$. Second, the former is generally not even a unitary transformation. More so, Costa \cite{costa2022no} shows that such transformation cannot produce a Process Matrix, which is, in a nutshell, a non-causally ordered generalization of quantum states \cite{chiribella2013quantum}. To elaborate, the Process Matrix Formalism \cite{chiribella2013quantum} (PMF) is a mathematical framework for defining quantum processes that are not constrained by a causal execution order.
The PMF defines what is a valid process matrix $W$, and while the direct physical meaning of process matrices is currently not well understood\footnote{In particular, the process matrix formalism is not known to be connected to gravitational interaction or general relativity, and is written only in the language of quantum information processing.}, the formalism remains an important tool. In particular, as in the case of \cite{costa2022no}, the formalism can be used to \emph{exclude} some transformations from being considered valid quantum transformations, in the more broad sense that is free of causal restrictions (that is, a transformation which is not a process matrix is in particular not a valid quantum process in standard quantum mechanics). Costa leaves open the question of understanding the computational power of any relaxation of the model of superposition of unitary orders.

\paragraph{Execution of superposition of unitary orders as a function of interference norm and phase alignment.}
We start our work by exploring more carefully the idea of a superposition of unitary orders, where our goal is to define a natural relaxation of it. From hereon we call the state $U_{1} U_{0} \ket{\psi} + U_{0} U_{1} \ket{\psi}$ the \emph{order interference} state, or OI state for short\footnote{As we mentioned before, generating the state $\ket{0} U_1 U_0 \cdot \ket{\psi} + \ket{1} U_0 U_1 \cdot \ket{\psi}$, which can be thought of as a superposition of evolution orders, is easy in quantum mechanics. However, while there is superposition, there is no interference between different evolution orders. Unlike some of the previous work on the subject (e.g. Quantum Switch \cite{chiribella2013quantum}), we view not the superposition, but the interference between different evolution orders as a possibly unique idea to quantum gravity.}. We can think of order superposition generalized to an arbitrary number $m \in \Nat$ of unitaries $\{ U_{i} \}_{i \in [m]}$, and the OI state is defined as
\begin{equation} \label{equation:introduction_OI_state}
    \OI
    \left(
    \{ U_{i} \}_{i \in [m]},
    \ket{\psi}
    \right)
    :=
    \sum_{\sigma \in S_{m}} \left( \prod_{i \in [m]} U_{\sigma^{-1}\left( i \right)} \right) \cdot \ket{\psi} \enspace ,
\end{equation}
where $S_{m}$ is the set of all permutations on the set $[m] := \{ 1, 2, 3, \cdots, m \}$, and the above product notation $\prod_{i \in [m]}$ creates multiplications between matrices enumerated from right to left, which is also the convention for matrix multiplication in the rest of this work, that is: 
$$
\prod_{i \in [m]} U_{\sigma^{-1}\left( i \right)}
:=
U_{\sigma^{-1}\left( m \right)} \cdot U_{\sigma^{-1}\left( m - 1 \right)} \cdots
U_{\sigma^{-1}\left( 2 \right)} \cdot U_{\sigma^{-1}\left( 1 \right)} \enspace .
$$

For $m$ unitaries, observe that the norm of the OI state ranges in $[0, m!]$. We view the norm of the OI state as a fundamental piece of information. As an example, in the case $\ket{\psi} = \ket{0}$, $m = 2$, $U_0 = X$, $U_1 = Z$, the norm is $0$:
$$
X\cdot Z\ket{0} + Z\cdot X\ket{0} = \ket{1} + (-1)\cdot \ket{1} = 0 \enspace .
$$
On the other hand, for every number of unitaries $m \in \Nat$ and state $\ket{\psi}$, if $U_1 = U_2 = \cdots = U_{m}$, one can verify that $m!$ is the norm.

The previous work of \cite{costa2022no} checks whether (according to the PMF) it is possible to always and unconditionally execute a superposition of unitary orders, for every unitaries $\{ U_{i} \}_{i \in [m]}$ and state $\ket{\psi}$. In our relaxed model, we have the following main differences. We take into consideration computational complexity, and specifically, the computational complexity of generating the OI state is determined by two properties of the interference\footnote{These complexity measures are tied to the physical assumptions given at the end of Section \ref{subsection:intro_gravity_superposition_entanglement}.}: (1) The norm of the OI state, and (2) The amount of phase alignment between execution orders (or spacetimes). More precisely, for every classical state $\ket{x}$ (for $x \in \{ 0, 1 \}^{n}$), we consider the phase and magnitude of $\ket{x}$ in the resulting state $\left( \prod_{ i \in [m] } U_{ \sigma^{-1}\left( i \right) } \right) \cdot \ket{\psi}$ of every execution order $\sigma \in S_{m}$. The amount of phase alignment is how similar are the phases of the same $\ket{x}$, between the results of differing execution orders, over all $x \in \{ 0, 1 \}^{n}$.

\paragraph{Connection to standard quantum computation.} 
Our next observation is that when trying to execute OI transformations in standard quantum mechanics, the consideration of the norm is in fact expressed (a limited phase alignment implies reduced norm, thus phase alignment possibly also plays a role).
More formally, we show there is a quantum algorithm $Q$ that given input unitary circuits $\{ U_{i} \}_{i \in [m]}$ (all act on the same number of qubits $n$) and an $n$-qubit state $\ket{\psi}$, executes in time $\poly\left( \sum_{i \in [m]}|U_{i}| \right)$, where $|U_{i}|$ is the circuit size of $U_{i}$, and outputs the normalized OI state $\OI\left( \{ U_{i} \}_{i \in [m]}, \ket{\psi} \right)$ with success probability $\left( \frac{ \norm{ \OI\left( \{ U_{i} \}_{i \in [m]}, \ket{\psi} \right) } }{ m! } \right)^{2}$. We elaborate formally on the above quantum procedure in Section \ref{subsection:OI_comparison_to_standard_quantum_computation}. There is no known way in quantum mechanics (and in particular, in the standard quantum circuit model) to generally and efficiently amplify the success probability of executing an OI transformation.

\paragraph{Computable Interference of Unitary Orders - Definition.}
Rather than either the bounded interference in quantum mechanics or an unconditional interference, in this work we investigate the computational power of quantum computing with \emph{computable interference} of unitary orders. In our model, while the probability to succeed in executing order interference is proportional to the norm of the OI state (as per standard quantum mechanics), we allow an investment of computational work (or energy) to amplify the success probability arbitrarily close to $1$ (up to phase misalignment). The ability to execute computable order interference is formally captured by an object we define, called an Order Interference (OI) oracle $\OOI$. The full definition (Definition \ref{definition:computable_order_interference_oracle}) of an OI oracle is given in Section \ref{subsection:new_definitions_order_interference}, and for the sake of this Section \ref{section:introduction} and Section \ref{section:overview}, it will be sufficient to focus on a relaxed variant of it (given in a Corollary \ref{corollary_definition:oi_pure_states} of the definition).

The relaxed definition of an OI oracle $\OOI$ is as follows. Let $\{ U_{i} \}_{i \in [m]}$ a set of $m \in \Nat$ unitary quantum circuits, all operating on the same number of qubits $n \in \Nat$, let $\ket{\psi}$ an $n$-qubit quantum state and let $\lambda \in \Nat$. An input to $\OOI$ is a triplet
$$
\left(
\{ U_{i} \}_{i \in [m]},
\ket{\psi}, 
\lambda
\right) \enspace ,
$$
and takes time complexity $O \left( \lambda + \sum_{i \in [m]}|U_{i}| \right)$. For $x \in \{ 0, 1 \}^{n}$, $\sigma \in S_{m}$ we write
$$
\prod_{i \in [m]} U_{\sigma^{-1}\left( i \right)} \cdot \ket{\psi}
:=
\sum_{ x \in \{ 0, 1 \}^{n} } \alpha_{x, \sigma} \ket{x}
\enspace .
$$
At the end of execution, with probability
\begin{equation} \label{equation:introduction_ICS_success_probability}
    \left(
    \frac
    { \sum_{ x \in \{ 0, 1 \}^{n} } |\sum_{ \sigma \in S_{m} } \alpha_{x, \sigma}| }
    { \sum_{ x \in \{ 0, 1 \}^{n} } \sum_{ \sigma \in S_{m} } |\alpha_{x, \sigma}| }
    \right)
    \frac
    { \frac{ \norm{\OI\left( \{ U_{i} \}_{i \in [m]}, \ket{ \psi } \right)} }{ m! } }
    { \frac{ \norm{\OI\left( \{ U_{i} \}_{i \in [m]}, \ket{ \psi } \right)} }{ m! } + \frac{1}{\lambda} } 
    \enspace ,
\end{equation}
the process measures a "success" signal for the transformation and obtains the normalization of $\OI\left( \{ U_{i} \}_{i \in [m]}, \ket{\psi}\right)$, and with the remaining probability the process measures a "fail" signal and obtains $\ket{0^n}$.



\subsection{Results} \label{subsection:intro_results}
We define a new complexity class named $\BQPOI$ (Definition \ref{definition:bqp_oi}), which stands for $\BQP$ (bounded-error quantum polynomial-time) with access to an Order Interference oracle $\OOI$. $\BQPOI$ is defined as the class of problems for which there exists a quantum algorithm with access to an OI oracle, such that the algorithm solves the problem (with bounded error) running in polynomial time in the input size.

We focus on two very well-studied computational tasks in computer science: The Graph Isomorphism problem ($\GI$, Definition \ref{definition:graph_isomorphism_problem}) and the Gap Closest Vector Problem ($\GCVP_{g(n)}$, Definition \ref{definition:gap_closest_vector_problem}). Both $\GI$ and $\GCVP_{ g(n) }$ for $g(n) := O\left( n\sqrt{n} \right)$ are believed by experts to be unsolvable in quantum polynomial time. In particular, $\GCVP_{ O\left( n\sqrt{n} \right) }$ is a key problem in the theory of computation and have an immense practical significance: Solving $\GCVP_{ O(n\sqrt{n}) }$ means solving the Learning With Errors ($\LWE$) problem \cite{regev2009lattices}, the computational hardness of which is the central pillar of post-quantum cryptography. We prove the following theorems. 

\begin{theorem}
    $\GI \in \BQPOI$.
\end{theorem}

\begin{theorem}
    There exists a positive absolute constant $c \in \bbR_{> 0}$ such that
    $\GCVP_{ c \cdot n\sqrt{n} } \in \BQPOI$.
\end{theorem}

We find the precise computational power of $\BQPOI$ as an interesting open question. While computable OI shows algorithmic results seemingly out of the reach of standard quantum computation, it is not obviously overpowered, in the sense that we did not find a trivial way to solve computational problems that cause major implications to the theory of computational complexity, e.g., problems complete for the classes $\NP$ or $\SZK$.

\vspace{1.5mm}
\paragraph{Scientific Contribution.}
As a high-level summary, this work shows a new connection between information processing models in quantum gravity, and some of the fundamental computational problems in computer science. More specifically, we make the following contributions to computational complexity and quantum information theory (both are explained in length in Section \ref{section:overview}).
\begin{itemize}
    \item 
    \textbf{In computational complexity,} we show a novel parameterization technique for classical circuits $C : \{ 0, 1 \}^{k'} \rightarrow \{ 0, 1 \}^{k}$, allowing to view invertibility of circuits as a spectrum rather than a binary predicate. We define a new natural and relaxed variant of the known Statistical Difference \cite{sahai2003complete} problem ($\SD$, Definition \ref{definition:statistical_difference_problem}), called the Sequentially-Invertible Statistical Difference problem ($\SISD$, Definition \ref{definition:sisd_problem}). Previous work \cite{goldreich1991proofs}, \cite{goldreich1998limits} shows how to reduce both $\GI$ and $\GCVP$ (respectively) to $\SD$, and as a technical contribution we show new reductions to the relaxed problem $\SISD$.
    These results are in classical computational complexity theory, independent of any non-classical physics or computation. 

    \item 
    \textbf{In quantum information theory,} we define a new computational model and complexity measure, capturing a simple analogue of quantum computation with unitary order superposition. We observe a connection between the model and the problem of Statistical Difference $\SD$. As a further technical step, we show a new algorithmic technique in the computational model, which in turn implies that for some parameters, the $\SISD$ problem can be solved by a quantum computer with unitary order superposition (in formal terms, that for some parameters, $\SISD \in \BQPOI$).
\end{itemize}
Under both disciplines, several open questions stem from our results.

\vspace{1.5mm}
The remaining of the paper is as follows. In Section \ref{section:overview} we give an overview of the techniques we develop in this work, their connection to new computational problems and also the problems $\GI$ and $\GCVP$. The Preliminaries are given in Section \ref{section:preliminaries}. In Section \ref{section:new_notions_and_definitions} we provide all of the new definitions from this work, including the definition of the OI oracle, the complexity class $\BQPOI$ and the computational problem $\SISD$. Section \ref{section:sisd_algorithm} contains our main quantum algorithm, and in Sections \ref{section:gi_reduction}, \ref{section:gcvp_reduction} we provide new reductions for $\GI$ and $\GCVP$, which, together with the results in Section \ref{section:sisd_algorithm}, show the containment of the problems in $\BQPOI$.

\section{Technical Overview} \label{section:overview}
In this section we give an overview of the main technical ideas in this work. In Section \ref{subsection:overview_computable_OI_and_state_generation} we form an understanding of computational abilities that computable OI has and a standard quantum computer does not seem to posses. We also connect these understandings to the Statistical Difference problem. In Section \ref{subsection:overview_SISD_problem} we show a new technique to parameterize invertibility of classical circuits and define a new computational problem, called the Sequentially Invertible Statistical Difference Problem $\SISD$. In Section \ref{subsection:overview_how_to_solve_sisd} we show a polynomial-time quantum algorithm using an OI oracle, that solves $\SISD$. In Section \ref{subsection:overview_reductions} we show classical polynomial-time reductions from the Graph Isomorphism Problem $\GI$ and the $O(n\sqrt{n})$-Gap Closest Vector Problem $\GCVP_{ O\left( n\sqrt{n} \right) }$ to the $\SISD$ problem.

\subsection{Computable Order Interference and Quantum State Generation} \label{subsection:overview_computable_OI_and_state_generation}
Before we are trying to use computable OI to solve new problems, it is important to intuitively understand what it enables that a standard quantum computer does not.
In the introduction we saw that a standard quantum computer successfully executes an OI transformation with probability $\left( \frac{ \norm{\OI\left( \{ U_{i} \}_{i \in [m]}, \ket{\psi}\right)} }{ m! } \right)^{2}$, that is, in the order of the scaled norm of the OI state. In the setting of computable OI (modeled by an OI oracle $\OOI$), this probability can be amplified as described in \ref{equation:introduction_ICS_success_probability}, depending on the amount of phase alignment (i.e., how much the non-amplifiable, first component in \ref{equation:introduction_ICS_success_probability}, is close to $1$). One can verify that to make the success probability \ref{equation:introduction_ICS_success_probability} at least a constant using computable OI, the complexity of amplification (captured by the parameter $\lambda$) has to be at least the inverse of the scaled norm $\Theta \left(\frac{ m! }{ \norm{\OI\left( \{ U_{i} \}_{i \in [m]}, \ket{\psi}\right)} }\right)$, depends on the wanted constant\footnote{More generally, assuming perfect phase alignment, to get a success probability of $1 - \frac{1}{k}$, executing the OI oracle with $\lambda = k \cdot \frac{ m! }{ \norm{\OI\left( \{ U_{i} \}_{i \in [m]}, \ket{\psi}\right)} }$ suffices.}. However, if we are already allowed to execute in the complexity scales of $\frac{ m! }{ \norm{\OI\left( \{ U_{i} \}_{i \in [m]}, \ket{\psi}\right)} }$, a standard quantum computer can successfully generate the OI state with high probability: Just by repeatedly trying to execute the OI transformation, where as we know, each attempt succeeds with the bounded probability $\left( \frac{ \norm{\OI\left( \{ U_{i} \}_{i \in [m]}, \ket{\psi}\right)} }{ m! } \right)^{2}$. When an OI transformation fails in either of the cases (computable OI or standard quantum computing) the state is considered as destroyed. So, in both cases, the computational complexity of successfully executing a single OI transformation is $\approx \frac{ m! }{ \norm{\OI\left( \{ U_{i} \}_{i \in [m]}, \ket{\psi}\right)} }$.

What is the complexity of executing $2$ consecutive OI transformations? For two different sets of unitaries $\{ U_{0, i} \}_{i \in [m_{0}]}$, $\{ U_{1, i} \}_{i \in [m_{1}]}$. Assuming both transformations are successful, this will create a state of different structure than what can be created by a single OI transformation. The gap in complexity between standard quantum computation and computable OI, becomes more apparent when we think of running $\ell$ consecutive OI transformations, for a general $\ell$ (assuming perfect phase alignment in both cases). Computable OI lets us maintain a probability close to $1$ in that case, using only complexity $\approx \ell^3 \cdot \frac{ m! }{ \norm{\OI\left( \{ U_{i} \}_{i \in [m]}, \ket{\psi}\right)} }$ over all $\ell$ attempted transformations, while a standard quantum computer needs all $\ell$ tries to succeed consecutively, which generally has a complexity increasing exponentially with $\ell$. One deduction is thus, that the power of computable OI is by executing OI transformations \emph{sequentially}. Assume we are starting from the state $\ket{0^{n}}$, and we sequentially execute $\ell$ OI transformations, for $\ell$ different sets of unitaries $\{ U_{1, i} \}_{i \in [m_{1}]}$, $\cdots, \{ U_{\ell, i} \}_{i \in [m_{\ell}]}$. The resulting quantum state is
\begin{equation} \label{equation:sequential_oi_transformation_state}
    \sum_{
    \sigma_{1} \in S_{m_{1}},
    \cdots,
    \sigma_{\ell} \in S_{m_{\ell}}
    }
    \left(
    \prod_{ i_{\ell} \in [m_{\ell}] }
    U_{ \ell, \sigma_{\ell}^{-1}\left( i_{\ell} \right) }
    \right)
    \cdots 
    \left(
    \prod_{ i_{1} \in [m_{1}] }
    U_{ 1, \sigma_{1}^{-1}\left( i_{1} \right) }
    \right)
    \ket{\psi}
    \enspace ,
\end{equation}
and has a unique structure, corresponding to a uniform superposition over all partial execution orders, where the partial order rule is that for $i, j \in [\ell]$, $i < j$, the execution of unitaries in the $i$-th set $\{ U_{i, k} \}_{k \in [m_{i}]}$ always comes before the execution of unitaries in the $j$-th set $\{ U_{j, k} \}_{k \in [m_{j}]}$.

\paragraph{The statistical difference problem and quantum state generation.}
When we think of the ability to generate new quantum states, there comes to mind one of the most fundamental hard problems in computer science -- the Statistical Difference ($\SD$) problem \cite{sahai2003complete}. As an intuitive description, the input to the $\SD$ problem is a pair of classically samplable distributions $D_{0}$, $D_{1}$, and the job of a solving algorithm is to distinguish whether the total variation distance between the distributions is small or large. Formally, for functions $a, b : \Nat \rightarrow [0, 1]$ such that $\forall n \in \Nat: a(n) \leq b(n)$, an input to the $\SD_{a(n), b(n)}$ problem is a pair $\left( C_0, C_1 \right)$ of classical circuits with the same output size, that is, $C_{b} : \{ 0, 1 \}^{ k_{b} } \rightarrow \{ 0, 1 \}^{ k }$ for some $k_{0}, k_{1}, k \in \Nat$. For each of the two circuits, we consider the output distribution $D_{b}$ which is generated by executing $C_{b}$ on a uniformly random input $z \in \{ 0, 1 \}^{k_{b}}$. The problem comes with the promise that it is either the case that the total variation distance between $D_{0}$, $D_{1}$, is at most $a(n)$, in which case we denote $\left( C_{0}, C_{1} \right) \in \YES$, or the distance is more than $b(n)$, in which case we denote $\left( C_{0}, C_{1} \right) \in \NO$. Generally speaking, the larger the gap $b(n) - a(n)$, the easier the problem. In this work, our focus is on the family of problems $\SD_{a(n), b(n)}$, for functions $a(n), b(n)$ such that there exists a polynomial $\poly(n)$ with $b(n)^{2} - 2a(n) + a(n)^{2} \geq \frac{1}{\poly(n)}$. From hereon we drop the subscript $a(n)$, $b(n)$ from $\SD_{a(n), b(n)}$ and write $\SD_{\poly}$ to implicitly refer to the family of problems $\SD_{a(n), b(n)}$ for such functions $a(n)$, $b(n)$. The $\SD_{\poly}$ problem is a notoriously hard problem in computer science (also for quantum computers), and it is at least as hard as some of the widely studied computational problems in cryptography and in combinatorics, for example: Quadratic Residuosity, Lattice Isomorphism, Discrete Logarithm, Decisional Diffie-Hellman, Graph Isomorphism, the Gap Closest Vector Problem and the Gap Shortest Vector Problem.

The thing that connects $\SD_{\poly}$ and the ability to generate non-trivial quantum states is a quantum procedure called the Swap Test \cite{barenco1997stabilization, buhrman2001quantum}. It is observed in \cite{aharonov2003adiabatic} that the $\SD_{\poly}$ problem could be efficiently solved if we had the ability to efficiently generate the output distribution states of classical circuits. If given input $\left( C_0, C_1 \right)$ we could efficiently generate the normalized $k$-qubit quantum states,
$$
\ket{ C_{0}\left( R \right) } :=
\sum_{ x \in \{ 0, 1 \}^{ k_{0} } } \ket{C_{0}(x)} \enspace ,
$$
$$
\ket{ C_{1}\left( R \right) } :=
\sum_{ x \in \{ 0, 1 \}^{ k_{1} } } \ket{C_{1}(x)} \enspace ,
$$
then we can execute a Swap Test on the pair of states $\ket{ C_{0}\left( R \right) }$, $\ket{ C_{1}\left( R \right) }$. A Swap Test approximates (with some inverse polynomial error) the inner product between the states as vectors, which in turn relates to the total variation distance between the output distributions $D_{0}$, $D_{1}$ (by using the relations between fidelity and total variation distance). Finally, as long as the gap $b(n)^{2} - 2a(n) + a(n)^{2}$ is at least inverse polynomial, the error in the approximation of the Swap Test can be overcome and we can decide whether the pair of circuits is in $\YES$ or $\NO$. So, the $\SD_{\poly}$ problem can be reduced to the quantum state generation of $\ket{ C_{0}\left( R \right) }$, $\ket{ C_{1}\left( R \right) }$. With accordance to the hardness of the $\SD_{\poly}$ problem, there is no known way for a quantum computer to generate in polynomial time the output distribution state for a general classical circuit $C$.

\subsection{The Sequentially Invertible Statistical Difference Problem} \label{subsection:overview_SISD_problem}
It is a standard fact in quantum computing that for a classical circuit $C : \{ 0, 1 \}^{k'} \rightarrow \{ 0, 1 \}^{k}$ there is an $\left( k' + k \right)$-qubit unitary $U_{x, C}$, of roughly the same computational complexity as $C$, that maps: $\forall x \in \{ 0, 1 \}^{k'}, y \in \{ 0, 1 \}^{k} : U_{x, C} \cdot \ket{x, y} = \ket{x, y \oplus C(x)}$. It is in particular easy to quantumly generate the $\left( k' + k \right)$-qubit state
$$
\ket{ R, C\left( R \right) } :=
\frac{1}{\sqrt{2^{ k' }}} \sum_{ x \in \{ 0, 1 \}^{k'} } \ket{x, C(x)} \enspace .
$$
The above state is not a valid solution to the quantum state generation problem (and cannot be used to solve $\SD$ with Swap Test, as described in Section \ref{subsection:overview_computable_OI_and_state_generation}). The state we want is the output distribution state $\ket{ C\left( R \right) }$, where only the output should appear in the superposition, excluding the input. Another known fact in quantum computing is that if a circuit $C : \{ 0, 1 \}^{k} \rightarrow \{ 0, 1 \}^{k}$ is both efficiently computable \emph{and} efficiently invertible (i.e., there is a classically efficiently computable $C^{-1} : \{ 0, 1 \}^{k} \rightarrow \{ 0, 1 \}^{k}$), then there is an efficient unitary $U_{C}$ on $k$ (rather than $k' + k$) qubits, mapping $\forall x \in \{ 0, 1 \}^{k} : U_{C} \cdot \ket{x} = \ket{C(x)}$ (possibly using additional ancillary qubits). In particular, for efficiently computable and invertible circuits, the quantum state generation is quantumly efficiently solvable. However, for all "interesting" cases of the $\SD_{\poly}$ problem (in particular, all instances captured by important cryptographic hard problems), the circuits $C_{0}$, $C_{1}$ are conjectured to not be efficiently invertible, and are sometimes known to be not invertible at all (i.e., they are not injective). In first glance, it isn't clear how computable OI can help in generating the output distribution states for non-invertible circuits, in a way that a standard quantum computer cannot.

We get a dichotomous situation, where if a circuit $C : \{ 0, 1 \}^{k'} \rightarrow \{ 0, 1 \}^{k}$ is invertible then we can generate the output distribution state $\ket{ C\left( R \right) }$ using a standard quantum computer, and if it is not invertible, we don't know a way to generate the state, even with computable OI. Could we look at invertibility in a more high-resolution way? Could we think about invertible classical computation as a numerical parameter, rather than binary?

\paragraph{Sequentially invertible distributions.}
Our next idea is to parameterize invertibility. Our parameterization technique does not rely on any non-classical computation and may be of independent interest. Note that sampler circuits are equivalent to randomized circuits, which take the form $C : \{ 0, 1 \}^{k} \times \{ 0, 1 \}^{k'} \rightarrow \{ 0, 1 \}^{k}$, where the randomness $z \in \{ 0, 1 \}^{k'}$ is only considered as auxiliary input, and together with it there's a main input $x \in \{ 0, 1 \}^{k}$. We can thus consider sampler and randomized circuits interchangeably.

Our definition of parameterized invertibility of a circuit is as follows. Let $C : \{ 0, 1 \}^{k'} \rightarrow \{ 0, 1 \}^{k}$ a circuit. For $r \in \{ 0, 1, \cdots, k' \}$, we say that $C$ is $r$-sequentially invertible if there is an $r$-invertible circuit sequence that samples from the output distribution of $C$\footnote{The full definition of $\left( r, t, \ell, \varepsilon \right)$-sequential invertibility, which takes into consideration more parameters and not only $r$, is given in Section \ref{subsection:sequential_invertibility_definitions}. For the needs for this overview, we give the minimal definitions.}. An $r$-invertible circuit sequence on $k$ bits is a sequence of pairs of circuits $\left( C_{i, \rightarrow}, C_{i, \leftarrow} \right)_{i \in [\ell]}$ such that for every $i \in [\ell]$,
    \begin{itemize}
        \item 
        For every $Y \in \{ \rightarrow, \leftarrow \}$, the circuit $C_{i, Y}$ uses $r_{i} \leq r$ random bits and maps from $k$ to $k$ bits, that is,
        $$
        C_{i, Y} : \{ 0, 1 \}^{ k } \times \{ 0, 1 \}^{ r_{i} } \rightarrow \{ 0, 1 \}^{ k }
        \enspace .
        $$
    
        \item 
        The circuit directions are inverses of each other, per hard-wired randomness, that is:
        $$
        \forall z \in \{ 0, 1 \}^{ r_{i} }, x \in \{ 0, 1 \}^{ k } :
        x = C_{i, \leftarrow}\left( C_{i, \rightarrow}\left( x; z \right); z \right) \enspace .
        $$
        Note that this also implies that for each hard-wired randomness, the circuits act as permutations on the set $\{ 0, 1 \}^{k}$.
    \end{itemize}
We say that the sequence samples from the output distribution of $C$ if we have the following equivalence of distributions,
$$
    \Bigl\{
    C(z) \; | \; z \gets \{ 0, 1 \}^{k'}
    \Bigr\}
    \equiv 
    $$
    $$
    \biggl\{ C_{\ell, \rightarrow }\left( \cdots C_{2, \rightarrow }\left( C_{1, \rightarrow }\left( 0^{ k } ; z_1 \right); z_2 \right) \cdots ; z_{\ell} \right)
    \; \Big| \;
    z_{1} \gets \{ 0, 1 \}^{ r_{1} }, \cdots, z_{\ell} \gets \{ 0, 1 \}^{ r_{\ell} } \biggr\}
    \enspace .
$$
The intuition is that the smaller the parameter $r$ is, "the more the circuit is invertible". We give a few examples. For any invertible circuit $C : \{ 0, 1 \}^{k} \rightarrow \{ 0, 1 \}^{k}$, $C$ is $0$-sequentially invertible (for $\ell = 1$, that is, by a single pair $\left( C_{1, \rightarrow}, C_{1, \leftarrow} \right)$). For any (arbitrary) function $C : \{ 0, 1 \}^{k'} \rightarrow \{ 0, 1 \}^{k}$, $C$ is $k'$-sequentially invertible (again for $\ell = 1$). The function $C_{\oplus} : \{ 0, 1 \}^{2\cdot k} \rightarrow \{ 0, 1 \}^{k}$, that for two strings $a, b \in \{ 0, 1 \}^{k}$ outputs their bit-wise logical XOR: $C_{\oplus}(a, b) := a \oplus b$, is $1$-sequentially invertible (by a sequence of $\ell = k$ pairs), and $C_{\land} : \{ 0, 1 \}^{2\cdot k} \rightarrow \{ 0, 1 \}^{k}$, that for two strings $a, b \in \{ 0, 1 \}^{k}$ outputs their bit-wise logical AND: $C_{\land}(a, b) := a \land b$, is $2$-sequentially invertible (again, by a sequence of $k$ pairs). We leave it as an instructive exercise for the reader to verify the above claims.

\paragraph{Definition of the $\SISD$ problem.}
We combine the ideas of sequential invertibility with the $\SD$ problem, and define the $\left( a(n), b(n), r(n) \right)$-Sequentially Invertible Statistical Difference Problem $\left( \SISD_{a(n), b(n), r(n)} \right)$ for functions $a(n), b(n) : \Nat \rightarrow [0, 1]$, $r(n) : \Nat \rightarrow \Nat$, such that for every $n \in \Nat$, $a(n) \leq b(n)$ and $r(n) \in \{ 0, 1, \cdots , n \}$. An input is a pair of sequences,
$$
\left( \left( C^{0}_{i, \rightarrow}, C^{0}_{i, \leftarrow} \right)_{i \in [\ell]},
\left( C^{1}_{i, \rightarrow}, C^{1}_{i, \leftarrow} \right)_{i \in [\ell]} \right) \enspace ,
$$
where for each $b \in \{ 0, 1 \}$, $\left( C^{b}_{i, \rightarrow}, C^{b}_{i, \leftarrow} \right)_{i \in [\ell]}$ is an $r(n)$-invertible circuit sequence on $k$ bits, where $n$ is the input size,
$$
n :=
\bigg| \left(
\left( C^{0}_{i, \rightarrow}, C^{0}_{i, \leftarrow} \right)_{i \in [\ell]},
\left( C^{1}_{i, \rightarrow}, C^{1}_{i, \leftarrow} \right)_{i \in [\ell]}
\right) \bigg|
\enspace .
$$

Similarly to the case of the regular Statistical Difference problem $\SD_{a(n), b(n)}$, we have the following promise on the gap of the total variation distance between the circuits. Formally, for $b \in \{ 0, 1 \}$ denote
$$
    D_{b}
    :=
    \biggl\{
    C^{b}_{\ell, \rightarrow}\left(
    \cdots
    C^{b}_{2, \rightarrow}\left(
    C^{b}_{1, \rightarrow}\left( 0^{k(n)} ; z_1 \right); z_2 \right)
    \cdots
    ; z_{\ell} \right)
    \biggr\}
    \enspace ,
$$
where the randomness sequence is sampled in both cases uniformly $z_{1} \gets \{ 0, 1 \}^{r^{b}_{1}(n)}, \cdots, z_{\ell} \gets \{ 0, 1 \}^{r^{b}_{\ell}(n)}$ for the suitable $b \in \{ 0, 1 \}$. We denote the total variation distance between the distributions $D_0$, $D_1$ with $d$, and it is promised that either $d \leq a(n)$, in which case we say that the input belongs to $\YES$, or we have $d > b(n)$, in which case we determine that the input belongs to $\NO$. The goal of the problem is to decide whether the input belongs to $\YES$ or to $\NO$. As our previous convention for the $\SD$ problem, we also drop the subscript $a(n), b(n)$ to implicitly be any pair of functions with difference $b(n)^{2} - 2a(n) + a(n)^{2}$ lower-bounded inverse polynomially, and write $\SISD_{\poly, \, r(n)}$.

\subsection{How to Solve $\SISD_{ \poly, \, O\left( \log(n) \right) }$ with Computable Order Interference} \label{subsection:overview_how_to_solve_sisd}
The $\SISD_{ \poly, \, r(n) }$ problem is a private case of the $\SD_{ \poly }$ problem (for any randomness $r(n) \in [0, n]$), so a statement that holds for general circuits in $\SD_{ \poly }$ also holds for $\SISD_{\poly, \, r(n)}$. In order to solve $\SISD_{ \poly, \, O( \log(n) ) }$ we will show how to efficiently solve the quantum state generation problem for sequentially invertible circuits with randomness $O\left( \log(n) \right)$. More precisely, given an input $r(n)$-invertible circuit sequence $C = \left( C_{i, \rightarrow}, C_{i, \leftarrow} \right)_{i \in [\ell]}$, such that $r(n) = O(\log(n))$, we show how to generate the output distribution state
\begin{equation}
\ket{ C(R) }
:=
\sum_{
\substack{
z_{1} \in \{ 0, 1 \}^{ r_{1} }, \\ \vdots \\ z_{\ell} \in \{ 0, 1 \}^{ r_{\ell} }
}
}
\ket{
    C_{\ell, \rightarrow}\left(
    \cdots
    C_{2, \rightarrow}\left(
    C_{1, \rightarrow}\left( 0^{k} ; z_1 \right);
    z_2 \right)
    \cdots ;
    z_{\ell} \right)
}
\enspace ,
\end{equation}
by a quantum computer with access to a computable order interference oracle $\OOI$, in polynomial time.

We can think of the superposition in the state $\ket{ C(R) }$ as a tree -- in level $i \in \{ 0, 1, \cdots, \ell \}$ we have a superposition over the randomness of all $i$ circuits $C_{1, \rightarrow}, \cdots, C_{i, \rightarrow}$, and the state $\ket{ C(R) }$ is simply the full tree of depth $\ell$. We can thus attempt to start from level $0$, i.e., the state $\ket{0^{k}}$, and progress one level at a time.
Assume we have a superposition over level $i \in \{ 0, 1, \cdots, \ell - 1 \}$, that is, we have the state,
\begin{equation}
\ket{ C_{1, \cdots, i}(R_{1, \cdots, i}) }
:=
\sum_{
\substack{
z_{1} \in \{ 0, 1 \}^{ r_{1} }, \\ \vdots \\ z_{i} \in \{ 0, 1 \}^{ r_{i} }
}
}
\ket{
    C_{i, \rightarrow}\left(
    \cdots
    C_{1, \rightarrow}\left( 0^{k} ; z_1 \right)
    \cdots ;
    z_{i} \right)
}
\enspace ,
\end{equation}
and we want to obtain the state of the next level of the distribution (with high probability),
\begin{equation}
\ket{ C_{1, \cdots, i + 1}(R_{1, \cdots, i + 1}) }
:=
\sum_{
\substack{
z_{1} \in \{ 0, 1 \}^{ r_{1} }, \\ \vdots \\ z_{i + 1} \in \{ 0, 1 \}^{ r_{i + 1} }
}
}
\ket{
    C_{i + 1, \rightarrow}\left(
    \cdots
    C_{1, \rightarrow}\left( 0^{k} ; z_1 \right)
    \cdots ;
    z_{i + 1} \right)
}
\enspace .
\end{equation}
For each randomness $z_{i + 1} \in \{ 0, 1 \}^{r_{i + 1}}$, the circuit $C_{i + 1, \rightarrow}\left( 
\cdot; z_{i + 1} \right)$ is efficiently computable and invertible, which, as we know, means we have the efficient $k$-qubit unitary $U_{i + 1, z_{i + 1}}$, mapping for every $x \in \{ 0, 1 \}^{k}$, $U_{i + 1, z_{i + 1}}\cdot \ket{x} = \ket{C_{i + 1, \rightarrow}(x ; z_{i + 1})}$. Given this structure, observe that the transformation we want to apply to the state $\ket{ C_{1, \cdots, i}(R_{1, \cdots, i}) }$ is,
$$
\sum_{ z_{i + 1} \in \{ 0, 1 \}^{r_{i + 1}} } U_{i + 1, z_{i + 1}} \enspace ,
$$
which will exactly turn the state into $\ket{ C_{1, \cdots, i + 1}(R_{1, \cdots, i + 1}) }$. An additional way to look at this is that we want to apply a superposition over unitary \emph{choices}, rather than orders, over the set $\{ U_{j} \}_{j \in [m]}$, where we changed notation by $m = 2^{\left(r_{i + 1}\right)}$, $\forall j \in [m] : U_{j} := U_{i + 1, j}$. The amount of randomness $r_{i + 1} \leq r(n)$ is logarithmic in $n$, thus the number of unitaries in the set $\{ U_{ j } \}_{ j \in [m] }$ is polynomial in $n$ (and all of these unitaries are of polynomial time complexity).

Our problem is reduced to a clean setting at this point. Given input a set of $m$ unitaries $\{ U_{i} \}_{i \in [m]}$ and a state $\ket{\psi}$, we want to execute an interference of choices rather than orders, that is, to obtain the normalization of
$$
\sum_{i \in [m]} U_{i} \cdot \ket{\psi} 
\enspace .
$$
For the solution to suffice, we need the success probability of choice interference to behave as the success probability of order interference, in both components of the probability: (1) the scaled norm, and (2) the phase alignment. The scaled norm should be the state's norm divided by the number of choices $m$, and the phases should align between choices, rather than all orders. Formally, we want to construct an efficient quantum algorithm that has access to the order interference oracle, and can execute an interference of choices with success probability of:
\begin{equation} \label{equation:overview_CI_success probability}
    \left(
    \frac
    { \sum_{ x \in \{ 0, 1 \}^{n} } | \sum_{ i \in [m] } \alpha_{x, i} | }
    { \sum_{ x \in \{ 0, 1 \}^{n} }  \sum_{ i \in [m] } | \alpha_{x, i} | }
    \right)
    \frac
    { \frac{ \norm{ \sum_{i \in [m]} U_{i} \cdot \ket{\psi} } }{ m } }
    { \frac{ \norm{ \sum_{i \in [m]} U_{i} \cdot \ket{\psi} } }{ m } + \frac{1}{\lambda} } 
    \enspace ,
\end{equation}
where for $i \in [m]$ we write,
$$
U_{i} \cdot \ket{\psi} := \sum_{ x \in \{ 0, 1 \}^{n} } \alpha_{x, i} \cdot \ket{x} \enspace ,
$$
Lastly, for the specific type of unitaries that we use for the algorithm for state generation (i.e., unitaries that implement the action of invertible classical circuits), all amplitudes are real and non-negative in each step, and thus in particular phases perfectly align. One can observe that in order to finish our algorithm for state generation (which in turn will finish the proof for $\SISD_{\poly, O\left( \log(n) \right)} \in \BQP^{\OI}$), it remains to only explain how to efficiently execute the above transformation with the target success probability.

\paragraph{Order interference plus small quantum memory implies choice interference.}
The observation at the center of our solution is that we can make unitary paths converge, while paying with limiting the variety of the superposition. For example, by taking all unitaries to be identical $U_{1} = \cdots = U_{m}$ the norm of the OI state is $m!$ and the phases of all unitary evolution orders align, which makes the success probability $1$. However, this is equivalent to simply executing $U_1^{m}$, which can be done by a standard quantum computer. Our solution to this is as follows.

Let $n \in \Nat$ the number of qubits in the state $\ket{\psi}$ (and the number of qubits that each of the unitaries $\{ U_{i} \}_{i \in [m]}$ act on), and consider adding an ancilla of $\ceil{\log{m}}$ qubits, so the state we are operating on is $\ket{\psi} \ket{0^{\ceil{\log{m}}}}$. For each unitary $U_{i}$ we make a slight change, and execute the modified unitary $\Tilde{U}_{i}$ which consists of two steps. (1) Execute the circuit $U^{c}_{i}$, which executes $U_{i}$ on the left $n$ qubits, conditioned on the state of the $\ceil{\log{m}}$ right qubits being $\ket{ 0^{\ceil{\log{m}}} }$. (2) Execute the unitary $U_{+1}$ on the right $\ceil{ \log{m} }$ qubits, which treats the qubits as a binary representation of a number in $\{ 0, 1, \cdots, m - 1, \cdots, 2^{ \ceil{\log{m}} } - 1 \}$, and increments by $1$ modulo $2^{ \ceil{\log{m}} }$. For the Unitary $U_{+1}$ we use the following fact, which we will also use later:
\begin{fact} \label{fact:adding_unitarily_a_mod_N}
    For any pair of numbers $a, N \in \Nat$ such that $a \leq N$, the function $f_{a, N} : \{ 0, 1 \}^{\lceil \log(N) \rceil} \rightarrow \{ 0, 1 \}^{\lceil \log(N) \rceil}$ that adds to the input $a$ modulo $N$ is efficiently computable and invertible. It follows there is a unitary circuit $U_{a, N}$ on $\lceil \log(N) \rceil$ qubits, executing in time $\poly\left( \log\left( a \cdot N \right) \right)$ that maps
    $$
    \forall b\in \{ 0, 1, \cdots, N - 2, N - 1 \} :
    U_{a, N} \cdot \ket{b} = \ket{b + a \text{ (mod) }N}
    \enspace .
    $$
\end{fact}
That is, if we know the numbers $a, N$ not as part of the computation but in advance, then it is possible to efficiently add $a$ modulo $N$, in superposition. The solution to executing choice interference is to execute a call to the OI oracle with input $\left( \{ \Tilde{U}_{i} \}_{i \in [m]} , \; \ket{\psi} \ket{0^{\ceil{\log{m}}}} \right)$, and then tracing out the rightmost $\ceil{\log{m}}$ qubits. 

Let us analyze the result of the above suggestion. We consider what happens to the modified state $\ket{\psi} \ket{0^{\ceil{\log{m}}}}$, in each of the unitary execution orders $\sigma \in S_{m}$, for the modified unitaries $\Tilde{U}_{i}$. Two things can be verified by the reader.
\begin{itemize}
    \item
    For $i \in [m]$ and any permutation $\sigma \in S_{m}$ such that $\sigma^{-1}\left( i \right) = 1$, the state in the left $n$ qubits at the end of the computation is $U_{i} \cdot \ket{\psi}$. In simple terms, the only unitary that effectively executes on the side of the input state $\ket{\psi}$ is the first one in the ordering $\sigma$.

    \item 
    When we look at the state of the right $\ceil{ \log(m) }$-qubit counter register at the end of the computation, for any execution order $\sigma \in S_{m}$, the state is always $m \mod 2^{ \ceil{ \log(m) } }$. This in particular means that the rightmost $\ceil{\log{m}}$ qubits are disentangled from the rest of the left $n$ qubits, and can be traced out without damaging coherence of the rest of the state.
\end{itemize}
Together, this also means that for every $i \in [m]$,
$$
\sum_{ \sigma \in S_{m} : \sigma^{-1}(1) = i }
\left( \prod_{j \in [m]} U_{\sigma^{-1}\left( j \right)} \right) \cdot \ket{\psi}
=
\left( m - 1 \right)!
\cdot
\left( U_{i} \cdot \ket{\psi} \right) \otimes \ket{ m \text{ mod } 2^{ \ceil{ \log(m) } } } \enspace .
$$
Finally, since all $m!$ permutations in $S_{m}$ can be partitioned into $m$ sets of equal size $\left( m - 1 \right)!$, such that for $i \in [m]$, set $i$ is all permutations $\sigma \in S_{m}$ such that $\sigma^{-1}(1) = i$, by calculation, we get exactly the desired success probability in \ref{equation:overview_CI_success probability}.

\subsection{Reducing $\GI$ and $\GCVP$ to $ \SISD_{\poly, O(\log(n))}$} \label{subsection:overview_reductions}
In the last part of the Technical Overview we will show that two problems that are believed to be hard for quantum computers, are solvable by a quantum computer with computable order interference, by reducing to the $\SISD_{\poly, O(\log(n))}$ problem. In the first part of this section we show the reduction from $\GI$ which is more immediate, and in the second part we show the more involved reduction from $\GCVP_{ O\left( n\sqrt{n} \right) }$. To this end, we will first review the existing reductions from each of the problems to the more general $\SD_{\poly}$ problem (which we don't know how to solve), and then move to our new reductions, to the easier problem of $\SISD_{\poly, O(\log(n))}$.

\paragraph{The Graph Isomorphism Problem and Statistical Difference.}
An input to the $\GI$ problem is a pair of (simple, undirected) graphs $\left( G_0, G_1 \right)$, and the problem is to decide whether the pair of graphs is isomorphic or not. More formally, for $n \in \Nat$ we define a graph as $G = \left( [n], E \right)$, where $E \subseteq [n] \times [n]$ is the set of edges. Two graphs are isomorphic iff there exists an isomorphism between them, that is, a bijection (or, permutation) $f$ on $[n]$ such that for every pair $u, v \in [n]$: $\{ u, v \} \in E_{0} \iff \{ f(u), f(v) \} \in E_{1}$, where for $b \in \{ 0, 1 \}$, $E_{b}$ is the set of edges in the graph $G_{b}$.

As stated earlier in Section \ref{subsection:overview_computable_OI_and_state_generation}, there is an existing reduction from $\GI$ to the standard statistical difference problem $\SD_{\poly}$. More precisely, \cite{goldreich1991proofs} shows a statistical zero-knowledge protocol for the $\GI$ problem, out of which there can be derived the reduction $\GI \leq_{p} \SD_{a, b}$, for $a = 0$, $b = 1$. The reduction to $\SD_{0, 1}$ means that for an input pair $\left( G_0, G_1 \right)$ to the $\GI$ problem, if the graphs are isomorphic then the output pair of circuits $\left( C_0, C_1 \right)$ are such that the statistical distance (another term for total variation distance) between their output distributions is $0$, and if the graphs are not isomorphic, the distance is $1$, which means that the supports of the output distributions are disjoint.

The reduction is as follows. Given the input graphs $\left( G_0, G_1 \right)$, the first circuit $C_0$ is defined such that for input randomness $r$ of polynomial length, it samples a uniformly random permutation $\sigma \in S_n$ on the elements $[n]$, applies $\sigma$ to the graph $G_0$, and the output of $C_0$ is the $\sigma$-permuted graph $\sigma\left( G_0 \right)$\footnote{Recall that permuting $G$ means applying $\sigma$ to the vertices of $G$, but keeping the edges connecting the same vertex numbers.}. The circuit $C_1$ is identical, only that it applies $\sigma$ to the second graph, $G_1$. To see the correctness of the reduction, recall that for two isomorphic graphs, their sets of all isomorphic graphs are equal, thus in that case, the circuits $C_0, C_1$ sample from the exact same distribution (uniform distribution over their isomorphic graphs!) and the statistical distance is $0$. In the other case where the graphs are not isomorphic, there does not exist a permutation $\sigma$ that will map one graph to the other (because such permutation will act as an isomorphism between the graphs), and in particular, the sets of permutations of the two graphs will never intersect, which will result in a statistical distance of $1$.

\paragraph{Our reduction $\GI \leq_{p} \SISD_{ \left( 0, 1 , \log(n) \right) }$ .}
It remains to give a more fine-grained reduction for $\GI$. The goal of our reduction is to stick to the previous reduction of \cite{goldreich1991proofs}, in the sense that the circuit $C_b$ will still aim to output a random permutation of the graph $G_{b}$ (for $b \in \{ 0, 1 \}$), but additionally, the sampling will be done in a $\log(n)$-sequentially invertible manner.

Our reduction follows by two arguments. The first step is that constant permutation mappings are efficiently computable and invertible. Formally, assume that we use adjacency matrix representation of a graph, so $n^{2}$ bits are sufficient to represent any $n$-vertex graph. Then, for a fixed permutation $\sigma \in S_{n}$, the circuit $C_{\sigma} : \{ 0, 1 \}^{ n^{2} } \rightarrow \{ 0, 1 \}^{ n^{2} }$ that permutes a given input graph $G$ with the permutation $\sigma$, is both efficiently computable and efficiently invertible. The second step is that a random permutation over $[n]$, which has $n!$ options, can be sampled in many small steps, each having a small (that is, polynomially-bounded) number of options, and crucially, each step is invertible. Specifically, the Fisher-Yates shuffle \cite{fisher1953statistical}, shows how to sample a random permutation in $n - 1$ steps, each step $i = n, n - 1, \cdots, 3, 2$ having $i$ options, and importantly, each step is a fixed permutation in of itself, which means the circuit that implements it is efficiently computable and invertible. Formally, for a decreasing $i = n, n - 1, \cdots, 3, 2$, pick a random number $j \in [i]$, and swap between $i$ and $j$.

It remains to compile our observations into a $\log(n)$-invertible circuit sequence. Consider $n - 1$ circuit pairs $\left( C_{i, \rightarrow}, C_{i, \leftarrow} \right)_{ i \in \{ n, n - 1, \cdots, 3, 2 \} }$, such that $C_{i, \rightarrow} : \{ 0, 1 \}^{ n^{2} } \times \{ 0, 1 \}^{ \lceil \log(n) \rceil } \rightarrow \{ 0, 1 \}^{ n^{2} }$ maps between $n$-vertex graphs and uses the bits of randomness to pick a random element in $j \in [i]$, and then applies a swap between vertices $i$ and $j$. Per fixed randomness $j \in [i]$, the circuit $C_{i, \rightarrow}\left( \cdot, j \right)$ performs a swap between vertices $i, j$, which one can observe to be a fixed permutation. By generating a circuit sequence $\left( C^{b}_{i, \rightarrow}, C^{b}_{i, \leftarrow} \right)_{i \in \{ n, n - 1, \cdots, 3, 2 \}}$ to permute the graph $G_b$ (for $b \in \{ 0, 1 \}$), where both sequences are generated as above according to the Fisher-Yates algorithm, the rest of the correctness of the reduction then follows similar lines to that of \cite{goldreich1991proofs}. The full proof is given in Section \ref{section:gi_reduction}.

\paragraph{The $g(n)$-Gap Closest Vector Problem.}
For a gap function $g : \Nat \rightarrow \bbR_{\geq 1}$, an input to the $\GCVP_{g(n)}$ problem is a triplet $\left( \Basis, \tVector, d \right)$, such that
\begin{itemize}
    \item
    $\Basis = \{ \Basis_{1}, \Basis_{2}, \cdots, \Basis_{n} \}$ is a basis for $\bbR^{n}$ with integer vectors, i.e., $\Basis \subseteq \bbZ^{n}$.

    \item 
    $\tVector \in \bbZ^{n}$ is an arbitrary integer vector.

    \item 
    $d \in \Nat$ is a natural number.
\end{itemize}
We define $\Lattice_{\Basis}$, the lattice given by the basis $\Basis$, as the set of all integer linear combinations of the vectors in $\Basis$, that is:
$$
\Lattice_{\Basis}
:=
\Bigl\{
\aVector \in \bbZ^{n}
\; \Big| \;
\exists \sVector \in \bbZ^{n} : \aVector = \Basis \cdot \sVector
\Bigr\}
\enspace .
$$
The input comes with a promise: It is either the case that the vector $\tVector$ is $d$-close to the lattice i.e., $\dist\left( \Lattice_{\Basis}, \tVector \right) \leq d$, in which case we denote $\left( \Basis, \tVector, d \right) \in \YES$, or the vector is farther than $d \cdot g(n)$ from the lattice i.e., $\dist\left( \Lattice_{\Basis}, \tVector \right) > d \cdot g(n)$, in which case we denote $\left( \Basis, \tVector, d \right) \in \NO$. In simple words, a solving algorithm for $\GCVP_{g(n)}$ gets the promise that the vector $\tVector$ is either close or far away from the lattice $\Lattice_{\Basis}$, and needs to decide which is the case. The bigger the gap $g(n)$ is, the easier the problem becomes\footnote{In particular, it is known that for an exponential $g(n) \approx 2^{n}$ the $\GCVP_{2^{n}}$ problem is solvable in classical polynomial time \cite{ajtai2001sieve}.}. What we will show is that for $g(n) = O\left( n\sqrt{n} \right)$ there is a reduction $\GCVP_{ g(n) } \leq_{p} \SISD_{\poly, O(\log(n))}$. The inspiration and starting point of our reduction is the reduction $\GCVP_{ \sqrt{ \frac{n}{O\left( \log\left( n \right) \right)} } } \leq_{p} \SD_{\poly}$ by Goldreich and Goldwasser \cite{goldreich1998limits}, which we recall next.

\paragraph{The Goldreich-Goldwasser reduction $\GCVP_{ \sqrt{ \frac{n}{O\left( \log\left( n \right) \right)} } } \leq_{p} \SD_{\left( 1 - \frac{1}{n^{c}}, \; 1 \right)}$ .}
The \cite{goldreich1998limits} reduction shows how, for a gap $\sqrt{ \frac{n}{O\left( \log\left( n \right) \right)} }$, we can reduce the Gap Closest Vector Problem to the Statistical Difference Problem $\SD$. Given a $\GCVP_{g(n)}$ instance $\left( \Basis, \tVector, d \right)$ (for $g(n) = \sqrt{ \frac{n}{O\left( \log\left( n \right) \right)} }$), the reduction executes in classical polynomial time and outputs $\left( C_{0}, C_{1} \right)$, such that in case the $\GCVP$ input is in $\YES$ then the statistical distance between the output distributions of the circuits is bounded by $1 - \frac{1}{n^{c}}$ (for some constant $c \in \Nat$), and in case the $\GCVP$ input is in $\NO$ then the statistical distance between the output distributions of the circuits is the maximal $1$. Given the input, the \cite{goldreich1998limits} reduction first defines an upper bound $M := M_{\left( \Basis, \tVector \right)} := 2^n \cdot \max_{\vVector \in \{ \bVector_{1}, \cdots, \bVector_{n}, \tVector \} } \norm{\vVector}$. The upper bound is used to treat an inherently infinite lattice in a finite way, by considering only $M$-bounded coordinates. The definition of the circuits $C_{0}$, $C_{1}$ is as follows.
The circuit $C_{0}$ first samples a uniformly random $M$-bounded coordinates vector $\sVector \in \bbZ_{M}^{n}$ and transforms the coordinates to a lattice vector $\vVector := \Basis \cdot \sVector$. The final output of $C_{0}$ is a uniformly random point $\pVector$ inside the $n$-dimensional ball of radius $\frac{d \cdot g(n)}{2}$ around $\vVector$. The circuit $C_{1}$ is the same, with one change: At the end, we add $\tVector$ to $\pVector$. Note that if the distance of the vector $\tVector$ from the lattice is $d$, then it can be written as $\tVector = \aVector + \eVector$, such that $\aVector \in \Lattice_{\Basis}$ and $\norm{ \eVector } = d$, such that $\eVector$ is the shortest vector which satisfies this equality.

The rationale behind the reduction is that when having a uniform distribution over a lattice, then adding a lattice vector to the output keeps the distribution identical. This logic stays consistent for a distribution that samples balls around the lattice points. So, when looking at the change between $C_0$, $C_1$, the lattice component $\aVector$ in the vector $\tVector$ gets "swallowed", and the only vector that skews the distribution is $\eVector$. One can imagine how adding $\eVector$ pushes all $n$-balls around the lattice points in its direction and length, and this is exactly the move from the output distribution of $C_0$ to the output distribution of $C_1$.
If $\left( \Basis, \tVector, d \right) \in \YES$ then the skewing vector $\eVector$ has short length $\leq d$ compared to the radius $\frac{d \cdot g(n)}{2}$ of the balls, which makes the balls have significant intersection, which in turn implies an upper bounded statistical distance between the output distributions of the two circuits (specifically, it can be proven that the statistical distance is $\leq 1 - \frac{1}{n^{c}}$ for some constant $c \in \Nat$) and thus $\left( C_{0}, C_{1} \right) \in \YES$. If $\left( \Basis, \tVector, d \right) \in \NO$ and $\eVector$ has length $> d \cdot g(n)$, one can observe how it pushes the $n$-balls (which have radius $\frac{d \cdot g(n)}{2}$, which is in turn half of the distance to the vector $\eVector$) to be completely disjoint from the original ones, which implies a statistical distance of $1$ and thus $\left( C_{0}, C_{1} \right) \in \NO$. 

\paragraph{Our reduction $\GCVP_{ O\left( n\sqrt{n} \right) } \leq_{p} \SISD_{\left( \frac{1}{4} + O\left( n^{ -\frac{1}{4} } \right) , \; \frac{3}{4} - e^{ -\Omega\left( n \right) }, \; 1 \right)}$ .}
The \cite{goldreich1998limits} reduction maps to circuits $\left( C_{0}, C_{1} \right)$ that sample a uniform point $\pVector$ inside the radius-$\frac{d \cdot g(n)}{2}$, $n$-dimensional ball, around a uniform lattice vector $\vVector$. While sampling a random lattice vector can be done in a sequentially-invertible manner, sampling a uniform point inside the $n$-ball by a sequentially invertible algorithm seems like a fundamental problem. More specifically (but still in a nutshell), the reason that sampling from the $n$-ball seems hard, is that known algorithmic techniques to get a uniformly random sample from the $n$-ball rely on \emph{normalization}. That is, we sample from some large set $S$, vectors that are not necessarily inside the wanted radius of the $n$-ball, and then normalize. The problem arises because normalization is a function of the entire vector, and it is inherently a non-invertible function (for example, all vectors of the same direction but different lengths, will map to the same normalized vector).

In our reduction we will still aim to sample from a distribution $P$ around uniform lattice vectors, but not from the uniform $n$-ball. We note that our distribution $P$ should have two main characteristics.
\begin{itemize}
    \item
    $n$-dimensional symmetry: The distribution $P$ should not have a preferred direction, out of the $n$ directions in $\bbZ^{n}$. This requirement is to deal with arbitrary locations of the vector $\tVector$ with respect to the lattice $\Lattice_{\Basis}$ (that is, if the $P$ we pick prefers some direction $i \in [n]$, then the vector $\tVector$ can be chosen as a function of that direction, tricking the reduction).
    
    \item
    Density: We imagine the support of $P$ as an $n$-dimensional shape, by looking at the elements that are farthest away from the origin and thinking on them as a boundary for the shape. The distribution $P$ should be as dense as possible in its shape, which effectively means having a uniform distribution within the shape's boundaries. The density requirement emerges from the case where $\left( \Basis, \tVector, d \right) \in \YES$ -- in that case we get a uniform sample from the shape around the lattice vectors, and the more the shape is sparse, the less our ability to give a good upper bound on the statistical distance between the output distributions of the circuits $C_{0}$, $C_{1}$.
\end{itemize}
It turns out that the distribution which maximizes the above two requirements is indeed the uniform distribution inside the $n$-ball. The challenge in our reduction is to find a proper distribution $P$ that satisfies the above, and also has a $O\left( \log(n) \right)$-sequentially-invertible sampling algorithm. In fact, we will show a distribution $P$ with a $1$-sequentially-invertible sampling algorithm.

\paragraph{A sequentially-invertible distribution $P$.}
The distribution $P$ we choose is the $n$-dimensional discrete truncated Gaussian distribution $D^{n}_{B}$ (as per Definition \ref{definition:discrete_truncated_gaussian}, $B \in \Nat$ is the truncation parameter, assume it to be $\approx \frac{d \cdot g(n)}{2 \cdot \sqrt{n}}$). It is a known fact that $D^{n}_{B}$ is invariant under rotations in $n$-space, which exactly means it does not prefer any direction of the $n$ directions in $\bbZ^{n}$. In terms of density, the Gaussian distribution is obviously less dense in its shape than the uniform $n$-ball, which is (part of) the reason that we need a bigger gap $g(n) = O\left( n\sqrt{n} \right)$ for which we know how to reduce $\GCVP_{g(n)}$ to $\SISD$, compared to the better gap $\sqrt{ \frac{n}{O\left( \log\left( n \right) \right)} }$ from the previous reduction. As part of the full proof of the reduction (given in Section \ref{section:gcvp_reduction}) we analyze the Gaussian distribution's sparsity in this manner, and the loss of hardness (which translates to the larger $g(n) \approx n\sqrt{n} >> \sqrt{ \frac{n}{O\left( \log\left( n \right) \right)} }$) that the sparsity incurs.

It remains to show a sequentially-invertible sampling algorithm for $P := D^{n}_{B}$. One of the basic properties of the Gaussian distribution says that sampling from $D^{n}_{B}$ can be done by sampling $n$ i.i.d. samples from $D_{B}$ (one sample for each of the coordinates). All we need to find is a sequentially-invertible algorithm to sample from the discrete truncated Gaussian distribution $D_{B}$.

Sampling from $D_{B}$ has a plethora of existing techniques and algorithms (e.g., \cite{gentry2008trapdoors, peikert2010efficient, micciancio2012trapdoors, rossi2019simple}), but it is unclear how to show that these can be executed by a sequentially invertible algorithm. In fact, we did not find any sequentially invertible algorithm that samples from a distribution with negligible statistical distance to $D_{B}$. Once we relax our needs to sample from the discrete Gaussian distribution only approximately (that is, with an inverse-polynomial statistical distance), we find a solution. We give the intuitive steps behind our solution in the following.
\begin{itemize}
    \item 
    We recall one of the basic theorems of probability theory - The Central Limit Theorem (CLT). The CLT implies that for a distribution $X$ with expectation $0$ and standard deviation $\sigma \in \bbR_{\geq 0}$, when we take the sum of $\kappa$ i.i.d. samples from $X$, that is, $S := X_{1} + \cdots + X_{\kappa}$, then the normalized sum $\frac{S}{\sqrt{\kappa}}$ approaches the Gaussian distribution (with standard deviation $\sigma$) as $\kappa$ approaches infinity. This helps us because now, instead of sampling from the Gaussian distribution (which we do not know how to do with a sequentially invertible algorithm), we can repeatedly sample from any distribution $X$, which we aim to be sequentially invertible. There are still two inherent barriers for the above approach:
    \begin{enumerate}
        \item
        While the CLT asserts closeness to the Gaussian distribution for \emph{some} number $\kappa$ of summands, it does not guarantee a specific rate of convergence, as the rate depends on the chosen distribution $X$. The first barrier implies the efficiency of the reduction.

        \item 
        The second barrier is more delicate: The CLT uses normalization (just like the vector normalization used for sampling from the $n$-ball). Since we are dealing with discrete distributions, normalizing means executing integer (that is, rounded) division by $\sqrt{\kappa}$, which is again a non-invertible operation. Without the normalization in the CLT, it is known that it does not always work.
    \end{enumerate}

    To overcome both barriers, we need to choose a distribution $X$ with two properties: (1) It can be sampled by a sequentially invertible algorithm, and (2) The rate of convergence of the sum $S$ (and \emph{not only} the normalized sum $\frac{S}{\sqrt{\kappa}}$) to the Gaussian distribution, is polynomially fast.
    
    \item
    We observe that the uniform distribution has two properties that work in our favour, for the distribution $X$. It can be shown that if $X$ has a certain discrete continuity property (which the discrete uniform distribution does have), then the un-normalized sum $S$ gives a good convergence rate to the Gaussian distribution. Formally, in Section \ref{section:gcvp_reduction} we show that summing $\kappa$ i.i.d. samples of the uniform distribution over $\{ -2^{\beta}, \cdots, -1, 0, 1, \cdots, 2^{\beta} \}$ for $2^{\beta} \approx \frac{B}{\sqrt{\kappa}}$ has statistical distance bounded by $O\left( \frac{1}{\sqrt{\kappa}} \right)$ to the discrete Gaussian distribution $D_{B}$.

    \item 
    A second property of the uniform distribution is that it can be sampled by a sequentially-invertible algorithm, in fact, by a $1$-invertible circuit sequence. To see this, observe that to sample uniformly at random from $\{ 0, 1, \cdots, 2^{\beta + 1} - 1 \}$, all we need to do is to break the sampled number into its summation as powers of $2$, and in the $i$-th circuit for $i \in [\beta]$, we add (or not) the number $2^{i}$. Summing random powers of $2$ up to power $\beta$ exactly gives a uniformly random sample in $\{ 0, 1, \cdots, 2^{\beta + 1} - 1 \}$. At the end we can deterministically move the distribution to center around zero, by subtracting $2^{\beta}$.

    \item 
    To see why the circuit sequence is invertible, observe that for every $i \in [\beta]$, for every possible randomness for the circuit $C_{i}$ (which is just one bit $b_{i} \in \{ 0, 1 \}$), the fixed-randomness circuit $C_{i}\left( \cdot ; b_{i} \right)$ always adds modulo $2^{\beta + 1}$ the same number $b_{i} \cdot 2^{i}$. By going back to Fact \ref{fact:adding_unitarily_a_mod_N}, it follows that the fixed-randomness $C_{i}\left( \cdot ; b_{i} \right)$ is efficiently computable and invertible, as needed.
\end{itemize}
To conclude, the circuit sequence $C^{0}$ samples a uniformly random lattice vector $\vVector$, and then adds an approximate discrete Gaussian vector $\eVector$, where the approximation is done by summing i.i.d. samples from the uniform distribution. The circuit sequence $C^{1}$ is the same, only that we add $\tVector$ in the end. Both circuit sequences are shown to be $1$-sequentially invertible. We refer the reader to Section \ref{section:gcvp_reduction} for the full details of the proof, including the analysis of the statistical distance between the output distributions of the circuits.

\section{Preliminaries} \label{section:preliminaries}
This section contains previously existing notions and definitions from the literature that are relevant to this work, mostly from quantum computation and computational complexity theory.

\begin{itemize}
    \item
    For $n \in \Nat$, define $[n] := \{ 1, 2, 3, \cdots, n \}$.

    \item 
    When we use the notation $\log(\cdot)$ we implicitly refer to the base-$2$ logarithm function $\log_{2}(\cdot)$.

    \item 
    The only non-standard notation will be for iterated matrix multiplication. Specifically, our default convention for iterated matrix multiplication, using the notation $\prod$, will be by enumerating from right to left (unlike the usual, left to right). Specifically, for a natural number $m \in \Nat$ and square complex matrices $\{ M_{i} \}_{i \in [m]}$ ($\forall i \in [m] : M_{i} \in \bbC^{n \times n}$ for some $n \in \Nat$), we define:
    $$
    \prod_{i \in [m]} M_{i}
    :=
    M_{m}\cdot M_{m - 1} \cdots M_{2} \cdot M_{1}
    \enspace . 
    $$

    \item 
    We follow the standard Dirac notations for quantum states in quantum information processing.
    For a variable $x$ (i.e., notation) that denotes a classical binary string, i.e., $x \in \{ 0, 1 \}^{*}$, the notation $\ket{x}$ refers to the pure quantum state of $|x|$ qubits, which equals the $x$-th standard basis element of the space $\bbC^{2^{|x|}}$. That is,
    $$
    \forall n \in \Nat : \forall x \in \{ 0, 1 \}^{n} : 
    \ket{x}
    :=
    e_{x}
    :=
    \left(
    0_{1}, 0_{2}, \cdots, 0_{x - 1}, 1_{x}, 0_{x + 1}, \cdots ,0_{2^{n} - 1}, 0_{2^{n}}
    \right)^{T}
    \enspace .
    $$
    In case $x \in \bbZ$ is a natural number, the notation $\ket{x}$ first translates $x$ to its $\ceil{ \log_{2}\left( x \right) }$-bit binary representation, and then when $x$ is now a binary string, applies the Dirac operator as above.

    \item 
    For a natural number $N \in \Nat$, the $N$-th complex root of unity is defined to be $\omega_{N} := e^{\frac{2 i \pi}{N}}$.

    \item 
    When we use the notation $\norm{ \cdot }$ for norm of a vector, we implicitly refer to the Euclidean, or $\ell_{2}$ norm $\norm{ \cdot }_{2}$, unless explicitly noted otherwise.

    \item 
    For two classical random variables (i.e., distributions) $D_{0}$, $D_{1}$ over binary strings of finite length $\leq k$, the total variation distance (sometimes referred to as statistical distance) between $D_{0}$ and $D_{1}$ is defined as follows:
    $$
    \norm{D_{0} - D_{1}}_{TV}
    :=
    \frac{1}{2}\norm{D_{0} - D_{1}}_{1}
    :=
    \frac{1}{2}
    \sum_{i \in [k], x \in \{ 0, 1 \}^{i}} |D_{0}\left( x \right) - D_{1}\left( x \right)|
    \enspace ,
    $$
    where for $b \in \{ 0, 1 \}$, $i \in [k]$ and $x \in \{ 0, 1 \}^{i}$, the notation $D_{b}\left( x \right)$ is the probability to sample $x$ in the distribution $D_{b}$.

    \item 
    For two classical random variables (i.e., distributions) $D_{0}$, $D_{1}$ over binary strings of finite length $\leq k$, the fidelity between $D_{0}$ and $D_{1}$ is defined as follows:
    $$
    F\left( D_{0}, D_{1} \right)
    :=
    \sum_{i \in [k], x \in \{ 0, 1 \}^{i}} \sqrt{ D_{0}\left( x \right)\cdot D_{1}\left( x \right) }
    \enspace ,
    $$
    where for $b \in \{ 0, 1 \}$, $i \in [k]$ and $x \in \{ 0, 1 \}^{i}$, the notation $D_{b}\left( x \right)$ is the probability to sample $x$ in the distribution $D_{b}$.

    \item 
    For every pair of classical random variables $D_{0}$, $D_{1}$ over binary strings of finite length $\leq k$, the following relation between the fidelity and total variation distance between distributions is known:
    $$
    1 - F\left( D_{0}, D_{1} \right) \leq
    \norm{D_{0} - D_{1}}_{TV}
    \leq
    \sqrt{ 1 - F\left( D_{0}, D_{1} \right)^{2} }
    \enspace .
    $$
    
    \item 
    Unless specifically noted otherwise, for any variable $u$ (not necessarily denoting a binary string), the notation $\ket{u}$ denotes a unit vector in complex vector space, of some finite dimension. That is, the notation usually refers to a quantum state and is thus a normalized vector, unless explicitly noted otherwise. Also, for any $n \in \Nat$, an $n$-qubit state is always a unit vector in $\bbC^{2^{n}}$.

    \item
    A \DPT algorithm is a classical deterministic polynomial-time Turing machine. It is a known fact in computational complexity theory that a \DPT is equivalent to a uniform family of polynomial-sized classical deterministic circuits. We will use these two notions interchangeably throughout this work as our model for efficient classical computation.

    \item
    A \PPT algorithm is a classical probabilistic polynomial-time Turing machine. It is a known fact in computational complexity theory that a \PPT is equivalent to a uniform family of polynomial-sized classical probabilistic circuits. We will use these two notions interchangeably throughout this work as our model for efficient probabilistic classical computation.

    \item
    A \QPT algorithm is a quantum polynomial-time Turing machine. It is a known fact in computational complexity theory that a \QPT is equivalent to a uniform family of polynomial-sized quantum circuits. We will use these two notions interchangeably throughout this work as our model for efficient quantum computation.

    \item 
    Unitary transformations are the transformations which quantum mechanics allows, when excluding measurements (and adding ancilla qubits). Throughout this work we will consider unitaries as matrices, transformations or circuits, interchangeably, depending on the context, using the following conventions.
    
    \item 
    For $n \in \Nat$, an $n$-qubit unitary is a unitary matrix $U \in \bbC^{2^{n} \times 2^{n}}$.

    \item 
    For $n \in \Nat$, an $n$-qubit general circuit (or $n$-qubit quantum algorithm) is a quantum circuit $C$ that can perform the constant-sized quantum gates $\{ X, Z, P, H, T, CNOT, SWAP, Tof \}$, execute measurement gates and add ancillary qubits registers containing $\ket{0^{k}}$.

    \item 
    For $n \in \Nat$, an $n$-qubit unitary circuit is a quantum circuit on $n$ qubits that only uses quantum gates $\{ X, Z, P, H, T, CNOT, SWAP, Tof \}$, adds no ancilla and uses no measurement gates.

    \item 
    For $n, k\in \Nat$, an $n$-qubit, $k$-ancilla unitary circuit $C$ is a quantum circuit on $n$ qubits that first concatenates a $k$-qubit ancilla $\ket{0^{k}}$ to the input, then executes only quantum gates $\{ X, Z, P, H, T, CNOT, SWAP, Tof \}$. $C$ satisfies that the ancillary qubits start and (crucially) return to $0^{k}$, that is, there exists an $n$-qubit unitary $U_{C}$ such that,
    $$
    \text{For every $n$-qubit state $\ket{\psi}$ } :
    C\left( \ket{\psi}\ket{0^{k}} \right) = \left( U\ket{\psi} \right) \ket{0^{k}} \enspace .
    $$
    In that case, we say that the quantum circuit $C$ implements the unitary $U_{C}$.

    \item 
    For $n \in \Nat$ and an $n$-qubit unitary transformation $U$, a unitary circuit for $U$, denoted $C_{U}$ (or denoted $U$ by overloading notation, when it is clear from the context whether we are referring to the matrix or circuit version of $U$) is an $n$-qubit, $k$-ancilla (for some $k \in \left( \Nat \cup \{ 0 \} \right)$) unitary quantum circuit that implements the unitary $U$.

    \item 
    Representing permutations: For $m \in \Nat$, our way for presenting a permutation $\sigma \in S_{m}$ in the form of a binary string in this paper, is by using $m \cdot \ceil{\log(m)}$ bits, interpreted as a list of $m$ numbers in $[m]$ -- an ordering of the set $[m]$ according to the permutation $\sigma$. Specifically, the number in the $i$-th packet of $\ceil{\log_{2}(m)}$ bits, is $\sigma^{-1}(i)$.

    \item 
    In this work we assume that quantum systems have finite dimensions. Formally, for every quantum register $R$ of $n \in \Nat$ qubits we assume there exists a quantum register $R'$ of finite size containing $n' \in \Nat$ qubits such that the quantum state of the joint system $\left( R, R' \right)$ is in a pure state, that is, the joint system is not entangled with any other outside quantum system.

    \item 
    For natural numbers $N_{0}, N_{1} \in \Nat$, the set $\bbZ_{\left( -N_{0}, N_{1} \right)}$ is defined as the set of integers between (and including) $-N_{0}$ and $N_{1}$, that is,
    $$
    \bbZ_{\left( -N_{0}, N_{1} \right)}
    :=
    \{ -N_{0}, \cdots, -1, 0, 1, \cdots, N_{1} \}
    \enspace .
    $$
\end{itemize}

\subsection{Standard Notions from Computational Complexity Theory}
In this work we study generalization of decision problems, which are called promise problems. Promise problems are defined as follows.

\begin{definition} [Promise Problem]
    Let $\{ 0, 1 \}^{*} := \bigcup_{ i \in \left( \Nat \cup \{ 0 \} \right) } \{ 0, 1 \}^{i}$ the set of all binary strings, of all lengths. A promise problem $\prod := \left( \YES , \NO \right)$ is given by two (possibly infinite) sets $\YES \subseteq \{ 0, 1 \}^{*}$, $\NO \subseteq \{ 0, 1 \}^{*}$ such that $\YES \cap \NO = \emptyset$.
\end{definition}

A decision problem, which is a private case of promise problems, is defined as follows. It basically means a promise problem where the input can be any string, and is not promised to be in any particular set.

\begin{definition} [Decision Problem]
    Let $\prod := \left( \YES , \NO \right)$ a promise problem. We say that $\prod$ is a decision problem if $\YES \cup \NO = \{ 0, 1 \}^{*}$.
\end{definition}

We next define languages, which are derived from decision problems.

\begin{definition} [Language]
    Let $\prod := \left( \YES , \NO \right)$ a decision problem. We say that $L$ is the language induced by the decision problem $\prod$ if $\YES = L$. 
\end{definition}

The complexity class $\BQP$ represents all computational problems that have a feasible solution on a quantum computer, and more broadly (and informally), $\BQP$ represents humanly-feasible quantum computation. The writing $\BQP$ stands for Bounded-Error Quantum Polynomial-Time, and defined as follows (which is one out of many equivalent definitions of this class).

\begin{definition} [The Complexity Class $\BQP$]
    The complexity class $\BQP$ is the set of promise problems $\prod := \left( \YES , \NO \right)$ that are solvable by a quantum polynomial-time algorithm. Formally, such that there exists a \QPT $M$ such that for every $x \in \YES$, $M(x) = 1$ with probability at least $p(x)$, and if $x \in \NO$, $M(x) = 0$ with probability at least $p(x)$, such that
    $$
    \forall x \in \left( \YES \cup \NO \right) : p(x) \geq 1 - 2^{-\poly\left(|x|\right)} \enspace ,
    $$
    for some polynomial $\poly(\cdot)$.
\end{definition}

We next discuss computational problems (that is, promise problems) on arbitrary mathematical objects. As in the standard conventions in computational complexity theory, arbitrary mathematical objects can be represented by binary strings, and thus promise problems capture any problem.

\paragraph{The Statistical Difference Problem.}
We next define the Statistical Difference ($\SD$) problem \cite{sahai2003complete}, a well-known promise problem in computer science.

\begin{definition} [The Statistical Difference Problem] \label{definition:statistical_difference_problem}
    Let $a, b : \Nat \rightarrow [0, 1]$ functions such that $\forall n \in \Nat : a(n) \leq b(n)$. The $\left( a(n), b(n) \right)$-gap Statistical Difference problem, denoted $\SD_{a(n), b(n)}$ is a promise problem, where the input is a pair of classical circuits $\left( C_{0}, C_{1} \right)$ with the same output size, $C_{0} : \{ 0, 1 \}^{k_{0}} \rightarrow \{ 0, 1 \}^{k}$, $C_{1} : \{ 0, 1 \}^{k_{1}} \rightarrow \{ 0, 1 \}^{k}$. For a pair of circuits we define their output distributions:
    $$
    \forall b \in \{ 0, 1 \} :
    D_{b} := 
    \{ C_{b}\left( x \right) | x \gets \{ 0, 1 \}^{k_{b}} \}
    \enspace .
    $$
    The problem is defined as follows.
    \begin{itemize}
        \item
        $\YES$ is the set of pairs of circuits $\left( C_{0}, C_{1} \right)$ such that $\norm{D_{0} - D_{1}}_{TV} \leq a(n)$,

        \item
        $\NO$ is the set of pairs of circuits $\left( C_{0}, C_{1} \right)$ such that $\norm{D_{0} - D_{1}}_{TV} > b(n)$,
    \end{itemize}
    such that $n$ is the description size of the pair of circuits,
    $$
    n :=
    |\left( C_{0}, C_{1} \right)|
    \enspace .
    $$
\end{definition}

In this work we mainly care about the statistical difference problem where the gap is not too small, and formally, lower-bounded inverse-polynomially. To this end we define the \emph{family} of problems $\SD_{\poly}$ as follows.

\begin{definition} [The Polynomial-Gap Statistical Difference Problem]
    The polynomially-bounded statistical difference problem, denoted $\SD_{\poly}$ is in fact a set of promise problems (and not a promise problem by itself). It is defined as follows.
    $$
    \SD_{\poly}
    :=
    \bigcup_{
    \substack{
    a(n), b(n) : \Nat \rightarrow [0,1], \\
    \exists \text{ a polynomial } \poly : \Nat \rightarrow \Nat \text{ such that } \forall n \in \Nat : b(n)^{2} - 2a(n) + a(n)^{2} \geq \frac{1}{\poly(n)}
    }
    }
    \{ \SD_{a(n), b(n)} \}
    \enspace .
    $$
\end{definition}

\paragraph{The Graph Isomorphism Problem.}
As part of this work we show how a strengthening of the model of quantum computation solves the graph isomorphism problem. To this end, we define (simple and undirected) graphs and then the graph isomorphism problem. First, a (simple and undirected) graph with $n$ vertices is a set of $n$ points, implicitly numbered $[n]$, where some pairs of points have edges between them. The edges are directionless lines, or connections, between a pair of points $i, j \in [n]$, $i \neq j$. In this work when we refer to graphs, they are implicitly simple and undirected.

\begin{definition} [Simple and Undirected Graph]
    For $n \in \Nat$, an $n$-vertex (simple and undirected) graph $G$ is given by the pair $\left( [n], E \right)$, such that $E$ contains size-$2$ subsets of $[n]$. 
\end{definition}

Observe that the maximal amount of edges, that is, the maximal size of $E$ for an $n$-vertex graph $G$, is ${n \choose 2} := \frac{n\left( n - 1 \right)}{2}$. This means that a graph can be represented by a binary string of length $O(n^2)$. Next, we define isomorphic graphs. The intuition behind graph isomorphism is that isomorphic graphs $G_{0}$, $G_{1}$ are essentially the exact same graph (they contain the same information on connections), and visually, by moving the vertices of (any embedding of) one of the graphs $G_{b}$, we get a picture of an embedding of the other $G_{\lnot b}$.

\begin{definition} [Isomorphic Graphs]
    Let $n \in \Nat$ and let $G_0 = \left( [n], E_{0} \right)$, $G_1 = \left( [n], E_{1} \right)$ a pair of $n$-vertex graphs. We say that $G_{0}$ and $G_{1}$ are isomorphic and denote $G_{0} \simeq G_{1}$ if there exists a bijection (or equivalently, permutation) $f$ on the set $[n]$ such that,
    $$
    \forall i, j \in [n] :
    \{ i, j \} \in E_{0} \iff \{ f(i), f(j) \} \in E_{1}
    \enspace .
    $$
\end{definition}

The Graph Isomorphism ($\GI$) problem is the computational problem of deciding whether a pair of input graphs are isomorphic or not. The $\GI$ problem has a long history of attempted solutions using polynomial-time algorithms, both classically and quantumly.

\begin{definition} [The Graph Isomorphism Problem] \label{definition:graph_isomorphism_problem}
    The Graph Isomorphism problem, denoted $\GI$ is a decision problem, where the input is a pair of graphs $\left( G_{0}, G_{1} \right)$ with the same number of vertices $n$. The problem is defined as follows.
    \begin{itemize}
        \item
        $\YES$ is the set of pairs of graphs $\left( G_{0}, G_{1} \right)$ such that $G_{0} \simeq G_{1}$.

        \item
        $\NO$ is $\{ 0, 1 \}^{*} \setminus \YES$.
    \end{itemize}
\end{definition}

As a final comment, the $\GI$ is a decision problem (and not only the more general, promise problem) since if it had a promise, this promise would turn out to be computationally trivial. That is, the set $\NO$ is the sets of strings that either (1) just don't represent a pair of graphs, or (2) represent a pair of graphs $\left( G_{0}, G_{1} \right)$, but this pair is not isomorphic. Checking whether an input string has the correct format and indeed represents a pair of graphs can be done in classical polynomial-time, and thus the graph isomorphism problem is usually thought of as a decision problem and not just a promise problem.

We use the following known facts about the distributions of random permutations of isomorphic graphs, and also on known algorithmic techniques for sampling a uniformly random permutation.

\begin{fact} [Permutations of isomorphic and non-isomorphic graphs] \label{fact:random_permutations_of_isomorphic_graphs}
    Let $\left( G_{0}, G_{1} \right)$ a pair of graphs on $n \in \Nat$ vertices each.
    
    In case the pair is isomorphic, then for a permutation $\sigma \in S_{n}$ on the set $[n]$, denote by $\sigma\left( G \right)$ the graph generated by rearranging the vertices of $G$ with accordance to the permutation $\sigma$. Then, the following distributions are identical:
    $$
    \{ \sigma\left( G_{0} \right) | \sigma \gets S_{n} \}
    \equiv
    \{ \sigma\left( G_{1} \right) | \sigma \gets S_{n} \}
    \enspace .
    $$

    In case the pair of graphs is not isomorphic, then the above distributions have disjoint supports.
\end{fact}

\begin{fact} [The Fisher-Yates permutation-sampling algorithm \cite{fisher1953statistical}] \label{fact:fisher_yates_algorithm}
    Let $n \in \Nat$, the following process produces a uniformly random permutation $\sigma \in S_{n}$.
    Start with any permutation of $\left( 1, 2, \cdots, n - 1, n \right)$, and for a decreasing $i = n, n - 1, \cdots, 3, 2$, for every $i$, make a swap between the element currently in slot $i$ and the element currently in slot $r_{i}$, where $r_{i}$ is a uniformly random element in $[i]$.
\end{fact}

\paragraph{The Gap Closest Vector Problem.}
We turn to the definition of the gap closest vector problem ($\GCVP$). We define lattices and then the $\GCVP$ problem. For a basis $\Basis$ of $\bbR^{n}$, a lattice with basis $\Basis$ is the infinite set of all \emph{integer} linear combinations of the basis vectors in $\Basis$. In this work, in order to not deal with real-valued vectors and finite precision, we consider only lattice bases with integer coordinates, i.e., bases of $\bbR^{n}$ from $\bbZ^{n}$. The computational hardness of these problems are polynomially equivalent.

\begin{definition} [An Integer Lattice]
    Let $n \in \Nat$ and let $\Basis := \{ \bVector_{1}, \bVector_{2}, \cdots, \bVector_{n} \}$ a basis for $\bbR^{n}$ with integer vectors, that is, $\Basis \subseteq \bbZ^{n}$. The lattice of $\Basis$, denoted $\Lattice_{\Basis}$, is defined as follows.
    $$
    \Lattice_{\Basis}
    :=
    \Bigl\{
    \aVector \in \bbZ^{n}
    \; \Big| \;
    \exists \sVector \in \bbZ^{n} : \aVector = \Basis \cdot \sVector
    \Bigr\}
    \enspace .
    $$
\end{definition}

For a lattice $\Lattice$ and a vector $\tVector \in \bbZ^{n}$, the distance $\Delta\left( \Lattice, \tVector \right)$ between $\Lattice$ and $\tVector$ is the minimum (Euclidean, $\ell_{2}$) distance between $\vVector$ and $\tVector$, where the minimum is taken over $\vVector \in \Lattice$. Next, we define the Gap Closest Vector Problem.

\begin{definition} [The Gap Closest Vector Problem] \label{definition:gap_closest_vector_problem}
    Let $g : \Nat \rightarrow \bbR_{\geq 1}$ a function. The $g(n)$-Gap Closest Vector Problem, denoted $\GCVP_{g(n)}$ is a promise problem, where the input is a triplet $\left( \Basis, \tVector, d \right)$ such that $\Basis \subseteq \bbZ^{n}$ is a basis for $\bbR^{n}$, $\tVector \in \bbZ^{n}$ is an integer vector and $d \in \Nat$. For an input triplet, the promise problem is defined as follows.
    \begin{itemize}
        \item
        $\YES$ is the set of triplets $\left( \Basis, \tVector, d \right)$ such that $\Delta\left( \Lattice_{\Basis}, \tVector \right) \leq d$.

        \item
        $\NO$ is the set of triplets $\left( \Basis, \tVector, d \right)$ such that $\Delta\left( \Lattice_{\Basis}, \tVector \right) > d\cdot g(n)$.
    \end{itemize}
\end{definition}

\subsection{Statistical and Algorithmic Tools}
We use numerous algorithmic and statistical results in this work.

\paragraph{Quantum Swap Test.} We use the Swap Test \cite{barenco1997stabilization, buhrman2001quantum} algorithm, which, among other things, can be used to distinguish between quantum states as a function of their inner product as vectors.

\begin{theorem} [Swap Test Algorithm] \label{theorem:swap_test}
    There exists a quantum algorithm $Q_{ST}\left( \cdot \right)$ that gets as input a $2 n$-qubit register $R$ and outputs a bit $b \in \{ 0, 1 \}$, and executes in complexity $O(n)$. If there exist two $n$-qubit states $\ket{\psi}$, $\ket{\phi}$ such that $R$ is in the separable state $\ket{\psi} \otimes \ket{\phi}$, then the algorithm outputs $1$ with probability $\frac{1}{2} + \frac{|\bra{\psi}\ket{\phi}|^{2}}{2}$, and with the remaining probability outputs $0$.
\end{theorem}

\paragraph{Deviation Bounds.}
We use the following known formulation of Chernoff's bound, which shows an exponentially fast (in the number of samples) decaying probability for the deviation of the average of many binary i.i.d. samples. One of the main usages of the below statement in repetition theorems in computer science, is that the repetition parameter $n \in \Nat$ in the below statement can be chosen as a function of both $\varepsilon \in [0, 1]$ and $p \in [0, 1]$.

\begin{theorem} [Chernoff's Bound for an Average of i.i.d. Samples] \label{theorem:chernoff_bound}
    Let $n \in \Nat$, let $X$ a random variable over $\{ 0, 1 \}$ (that is, $X = 1$ with probability $p \in [0, 1]$ and $X = 0$ with probability $1 - p$) and let $\{ X_{i} \}_{i \in [n]}$ $n$ i.i.d. samples from $X$. Let $A_{X, n} := \frac{ \sum_{i \in [n]} X_{i} }{ n }$ the average of the sum of samples, and let $E\left( A_{X, n} \right)$ the expectation of the average. Then for every $\varepsilon \in (0, 1)$,
    $$
    \Pr_{X_{1}, \cdots, X_{n}}\left[ \big| A_{X, n} - E\left( A_{X, n} \right) \big| \geq \varepsilon \cdot p  \right]
    \leq
    2 \cdot e^{ - n \cdot \frac{p \cdot \varepsilon^{2}}{ 3 } }
    \enspace .
    $$
\end{theorem}

\paragraph{Gaussian distribution variants.}
As part of the reduction $\GCVP \leq_{p} \SISD_{\poly, 1}$, which is presented in Section \ref{section:gcvp_reduction}, we will approximate the discrete truncated Gaussian distribution, which is defined as follows.

\begin{definition} [Discrete truncated Gaussian distribution] \label{definition:discrete_truncated_gaussian}
    For $B \in \Nat$, we define the discrete truncated Gaussian distribution $D_{B}$ as the distribution with sample space $\{ -B, \cdots, -1, 0, 1, \cdots, B \}$ defined as,
    $$
    \forall x \in \{ -B, \cdots, -1, 0, 1, \cdots, B \} :
    $$
    $$
    D_{B}\left( x \right)
    :=
    \frac
    {
    e^{ -\frac{ \pi \cdot |x|^{2} }{ B^{2} } }
    }
    {
    \sum_{ y \in \{ -B, \cdots, -1, 0, 1, \cdots, B \} }
    e^{ -\frac{\pi \cdot |y|^{2} }{ B^{2} } }
    }
    \enspace .
    $$
\end{definition}

As part of the approximation of the discrete Gaussian distribution we will use a closely related variant, which is the \emph{rounded} truncated Gaussian distribution. To this end we first define the continuous Gaussian distribution, then rounded Gaussian distribution, and then the truncated version of the rounded distribution.

\begin{definition} [Continuous Gaussian distribution] \label{definition:continuous_gaussian}
    Let $\mu \in \bbR$ and $\sigma \in \bbR_{\geq 0}$. The continuous Gaussian (or normal) distribution with mean $\mu$ and standard deviation $\sigma$, denoted $\mathcal{N}\left( \mu, \sigma \right)$ is the distribution with sample space $\bbR$ such that,
    $$
    \forall x \in \bbR :
    \mathcal{N}\left( \mu, \sigma \right)\left( x \right)
    :=
    \frac{1}{ \sigma \sqrt{ 2 \pi } } \cdot e^{ -\frac{1}{2}\left( \frac{ x - \mu }{ \sigma } \right)^{2} }
    \enspace .
    $$
\end{definition}

\begin{definition} [Rounded Gaussian distribution] \label{definition:rounded_gaussian}
    Let $\mu \in \bbR$ and $\sigma \in \bbR_{\geq 0}$. The rounded Gaussian distribution with mean $\mu$ and standard deviation $\sigma$, denoted $\mathcal{N}^{\bbZ}\left( \mu, \sigma \right)$ is the distribution with sample space $\bbZ$ such that,
    $$
    \forall x \in \bbZ :
    \mathcal{N}^{\bbZ}\left( \mu, \sigma \right)\left( x \right)
    :=
    \int_{ ( x - \frac{1}{2}, x + \frac{1}{2} ] }\mathcal{N}\left( \mu, \sigma \right)\left( x \right)
    \enspace ,
    $$
    where $\mathcal{N}\left( \mu, \sigma \right)$ is the probability density function of the continuous Gaussian distribution (Definition \ref{definition:continuous_gaussian}) with mean $\mu$ and standard deviation $\sigma$.
\end{definition}

\begin{definition} [Rounded truncated Gaussian distribution] \label{definition:rounded_truncated_gaussian}
    Let $\mu \in \bbR$, $\sigma \in \bbR_{\geq 0}$ and $B \in \Nat$. The rounded $B$-truncated Gaussian distribution with mean $\mu$ and standard deviation $\sigma$, denoted $\mathcal{N}^{B}\left( \mu, \sigma \right)$ is the distribution with sample space $\{ -B, \cdots, -1, 0, 1, \cdots, B \}$ such that,
    $$
    \forall x \in \{ -B, \cdots, -1, 0, 1, \cdots, B \} :
    $$
    $$
    \mathcal{N}^{B}\left( \mu, \sigma \right)\left( x \right)
    :=
    \frac
    {
    \mathcal{N}^{\bbZ}\left( \mu, \sigma \right)\left( x \right)
    }
    {
    \sum_{ y \in \{ -B, \cdots, 1, 0, 1, \cdots, B \} }
    \mathcal{N}^{\bbZ}\left( \mu, \sigma \right)\left( y \right)
    }
    \enspace ,
    $$
    where $\mathcal{N}^{\bbZ}\left( \mu, \sigma \right)$ is the probability density function of the rounded Gaussian distribution (Definition \ref{definition:rounded_gaussian}) with mean $\mu$ and standard deviation $\sigma$.
\end{definition}

We use the following known statistical relations between the rounded Gaussian distribution and variations of the Gaussian distribution.

\begin{fact} [Distance properties of the rounded Gaussian] \label{fact:rounded_gaussian_properties}
    There exists a positive absolute constant $c_{G} \in \bbR_{> 0}$ such that, for every $B \in \Nat$, all of the above total variation distances are bounded by $\frac{ c_{G} }{ B }$:
    $$
    \norm{
    \mathcal{D}\left( \mathcal{N}^{ \bbZ }\left( 0, B - 1 \right) \right)
    -
    \mathcal{D}\left( \mathcal{N}^{ \bbZ }\left( 0, B \right) \right)
    }_{TV}
    \enspace ,
    $$
    $$
    \norm{
    \mathcal{D}\left( \mathcal{N}^{ \bbZ }\left( 0, B \right) \right)
    -
    \mathcal{D}\left( \mathcal{N}^{ B }\left( 0, B \right) \right)
    }_{TV}
    \enspace ,
    $$
    $$
    \norm{
    \mathcal{D}\left( \mathcal{N}^{ B }\left( 0, B \right) \right)
    -
    \mathcal{D}\left( D_{B} \right)
    }_{TV}
    \enspace .
    $$
\end{fact}

\begin{fact} [Multivariate discrete Gaussian vector norm concentration bounds] \label{fact:gaussian_vector_norm_concentration}
    There exists an absolute positive constant $c_{D} \in \bbR_{> 0}$ such that for all $B, n \in \Nat$,
    $$
    \Pr_{ u \gets D_{B}^{n} }
    \Big[
    \frac{\sqrt{n} \cdot B}{\sqrt{2}}
    \leq
    \norm{u}
    \leq
    \left( 2 - \frac{1}{\sqrt{2}} \right)\cdot \sqrt{n} \cdot B
    \Big]
    \geq
    1 - e^{-c_{D} \cdot n}
    \enspace .
    $$
\end{fact}

We also use the following standard property of total variation distance between functions of independent r.v.'s.

\begin{fact} [Total variation distance between switched-input random variable] \label{fact:tv_distance_switched_rv}
    Let $X, Y$ independent random variables and let $f, g$ (not necessarily deterministic) functions on the union of the supports of $X$ and $Y$. Then, there exists a number $d_{\left( X, Y \right)} \in [ -\norm{ \mathcal{D}\left( X \right) - \mathcal{D}\left( Y \right) }_{TV}, \norm{ \mathcal{D}\left( X \right) - \mathcal{D}\left( Y \right) }_{TV} ]$, such that
    $$
    \norm{ \mathcal{D}\left( f\left( X \right) \right) - \mathcal{D}\left( g\left( X \right) \right) }_{TV}
    =
    \norm{ \mathcal{D}\left( f\left( Y \right) \right) - \mathcal{D}\left( g\left( Y \right) \right) }_{TV}
    +
    2 \cdot d_{\left( X, Y \right)}
    \enspace .
    $$
\end{fact}

\paragraph{Strengthening of the central limit theorem.}
The central limit theorem (CLT), one of the fundamental theorems in statistics, tells us that when we take $\kappa \in \Nat$ independent and continuous random variables $X_{1}, \cdots, X_{\kappa}$ with mean $0$, the random variable $\frac{ \sum_{i \in [\kappa]} X_{i} }{ \sqrt{\kappa} }$ approaches the Gaussian distribution, as $\kappa$ goes to infinity. Our main barrier with using the CLT is the combination of two facts: (1) The CLT needs division for convergence, i.e. it looks at the variable $\frac{ \sum_{i \in [\kappa]} X_{i} }{ \sqrt{\kappa} }$ rather than $\sum_{i \in [\kappa]} X_{i}$, and (2) As part of the Technical Overview (Section \ref{section:overview}) and Section \ref{section:gcvp_reduction} we explain why we need to sample from the Gaussian distribution using a sequence of invertible functions. As we are dealing with integers only, this means that in order to use the CLT we need to execute integer (i.e., rounded) division, which is not invertible. Other, more minor problems are that the CLT guarantees convergence in Wasserstein-1 distance and we would like convergence in the stronger total variation distance, and finally, the CLT does not tell us the rate at which the samples approach the Gaussian distribution, as this depends, among other things, on the probability density functions of the summed variables $\{ X_{i} \}_{i \in [\kappa]}$. We will use the following strengthening Theorem from \cite{chen2010normal}:

\begin{theorem} [Discrete un-normalized Berry-Esseen bound for total variation distance -- Theorem 7.4 from \cite{chen2010normal}] \label{theorem:discrete_TV_distance_berry_esseen}
    Let $\kappa \in \Nat$ and let $\{ X_{i} \}_{i \in [\kappa]}$ independent random variables with finite supports contained in $\bbZ$, and denote $S := \sum_{i \in [\kappa]} X_{i}$. Denote,
    \begin{itemize}
        \item
        For every $i \in [\kappa]$, the expectation of $X_{i}$, $\mu_{i} := \bbE\left( X_{i} \right)$.

        \item 
        For every $i \in [\kappa]$, the variance of $X_{i}$, $\sigma_{i}^{2} := \Var\left( X_{i} \right) := \bbE\left( X_{i}^{2} \right) - \bbE\left( X_{i} \right)^{2}$.

        \item 
        The variance of the sum $S$, $\sigma^{2} := \Var\left( S \right)$, which equals $\sum_{i \in [\kappa]} \sigma_{i}^{2}$, due to the random variables being independent.

        \item 
        For every $i \in [\kappa]$, the third moment of $X_{i}$, $\gamma_{i} := \frac{ \bbE_{X_{i}}\left( |X_{i} - \mu_{i}|^{3} \right) }{ \sigma^{3} }$.

        \item 
        The expectation of $S$, $\mu := \sum_{ i \in [\kappa] }\mu_{i}$.

        \item 
        The sum of third moments $\gamma := \sum_{ i \in [\kappa] }\gamma_{i}$.

        \item 
        For every $i \in [ \kappa ]$, the random variable $S^{(i)} := S - X_{i}$, that is, the original full sum $S$, excluding the $i$-th variable.
    \end{itemize}
    Then we have the following upper bound on the total variation distance,
    $$
    \norm{ \mathcal{D}\left( S \right) - \mathcal{D}\left( \mathcal{N}^{ \bbZ }\left( \mu, \sigma \right) \right) }_{TV}
    $$
    $$
    \leq
    \frac{3}{2} \cdot \sigma \cdot \sum_{i \in [\kappa]}
    \left(
    \left( \gamma_{i} + \frac{ 2 \cdot \sigma_{i}^{2} }{ 3 \cdot \sigma^{3} } \right)
    \cdot
    \norm{ \mathcal{D}\left( S^{(i)} \right) - \mathcal{D}\left( S^{(i)} + 1 \right) }_{TV}
    \right)
    $$
    $$
    +
    \left( 5 + 3\sqrt{\frac{\pi}{8}} \right)\cdot \gamma
    +
    \frac{1}{\sigma \cdot 2 \sqrt{2\pi}}
    \enspace .
    $$
\end{theorem}

\section{New Notions and Definitions} \label{section:new_notions_and_definitions}
In this section we define the new notions that this work presents. We start with the computational model and then a comparison to standard quantum computation. We then proceed to defining sequentially invertible circuits and the sequentially invertible statistical difference problem $\SISD$.

\subsection{Quantum Computation with Computable Order Interference} \label{subsection:new_definitions_order_interference}
At the heart of our model is the order interference (OI) oracle, which gets as input an amplification parameter $\lambda \in \Nat$, a set of $m \in \Nat$ unitaries $\{ U_{i} \}_{i \in [m]}$ (all acting on the same number of qubits, $n \in \Nat$) and an $n$-qubit state $\ket{\psi}$. The oracle then makes an attempt to execute an OI transformation, and if the OI transformation succeeds, we get the OI state. The relevant definitions follow. 

\begin{definition} [The Uniform Order Interference State] \label{definition:uniform_order_interference_state}
    Let $\{ U_{i} \}_{i \in [m]}$ a set of $m \in \Nat$ unitary transformations which all operate on the same number of qubits $n \in \Nat$ and let $\ket{\psi}$ an $n$-qubit quantum state.
    
    The uniform order interference state of the pair $\left( \{ U_{i} \}_{i \in [m]}, \ket{\psi} \right)$ is
    $$
    \OI\left( \{ U_{i} \}_{i \in [m]}, \ket{\psi} \right)
    :=
    \sum_{\sigma \in S_{m}} \left( \prod_{i \in [m]} U_{\sigma^{-1}\left( i \right)} \right) \cdot \ket{\psi} \enspace .
    $$
\end{definition}

After defining the uniform OI state, we define the uniform OI oracle as follows.

\begin{definition} [Computable Uniform Order Interference Oracle] \label{definition:computable_uniform_order_interference_oracle}
    The computable uniform order interference oracle, denoted $\OOI$, is an oracle with input triplet $\left( \{ U_{i} \}_{i \in [m]}, R, \lambda \right)$, such that,
    \begin{itemize}
        \item
        $\{ U_{i} \}_{i \in [m]}$ is a set of $m \in \Nat$ unitary circuits which all operate on the same number of qubits $n \in \Nat$,

        \item
        $R$ is an $n$-qubit quantum register, and

        \item 
        $\lambda \in \Nat$ is a natural number.
    \end{itemize}

    Given valid input $\left( \{ U_{i} \}_{i \in [m]}, R, \lambda \right)$, let $R'$ a $n'$-qubit register such that the joint system $\left( R, R' \right)$ is in a pure state, written as,
    $$
    \sum_{ y \in \{ 0, 1 \}^{n'} }
    \alpha_{y} \cdot \ket{\psi_{y}}_{R}\ket{y}_{R'}
    \enspace .
    $$

    The oracle call takes time complexity $\sum_{i \in [m]}|U_{i}| + \lambda$. At the end of the oracle execution it returns a success/failure bit $b \in \{ 0, 1 \}$, where the success probability is
    $$
    \min_{ y \in \{ 0, 1 \}^{n'} }
    \left(
    \left(
    \frac
    { \sum_{ x \in \{ 0, 1 \}^{n} } |\sum_{ \sigma \in S_{m} } \alpha_{y, x, \sigma}| }
    { \sum_{ x \in \{ 0, 1 \}^{n} } \sum_{ \sigma \in S_{m} } |\alpha_{y, x, \sigma}| }
    \right)
    \cdot 
    \frac
    { \frac{ \norm{\OI\left( \{ U_{i} \}_{i \in [m]}, \ket{\psi_{y}} \right)} }{ m! } }
    { \frac{ \norm{\OI\left( \{ U_{i} \}_{i \in [m]}, \ket{\psi_{y}} \right)} }{ m! } + \frac{1}{\lambda} } 
    \right)
    \enspace ,
    $$
    where for $y \in \{ 0, 1 \}^{n'}$, $\sigma \in S_{m}$ we write
    $$
    \prod_{i \in [m]} U_{\sigma^{-1}\left( i \right)} \cdot \ket{\psi_{y}}
    :=
    \sum_{ x \in \{ 0, 1 \}^{n} } \alpha_{y, x, \sigma} \ket{x}
    \enspace .
    $$
    If the call succeeded, the state in $\left( R, R' \right)$ is the normalization of,
    $$
    \sum_{ y \in \{ 0, 1 \}^{n'} }
    \alpha_{y} \cdot \OI\left( \{ U_{i} \}_{i \in [m]}, \ket{\psi_{y}} \right)_{R}\ket{y}_{R'}
    \enspace ,
    $$
    and if it failed, the state in $\left( R, R' \right)$ may be arbitrary.

    - The unitaries $\{ U_{i} \}_{i \in [m]}$ can also be given in the form of controlled oracle access. In that case, the complexity $|U_{i}|$ for $i \in [m]$ is the complexity of making an oracle call to the $(n + 1)$-qubit unitary which is the controlled $U_{i}$. 
\end{definition}

To make using an OI oracle easier, when the $n$-qubit input register $R$ to the oracle is in a pure state, the outcome of an oracle access is simplified.

\begin{corollary} [Corollary - Uniform OI on pure states] \label{corollary_definition:oi_pure_states}
    Let $\left( \{ U_{i} \}_{i \in [m]}, R, \lambda \right)$ such that $\{ U_{i} \}_{i \in [m]}$ is a set of $m \in \Nat$ unitary circuits which all operate on the same number of qubits $n \in \Nat$, $\lambda \in \Nat$ is a natural number and $R$ is an $n$-qubit register in some pure state $\ket{\psi}$.

    Then, the oracle call $\OOI\left( \{ U_{i} \}_{i \in [m]}, R, \lambda \right)$ succeeds with probability,
    $$
    \left(
    \frac
    { \sum_{ x \in \{ 0, 1 \}^{n} } |\sum_{ \sigma \in S_{m} } \alpha_{x, \sigma}| }
    { \sum_{ x \in \{ 0, 1 \}^{n} } \sum_{ \sigma \in S_{m} } |\alpha_{x, \sigma}| }
    \right)
    \frac
    { \frac{ \norm{\OI\left( \{ U_{i} \}_{i \in [m]}, \ket{ \psi } \right)} }{ m! } }
    { \frac{ \norm{\OI\left( \{ U_{i} \}_{i \in [m]}, \ket{ \psi } \right)} }{ m! } + \frac{1}{\lambda} } 
    \enspace ,
    $$
    where for $\sigma \in S_{m}$ we write
    $$
    \prod_{i \in [m]} U_{\sigma^{-1}\left( i \right)} \cdot \ket{\psi} := \sum_{ x \in \{ 0, 1 \}^{n} } \alpha_{x, \sigma} \ket{x}
    \enspace .
    $$
    and with (at least) the same probability, the state in $R$ after the oracle call is the normalization of the uniform OI state, $\OI\left( \{ U_{i} \}_{i \in [m]}, \ket{ \psi } \right)$. In case the call fails, the state in $R$ may be arbitrary.
\end{corollary}

\paragraph{General form of OI states and oracles.}
While this work studies only the the computational power of a quantum computer having the ability to generate uniform superpositions of unitary execution orders, it is an interesting mathematical question for future work to understand what's possible to compute when the superposition of orders is not necessarily uniform. To this end, we next define the generalized OI state and generalized OI oracle. Both of the latter have a fourth parameter $U_{O}$, a unitary circuit which determines the superposition of orders. We are overloading notations and still use $\OI$, such that when we are dropping the parameter for $U_{O}$, we are automatically referring to the uniform definitions of OI state and OI oracle.

\begin{definition} [The Order Interference State] \label{definition:order_interference_oracle}
    Let $\{ U_{i} \}_{i \in [m]}$ a set of $m \in \Nat$ unitary transformations which all operate on the same number of qubits $n \in \Nat$, let $\ket{\psi}$ an $n$-qubit quantum state, and let $U_{O}$ a $k$-qubit unitary circuit on registers $\left( O, A \right)$ where $O$ has $m\cdot \ceil{\log\left( m \right)}$ qubits and $A$ has the rest.

    Write the generated state by $U_{O}$ in the following way (recall that we use $m\cdot \ceil{\log\left( 
m \right)}$ bits to represent a permutation in $S_{m}$):
    $$
    U_{O} \cdot \ket{0^{k}}
    =
    \sum_{\sigma \in S_{m}} \alpha_{\sigma} \cdot \ket{\sigma}_{O}\ket{\psi_{\sigma}}_{A} \enspace .
    $$
    
    The order interference state of the triplet $\left( \{ U_{i} \}_{i \in [m]}, \ket{\psi}, U_{O} \right)$ is defined as follows.
    $$
    \OI\left( \{ U_{i} \}_{i \in [m]}, \ket{\psi}, U_{O} \right)
    :=
    \sum_{\sigma \in S_{m}} \alpha_{\sigma} \cdot \left( \prod_{i \in [m]} U_{\sigma^{-1}\left( i \right)} \right) \cdot \ket{\psi} \enspace .
    $$

    - When the notation $\OI\left( \cdot \right)$ is used with $2$ parameters instead of $3$, we automatically refer to the uniform OI state, as in definition \ref{definition:uniform_order_interference_state}.
\end{definition}

We next define the generalized OI oracle. One thing to notice is that unlike in the case of the uniform OI state, the general OI state is such that amplitudes are taken into consideration. The consideration makes a shorter vector. This is taken into consideration when calculating the norm of the general OI state, which is relevant for the success probability of an oracle call. Formally, this pops up in the below definition of a general OI oracle, where the success probability of the call scales down the OI state by a factor of $\sqrt{m!}$ rather than by a factor of $m!$, as in the uniform case.

\begin{definition} [Computable Order Interference Oracle] \label{definition:computable_order_interference_oracle}
    A computable order interference oracle, denoted $\OOI$, is an oracle with input quadruple $\left( \{ U_{i} \}_{i \in [m]}, R, U_{O}, \lambda \right)$, such that,
    \begin{itemize}
        \item
        $\{ U_{i} \}_{i \in [m]}$ is a set of $m \in \Nat$ unitary circuits which all operate on the same number of qubits $n \in \Nat$,

        \item
        $R$ is an $n$-qubit quantum register,

        \item
        $U_{O}$ is a unitary on registers $\left( O, A \right)$ where $O$ has $m\cdot \ceil{\log\left( m \right)}$ qubits and $A$ has some number of qubits $k \in \Nat$,

        \item 
        and $\lambda \in \Nat$ is a natural number.
    \end{itemize}
    
    Given valid input $\left( \{ U_{i} \}_{i \in [m]}, R, U_{O}, \lambda \right)$, let $R'$ a $n'$-qubit register such that the joint system $\left( R, R' \right)$ is in a pure state, written as,
    $$
    \sum_{ y \in \{ 0, 1 \}^{n'} }
    \alpha_{y} \cdot \ket{\psi_{y}}_{R}\ket{y}_{R'}
    \enspace .
    $$

    The oracle call takes time complexity $|U_{O}| + \sum_{i \in [m]}|U_{i}| + \lambda$. At the end of the oracle execution it returns a success/failure bit $b \in \{ 0, 1 \}$, where the success probability is
    $$
    \min_{ y \in \{ 0, 1 \}^{n'} }
    \left(
    \left(
    \frac
    { \sum_{ x \in \{ 0, 1 \}^{n} } |\sum_{ \sigma \in S_{m} } \alpha_{y, x, \sigma}| }
    { \sum_{ x \in \{ 0, 1 \}^{n} } \sum_{ \sigma \in S_{m} } |\alpha_{y, x, \sigma}| }
    \right)
    \frac
    { \frac{ \norm{\OI\left( \{ U_{i} \}_{i \in [m]}, \ket{\psi_{y}}, U_{O} \right)} }{ \sqrt{ m! } } }
    { \frac{ \norm{\OI\left( \{ U_{i} \}_{i \in [m]}, \ket{\psi_{y}}, U_{O} \right)} }{ \sqrt{ m! } } + \frac{1}{\lambda} } 
    \right)
    \enspace ,
    $$
    where for $y \in \{ 0, 1 \}^{n'}$, $\sigma \in S_{m}$ we write
    $$
    \prod_{i \in [m]} U_{\sigma^{-1}\left( i \right)} \cdot \ket{\psi_{y}}
    :=
    \sum_{ x \in \{ 0, 1 \}^{n} } \alpha_{y, x, \sigma} \ket{x}
    \enspace .
    $$
    If the call succeeded, the state in $\left( R, R' \right)$ is the normalization of,
    $$
    \sum_{ y \in \{ 0, 1 \}^{n'} }
    \alpha_{y} \cdot \OI\left( \{ U_{i} \}_{i \in [m]}, \ket{\psi_{y}}, U_{O} \right)_{R}\ket{y}_{R'}
    \enspace ,
    $$
    and if it failed, the state in $\left( R, R' \right)$ may be arbitrary.

    - When the notation of an oracle call $\OOI\left( \cdot \right)$ is used with $3$ parameters instead of $4$, we automatically refer to the computable uniform OI oracle, as in definition \ref{definition:computable_uniform_order_interference_oracle}.

    - The unitaries $\{ U_{i} \}_{i \in [m]}$ can also be given in the form of controlled oracle access. In that case, the complexity $|U_{i}|$ for $i \in [m]$ is the complexity of making an oracle call to the $(n + 1)$-qubit unitary which is the controlled $U_{i}$. 
\end{definition}

\begin{remark} [General computable OI implies uniform computable OI]
    The explanation for why the uniform versions of the OI notions are indeed private cases of the general cases, and formally, why the general computable OI oracle can simulate the uniform OI oracle, is as follows. By using the unitary $U_{R}$ from Claim \ref{claim:factorial_transformation}, followed by the unitary $U_{FY}$ from Claim \ref{claim:fisher_yates_sampler}, we have a unitary $U_{uni}$ on $2\cdot m \cdot \ceil{ \log\left( m \right) }$ qubits (possibly using extra ancillary qubits) that generates,
    $$
    U_{uni} \cdot \ket{ 0^{ 2 \cdot m \cdot \ceil{ \log\left( m \right) } } }
    =
    \frac{1}{ \sqrt{m!} }
    \sum_{ \sigma \in S_{m} }
    \ket{ r_{\sigma} }\ket{ \sigma }
    \enspace .
    $$
    More elaboration on this is given as part of the proof of Lemma \ref{lemma:OI_quantum_simulation}.
    Thus, in order to apply uniform OI transformations using the general OI oracle, we execute a call to the general OI oracle with $U_{O} := U_{uni}$.
\end{remark}

Throughout we use the following lemma, which compiles the definition of a uniform OI oracle and the needed amplification parameter to make the success probability arbitrarily close to $1$.

\begin{lemma}[Amplification of Success Probability of OI Transformations] \label{lemma:OI_oracle_amplification_success_probability}
    Let $n, m \in \Nat$, let $\{ U_{i} \}_{i \in [m]}$ a set of $n$-qubit unitaries, and let $\ket{\psi}$ an $n$-qubit quantum state.
    Let $k \in \Nat$ and let
    $$
    \lambda
    :=
    \ceil{
    k\cdot
    \frac{m!}{ \norm{ \OI\left( \{ U_{i} \}_{i \in [m]}, \ket{\psi} \right) } }
    }
    \enspace .
    $$
    Then,
    $$
    \frac
        { \frac{ \norm{\OI\left( \{ U_{i} \}_{i \in [m]}, \ket{\psi}\right)} }{ m! } }
        { \frac{ \norm{\OI\left( \{ U_{i} \}_{i \in [m]}, \ket{\psi}\right)} }{ m! } + \frac{1}{\lambda} }
    \geq 1 - \frac{1}{k}
    \enspace .
    $$
\end{lemma}

\begin{proof}
    To prove the claim, all that is needed is to calculate the success probability, with accordance to the definition of the OI oracle.
    \begin{align*}
    &
        \frac
        { \frac{ \norm{\OI\left( \{ U_{i} \}_{i \in [m]}, \ket{\psi}\right)} }{ m! } }
        { \frac{ \norm{\OI\left( \{ U_{i} \}_{i \in [m]}, \ket{\psi}\right)} }{ m! } + \frac{1}{\lambda} }
        \\&=
        \frac
        { \frac{ \norm{\OI\left( \{ U_{i} \}_{i \in [m]}, \ket{\psi}\right)} }{ m! }  + \frac{1}{\lambda} - \frac{1}{\lambda}}
        { \frac{ \norm{\OI\left( \{ U_{i} \}_{i \in [m]}, \ket{\psi}\right)} }{ m! } + \frac{1}{\lambda} }
        \\&=
        1 - 
        \frac
        { \frac{1}{\lambda} }
        { \frac{ \norm{\OI\left( \{ U_{i} \}_{i \in [m]}, \ket{\psi}\right)} }{ m! } + \frac{1}{\lambda} }
        \\&=
        1 - 
        \frac
        { 1 }
        { \lambda \cdot \frac{ \norm{\OI\left( \{ U_{i} \}_{i \in [m]}, \ket{\psi}\right)} }{ m! } + 1 }
        \\&\geq
        1 - 
        \frac
        { 1 }
        { \lambda \cdot \frac{ \norm{\OI\left( \{ U_{i} \}_{i \in [m]}, \ket{\psi}\right)} }{ m! } }
        \\&
        \underset{\left( \text{Definition of $\lambda$} \right)}{\geq}
        1 - 
        \frac
        { 1 }
        { k }
        \enspace .
    \end{align*}
\end{proof}

\paragraph{Quantum Computation with Order Interference - Complexity Class.}
We give the definition of the set of computational problems solvable in polynomial time, by a quantum computer with the ability to create superposition of unitary execution orders. The class $\BQP^{\OI}$ capturing the computational power of a polynomial-time quantum computer with access to a (generalized) OI oracle.

\begin{definition} [The Complexity Class $\BQP^{\OI}$] \label{definition:bqp_oi}
    The complexity class $\BQP^{\OI}$ is the set of promise problems $\prod := \left( \YES , \NO \right)$ such that there exists $M^{\OI}$ a $\QPT$ with access to an OI oracle $\OOI$ (as in Definition \ref{definition:computable_order_interference_oracle}), such that for every $x \in \YES$, $M^{\OI}(x) = 1$ with probability at least $p(x)$, and if $x \in \NO$, $M^{\OI}(x) = 0$ with probability at least $p(x)$, such that
    $$
    \forall x \in \left( \YES \cup \NO \right) : p(x) \geq 1 - 2^{-\poly\left(|x|\right)} \enspace ,
    $$
    for some polynomial $\poly(\cdot)$.
\end{definition}

\begin{remark} [Different versions of the OI oracle] \label{remark:stronger_OI_oracle_versions}
The definition of the success probability we use in Definitions \ref{definition:computable_uniform_order_interference_oracle}, \ref{definition:computable_order_interference_oracle} is restrictive in two regards.
First, the non-amplifiable component,
$$
\frac
{ \sum_{ x \in \{ 0, 1 \}^{n} } | \sum_{ \sigma \in S_{m} } \alpha_{x, \sigma} | }
{ \sum_{ x \in \{ 0, 1 \}^{n} } \sum_{ \sigma \in S_{m} } | \alpha_{x, \sigma} | }
$$
is intended to restrict the success probability, proportionally to how much the phases of states in different spacetimes (or, unitary evolution orders) are misaligned, as we do not assume a particular destructive interference dynamics between the same states from different spacetimes. The below Definition \ref{definition:no_phase_alignment_computable_order_interference_oracle} is equivalent to the original definition of the OI oracle, without the phase alignment restriction. This definition was used as the main definition for the OI oracle in an earlier version of this work.

Second, the amplifiable norm component 
$$
\frac
{ \frac{ \norm{\OI\left( \{ U_{i} \}_{i \in [m]}, \ket{ \psi } \right)} }{ m! } }
{ \frac{ \norm{\OI\left( \{ U_{i} \}_{i \in [m]}, \ket{ \psi } \right)} }{ m! } + \frac{1}{\lambda} }
$$
is intended to restrict the ability to arbitrarily execute OI transformation, ignoring the norm of the OI state to-be-generated. 
The below Definition \ref{definition:no_norm_amplification_computable_order_interference_oracle} is equivalent to the original definition of the OI oracle, without the norm amplification restriction. Note that in the case of Definition \ref{definition:no_norm_amplification_computable_order_interference_oracle}, the OI oracle's input consists of one less parameter, as the amplification parameter $\lambda$ is dropped. 

As noted in \cite{Priv_Comm_Scott, Priv_Comm_Joe} following an earlier version of this work, a quantum computer with access to a strengthened OI oracle as in Definition \ref{definition:no_phase_alignment_computable_order_interference_oracle} allows for a polynomial time solution for the entire complexity class $\NP$ (and even $\QCMA$). We leave it as an open problem what is the computational power of quantum computers with access to oracles as in Definitions \ref{definition:no_phase_alignment_computable_order_interference_oracle}, \ref{definition:no_norm_amplification_computable_order_interference_oracle}.
\end{remark}

\begin{definition} [Computable Order Interference Oracle without Phase Alignment] \label{definition:no_phase_alignment_computable_order_interference_oracle}
    A computable order interference oracle with no phase alignment, denoted $\OOI^{\omega}$, is an oracle with input quadruple $\left( \{ U_{i} \}_{i \in [m]}, R, U_{O}, \lambda \right)$, such that,
    \begin{itemize}
        \item
        $\{ U_{i} \}_{i \in [m]}$ is a set of $m \in \Nat$ unitary circuits which all operate on the same number of qubits $n \in \Nat$,

        \item
        $R$ is an $n$-qubit quantum register,

        \item
        $U_{O}$ is a unitary on registers $\left( O, A \right)$ where $O$ has $m\cdot \ceil{\log\left( m \right)}$ qubits and $A$ has some number of qubits $k \in \Nat$,

        \item 
        and $\lambda \in \Nat$ is a natural number.
    \end{itemize}
    
    Given valid input $\left( \{ U_{i} \}_{i \in [m]}, R, U_{O}, \lambda \right)$, let $R'$ a $n'$-qubit register such that the joint system $\left( R, R' \right)$ is in a pure state, written as,
    $$
    \sum_{ y \in \{ 0, 1 \}^{n'} }
    \alpha_{y} \cdot \ket{\psi_{y}}_{R}\ket{y}_{R'}
    \enspace .
    $$

    The oracle call takes time complexity $|U_{O}| + \sum_{i \in [m]}|U_{i}| + \lambda$. At the end of the oracle execution it returns a success/failure bit $b \in \{ 0, 1 \}$, where the success probability is
    $$
    \min_{ y \in \{ 0, 1 \}^{n'} }
    \left(
    \frac
    { \frac{ \norm{\OI\left( \{ U_{i} \}_{i \in [m]}, \ket{\psi_{y}}, U_{O} \right)} }{ \sqrt{ m! } } }
    { \frac{ \norm{\OI\left( \{ U_{i} \}_{i \in [m]}, \ket{\psi_{y}}, U_{O} \right)} }{ \sqrt{ m! } } + \frac{1}{\lambda} } 
    \right)
    \enspace .
    $$
    If the call succeeded, the state in $\left( R, R' \right)$ is the normalization of,
    $$
    \sum_{ y \in \{ 0, 1 \}^{n'} }
    \alpha_{y} \cdot \OI\left( \{ U_{i} \}_{i \in [m]}, \ket{\psi_{y}}, U_{O} \right)_{R}\ket{y}_{R'}
    \enspace ,
    $$
    and if it failed, the state in $\left( R, R' \right)$ may be arbitrary.

    - The unitaries $\{ U_{i} \}_{i \in [m]}$ can also be given in the form of controlled oracle access. In that case, the complexity $|U_{i}|$ for $i \in [m]$ is the complexity of making an oracle call to the $(n + 1)$-qubit unitary which is the controlled $U_{i}$. 
\end{definition}

\begin{definition} [Computable Order Interference Oracle without Norm Amplification] \label{definition:no_norm_amplification_computable_order_interference_oracle}
    A computable order interference oracle without norm amplification, denoted $\mathcal{O}_{ \norm{ \OI } }$, is an oracle with input triplet $\left( \{ U_{i} \}_{i \in [m]}, R, U_{O} \right)$, such that,
    \begin{itemize}
        \item
        $\{ U_{i} \}_{i \in [m]}$ is a set of $m \in \Nat$ unitary circuits which all operate on the same number of qubits $n \in \Nat$,

        \item
        $R$ is an $n$-qubit quantum register,

        \item
        $U_{O}$ is a unitary on registers $\left( O, A \right)$ where $O$ has $m\cdot \ceil{\log\left( m \right)}$ qubits and $A$ has some number of qubits $k \in \Nat$,
    \end{itemize}
    
    Given valid input $\left( \{ U_{i} \}_{i \in [m]}, R, U_{O} \right)$, let $R'$ a $n'$-qubit register such that the joint system $\left( R, R' \right)$ is in a pure state, written as,
    $$
    \sum_{ y \in \{ 0, 1 \}^{n'} }
    \alpha_{y} \cdot \ket{\psi_{y}}_{R}\ket{y}_{R'}
    \enspace .
    $$

    The oracle call takes time complexity $|U_{O}| + \sum_{i \in [m]}|U_{i}|$. At the end of the oracle execution it returns a success/failure bit $b \in \{ 0, 1 \}$, where the success probability is
    $$
    \min_{ y \in \{ 0, 1 \}^{n'} }
    \left(
    \frac
    { \sum_{ x \in \{ 0, 1 \}^{n} } |\sum_{ \sigma \in S_{m} } \alpha_{y, x, \sigma}| }
    { \sum_{ x \in \{ 0, 1 \}^{n} } \sum_{ \sigma \in S_{m} } |\alpha_{y, x, \sigma}| }
    \right)
    \enspace ,
    $$
    where for $y \in \{ 0, 1 \}^{n'}$, $\sigma \in S_{m}$ we write
    $$
    \prod_{i \in [m]} U_{\sigma^{-1}\left( i \right)} \cdot \ket{\psi_{y}}
    :=
    \sum_{ x \in \{ 0, 1 \}^{n} } \alpha_{y, x, \sigma} \ket{x}
    \enspace .
    $$
    If the call succeeded, the state in $\left( R, R' \right)$ is the normalization of,
    $$
    \sum_{ y \in \{ 0, 1 \}^{n'} }
    \alpha_{y} \cdot \OI\left( \{ U_{i} \}_{i \in [m]}, \ket{\psi_{y}}, U_{O} \right)_{R}\ket{y}_{R'}
    \enspace ,
    $$
    and if it failed, the state in $\left( R, R' \right)$ may be arbitrary.

    - The unitaries $\{ U_{i} \}_{i \in [m]}$ can also be given in the form of controlled oracle access. In that case, the complexity $|U_{i}|$ for $i \in [m]$ is the complexity of making an oracle call to the $(n + 1)$-qubit unitary which is the controlled $U_{i}$. 
\end{definition}

\subsection{Comparison to Standard Quantum Computation} \label{subsection:OI_comparison_to_standard_quantum_computation}
In this section we give a brief comparison between quantum computation with and without computable OI. Specifically, we show techniques in (standard) quantum computation to simulate a computable uniform OI oracle. First, we show a quantum polynomial-time algorithm for the generation of an entangled uniform superposition of permutation states.

\begin{claim} [Factorial transformation] \label{claim:factorial_transformation}
    For every $m \in \Nat$ there exists a polynomial size (in $m$) unitary circuit $U_{R}$ on $m \cdot \ceil{ \log\left( m \right) }$ qubits (possibly using ancillary qubits) that maps
    $$
    \forall x_{m} \in [m], x_{m - 1} \in [m - 1], \cdots, x_{2} \in [2] :
    U_{R} \cdot \ket{ x_{m}, x_{m - 1}, \cdots, x_{2} }
    =
    $$
    $$
    \sum_{ i_{m} \in [m], i_{m - 1} \in [m - 1], \cdots, i_{2} \in [2] }
    \frac{1}{ \sqrt{m!} } \left( \prod_{j \in [m]} \omega_{j}^{x_{j} \cdot i_{j}} \right) \ket{ i_{m}, i_{m - 1}, \cdots, i_{2} }
    $$
\end{claim}

\begin{proof}
    The proof follows readily by recalling the ability to execute the quantum Fourier transform $\QFT_{N}$ for any number $N$ (where $N$ is not necessarily a power of $2$), using a $\ceil{ \log_{2}\left( N \right) }$-qubit unitary quantum circuit $U_{\QFT, N}$ (which possibly uses ancillary qubits) \cite{kitaev1995quantum, mosca2004exact} that executes in polynomial time in $\log_{2}\left( N \right)$. 

    Finally, by executing in parallel, $m$ Fourier transforms of different orders, that is, the $i$-th transform is $\QFT_{i}$, we get our wanted transformation. Formally, define $U_{ R } := \bigotimes_{ j = m, m - 1, \cdots, 3, 2 } U_{\QFT, j}$, and we get,
    $$
    U_{R} \cdot \ket{ x_{m}, x_{m - 1}, \cdots, x_{2} }
    =
    \bigotimes_{ j = m, m - 1, \cdots, 3, 2 } U_{\QFT, j} \cdot \ket{ x_{j} }
    $$
    $$
    =
    \bigotimes_{ j = m, m - 1, \cdots, 3, 2 }
    \left(
    \sum_{ i_{j} \in [j] }
    \omega_{j}^{x_{j} \cdot i_{j}}
    \frac{1}{ \sqrt{j} }\ket{ i_{j} }
    \right)
    $$
    $$
    =
    \sum_{ i_{m} \in [m], i_{m - 1} \in [m - 1], \cdots, i_{2} \in [2] }
    \left(
    \bigotimes_{ j = m, m - 1, \cdots, 3, 2 }
    \frac{1}{ \sqrt{j} } \omega_{j}^{x_{j} \cdot i_{j}} \ket{ i_{j} }
    \right)
    $$
    $$
    =
    \sum_{ i_{m} \in [m], i_{m - 1} \in [m - 1], \cdots, i_{2} \in [2] }
    \frac{1}{ \sqrt{ m! } } \left( \prod_{j \in [m]} \omega_{j}^{x_{ j } \cdot i_{j}} \right) \ket{ i_{m}, i_{m - 1}, \cdots, i_{2} } \enspace ,
    $$
    as needed.
\end{proof}

We next state a claim, part of the proof for which, is part of the proof of correctness of the Fisher-Yates sampling algorithm. The FY algorithm samples a uniformly random permutation $\sigma \in S_{m}$ by having access to randomness of specific structure.

\begin{claim} [Fisher-Yates Permutation Sampler] \label{claim:fisher_yates_sampler}
    Let $m \in \Nat$, let $R_{m} := [m] \times [m - 1] \times \cdots \times [2]$ and let $S_{m}$ the set of all permutations on $[m]$. An element from each of the sets $R_{m}$, $S_{m}$ can be represented by $m\cdot \ceil{ \log\left( m \right) }$ bits.
    \begin{itemize}
        \item
        There exists a bijection $B_{FY} : R_{m} \rightarrow S_{m}$. Note that this means that a uniformly random input from $R_{m}$ produces a uniformly random permutation from $S_{m}$.

        \item 
        There exists a polynomial size (in $m$) unitary circuit $U_{FY}$ on $2\cdot m\cdot \ceil{ \log\left( m \right) }$ qubits (possibly using extra ancillary qubits), such that,
        $$
        \forall r \in R_{m}, y \in \{ 0, 1 \}^{ m\cdot \ceil{ \log\left( m \right) } } : U_{FY} \cdot \ket{r, y} = \ket{ r, y \oplus B_{FY}\left( r \right) } \enspace .
        $$
    \end{itemize}
\end{claim}

\begin{proof}
    The bijection $B_{FY}$ is the mapping from the Fisher-Yates \cite{fisher1953statistical} algorithm. Basically, $B_{FY}$ operates as follows: It starts from the identity permutation: $\sigma_{I} := \left( 1, 2, 3, \cdots, m \right)$, which keeps every element in its place. Then, for every $i = 1, 2, \cdots m - 1$, take $r_{i}$ the $i$-th coordinate in $r$, and swap element $1$ with element $r_{i}$.

    Finally, the unitary $U_{FY}$ follows as a standard fact in quantum computing: Since the function $B_{FY}$ is classically efficiently computable, there is a unitary that computes it in particular.
\end{proof}

Given the correspondence between elements $r \in R_{m}$ and permutations $\sigma \in S_{m}$, for $r \in R_{m}$ we denote $\sigma_{r} := B_{FY}(r)$, and for $\sigma \in S_{m}$ we denote $r_{\sigma} := B^{-1}_{FY}(\sigma)$. Next, we state without proof a trivial fact in quantum computation, namely, that if we are given (even only an oracle) access to controlled versions of unitaries $\{ U_{i} \}_{i \in [m]}$ and any superposition of permutations over the set $S_{m}$, we can generate an \emph{entangled} superposition of unitary execution orders, where the entanglement is between the order register (holding the order information, in the form of the permutation $\sigma \in S_{m}$) and the target register (holding the result of the chosen order of execution).

\begin{fact} \label{fact:entangled_ordered_execution}
    For every $m, n \in \Nat$, there exists a polynomial size (in $m, n$) unitary circuit $U_{O}$ on $m \cdot \ceil{\log\left( m \right)} + n$ qubits (possibly using extra ancillary qubits), that for every set of $m$ unitaries $\{ U_{i} \}_{i \in [m]}$ all act on the same number of qubits $n$, given oracle access to controlled versions of the unitaries in the set, executes the following unitary transformation:
    $$
    \forall \sigma \in S_{m}, x \in \{ 0, 1 \}^{n} :
    U_{O}^{\left( U_{1}^{c}, U_{2}^{c}, \cdots, U_{m}^{c} \right)} \cdot \ket{\sigma}\ket{x}
    $$
    $$
    \ket{\sigma} \left( \prod_{i \in [m]} U_{\sigma^{-1}\left( i \right)} \right) \cdot \ket{x} \enspace .
    $$
\end{fact}

As a brief explanation, the above fact follows by conditioning the execution of unitaries at each step, going (from left to right) on the elements of $\sigma$, and the unitary $U_{i}$ executes if and only if the current element of $\sigma$ is $i$. Finally, we put the above pieces together in order to simulate an execution of an OI oracle, using a regular quantum computer. The simulation of the OI oracle is without the amplification part.

\begin{lemma} \label{lemma:OI_quantum_simulation}
    There is a quantum polynomial time (in its input size) algorithm $Q$, that for input $\{ U_{i} \}_{i \in [m]}$ a set of $m$ unitary circuits, such that all of them act on the same number of qubits $n$, and $\ket{\psi}$ an $n$-qubit state, outputs the normalization of the OI state $\OI\left( \{ U_{i} \}_{i \in [m]}, \ket{\psi} \right)$ with probability at least
    $$
    \left(
    \frac
    { \norm{ \OI\left( \{ U_{i} \}_{i \in [m]}, \ket{\psi} \right) } }
    { m! }
    \right)^{2}
    \enspace .
    $$
\end{lemma}

\begin{proof}
    The algorithm $Q$ is as follows. First, take the unitary $U_{R}$ from Claim \ref{claim:factorial_transformation} and execute it on $\ket{0^{ m \cdot \ceil{\log\left( m \right)} }}$, to obtain
    $$
    \sum_{ i_{m} \in [m], i_{m - 1} \in [m - 1], \cdots, i_{2} \in [2] }
    \frac{1}{ \sqrt{m!} } \ket{ i_{m}, i_{m - 1}, \cdots, i_{2} }
    \enspace .
    $$
    Next, apply the unitary $U_{FY}$ from Claim \ref{claim:fisher_yates_sampler} on the above state concatenated with $\ket{0^{ m \cdot \ceil{\log\left( m \right)} }}$, to obtain
    $$
    \sum_{ i_{m} \in [m], i_{m - 1} \in [m - 1], \cdots, i_{2} \in [2] }
    \frac{1}{ \sqrt{m!} } \ket{ i_{m}, i_{m - 1}, \cdots, i_{2} }\ket{ \sigma_{\left( i_{m}, i_{m - 1}, \cdots, i_{2} \right)} }
    $$
    $$
    =
    \sum_{ r \in R_{m} }
    \frac{1}{ \sqrt{m!} } \ket{ r }\ket{ \sigma_{ r } }
    =
    \frac{1}{ \sqrt{m!} }
    \sum_{ \sigma \in S_{m} }
    \ket{ r_{\sigma} }\ket{ \sigma }
    \enspace ,
    $$
    where the first equality above follows because the function $B_{FY}$ from Claim \ref{claim:fisher_yates_sampler} is a bijection, so a sum over $R_{m}$ can be switched with a sum over $S_{m}$.

    Next, we execute the unitary $U_{O}$ (from Fact \ref{fact:entangled_ordered_execution}) on the right register (containing $\sigma$) of the above state, concatenated with the input register containing the state $\ket{\psi}$, where the oracle access of the unitary $U_{O}$ comes from the fact that the algorithm $Q$ has the unitaries $\{ U_{i} \}_{i \in [m]}$ as part of its input. We get the state
    $$
    \frac{1}{ \sqrt{m!} }
    \sum_{ \sigma \in S_{m} }
    \ket{ r_{\sigma} }\ket{ \sigma }\left( \prod_{i \in [m]} U_{\sigma^{-1}\left( i \right)} \right) \cdot \ket{\psi} 
    \enspace .
    $$

    We now un-compute in two steps. In the first step we execute $U_{FY}$ once again, to uncompute the information of the permutation, and leave only the information of the element in $R_{m}$. After executing $U_{FY}$ and tracing out the middle register that contained $\sigma$, the state is
    $$
    \frac{1}{ \sqrt{m!} }
    \sum_{ \sigma \in S_{m} }
    \ket{ r_{\sigma} }\left( \prod_{i \in [m]} U_{\sigma^{-1}\left( i \right)} \right) \cdot \ket{\psi} 
    \enspace .
    $$
    Finally, the only probabilistic step of the algorithm is forgetting the information in the left register, containing the information $r \in R_{m}$. This is done by executing the projective measurement that tries to project the left register on $\ket{0^{ m \cdot \ceil{\log\left( m \right)} }}$. Formally, we execute the unitary $U_{R}$ once again on the left register. Let us calculate what is the result of the execution.
    $$
    \sum_{\sigma \in S_{m}}
    \frac{1}{ \sqrt{m!} }
    \left( U_{R} \cdot \ket{ r_{\sigma} } \right)
    \left( \prod_{i \in [m]} U_{\sigma^{-1}\left( i \right)} \right) \cdot \ket{\psi}
    $$
    $$
    =
    \sum_{\sigma \in S_{m}}
    \frac{1}{ \sqrt{m!} }
    \left(
    \sum_{ r \in R_{m} }
    \frac{1}{ \sqrt{m!} } \left( \prod_{j \in [m]} \omega_{j}^{ \left( r_{\sigma} \right)_{j} \cdot \left( r \right)_{j} } \right) \ket{ r }
    \right)
    \left( \prod_{i \in [m]} U_{\sigma^{-1}\left( i \right)} \right) \cdot \ket{\psi}
    $$
    $$
    =
    \frac{1}{ m! }
    \sum_{ r \in R_{m} }
    \ket{ r }
    \left(
    \sum_{\sigma \in S_{m}}
    \left(
    \prod_{j \in [m]} \omega_{j}^{ \left( r_{\sigma} \right)_{j} \cdot \left( r \right)_{j} }
    \right)
    \left(
    \prod_{i \in [m]} U_{\sigma^{-1}\left( i \right)}
    \right)
    \right) \cdot \ket{\psi}
    $$
    $$
    =
    \frac{1}{ m! }
    \sum_{ r \in R_{m} \setminus \{ 0_{R_{m}} \} }
    \ket{ r }
    \left(
    \sum_{\sigma \in S_{m}}
    \left(
    \prod_{j \in [m]} \omega_{j}^{ \left( r_{\sigma} \right)_{j} \cdot \left( r \right)_{j} }
    \right)
    \left(
    \prod_{i \in [m]} U_{\sigma^{-1}\left( i \right)}
    \right)
    \right) \cdot \ket{\psi}
    $$
    $$
    +
    \frac{1}{ m! }
    \ket{ 0^{ m \cdot \ceil{ \log\left( m \right) } } }
    \left(
    \sum_{\sigma \in S_{m}}
    \left(
    \prod_{i \in [m]} U_{\sigma^{-1}\left( i \right)}
    \right)
    \right) \cdot \ket{\psi}
    \enspace .
    $$
    Now, note that the last state in the above summation is 
    $$
    \frac{1}{ m! }
    \ket{ 0^{ m \cdot \ceil{ \log\left( m \right) } } }
    \left(
    \sum_{\sigma \in S_{m}}
    \left(
    \prod_{i \in [m]} U_{\sigma^{-1}\left( i \right)}
    \right)
    \right) \cdot \ket{\psi}
    $$
    $$
    =
    \frac{ \norm{ \OI\left( \{ U_{i} \}_{i \in [m]}, \ket{\psi} \right) } }{ m! }
    \ket{ 0^{ m \cdot \ceil{ \log\left( m \right) } } }
    \frac{1}{ \norm{ \OI\left( \{ U_{i} \}_{i \in [m]}, \ket{\psi} \right) } }
    \OI\left( \{ U_{i} \}_{i \in [m]}, \ket{\psi} \right)
    \enspace ,
    $$
    which means that the probability to measure the left register to be $\ket{ 0^{ m \cdot \ceil{ \log\left( m \right) } } }$ and obtain the normalized OI state in the right target register is exactly $\left( \frac{ \norm{ \OI\left( \{ U_{i} \}_{i \in [m]}, \ket{\psi} \right) } }{ m! } \right)^{2}$, as needed.
\end{proof}

To conclude the comparison between quantum computation with and without computable OI, it seems that the only difference is the ability of the OI oracle to amplify the success probability of OI transformations - while a single call to the OI oracle can be simulated efficiently by a quantum computer, for $\ell \in \Nat$ queries, this complexity increases exponentially. This relation, between quantum computation with and without computable OI is somewhat reminiscent of the relationship between quantum and classical computation.

Specifically, an $\ell$-qubit quantum circuit can be broken down (without the loss of generality) into 2-qubit quantum gates. A classical computer can efficiently simulate the action of one step of a quantum computer, which is a single 2-qubit unitary. However, the more gates the classical computer needs to simulate, the complexity of simulation increases exponentially with the number of gates and qubits $\ell$.

\subsection{Sequentially Invertible Distributions and Statistical Difference} \label{subsection:sequential_invertibility_definitions}
As part of this work we define numerous new notions, with no dependence on any non-classical computing or information. Specifically, we define sequentially invertible distributions, and a new computational problem, called the Sequentially Invertible Statistical Difference ($\SISD$) problem. We start with the definitions of sequential invertibility.

\begin{definition} [$\left( r, t, \ell \right)$-invertible circuit sequence] \label{definition:invertible_circuit_sequence}
    For $r, t, \ell, k \in \Nat$, an $\left( r, t, \ell \right)$-invertible circuit sequence on $k$ bits, is a sequence of $\ell$ pairs of randomized circuits $\left( C_{i, \rightarrow}, C_{i, \leftarrow} \right)_{i \in [\ell]}$ such that for every $i \in [\ell]$,
    \begin{itemize}
        \item 
        For every $Y \in \{ \rightarrow, \leftarrow \}$, the circuit $C_{i, Y}$ is of size at most $t$, uses $r_{i} \leq r$ random bits and maps from $k$ to $k$ bits, that is,
        $$
        C_{i, Y} : \{ 0, 1 \}^{ k } \times \{ 0, 1 \}^{ r_{i} } \rightarrow \{ 0, 1 \}^{ k }
        \enspace .
        $$
    
        \item 
        The circuit directions are inverses of each other, per hard-wired randomness, that is:
        $$
        \forall z \in \{ 0, 1 \}^{ r_{i} }, x \in \{ 0, 1 \}^{ k } :
        x = C_{i, \leftarrow}\left( C_{i, \rightarrow}\left( x; z \right); z \right) \enspace .
        $$
        Note that this also implies that for each hard-wired randomness, the circuits act as permutations on the set $\{ 0, 1 \}^{k}$.
    \end{itemize}

    For $r \in \Nat$, an $r$-invertible circuit sequence, is an $\left( r, t, \ell \right)$-invertible circuit sequence, for finite and but unbounded $t, \ell$.
\end{definition}

\begin{definition} [Invertible circuit sequence sampling from $D$]
    Let $r, t, \ell, k \in \Nat$, $\varepsilon \in [0, 1]$, let $S := \left( C_{i, \rightarrow}, C_{i, \leftarrow} \right)_{i \in [\ell]}$ an $\left( r, t, \ell \right)$-invertible circuit sequence on $k$ bits, and let $D$ a distribution on $\{ 0, 1 \}^{k}$.
    
    We say that $S$ samples $\varepsilon$-close to $D$ if for the following distribution $D_{S}$ we have $\norm{ D_{S} - D }_{TV} \leq \varepsilon$,
    $$
        D_{S}
        =
        \biggl\{
        C_{\ell, \rightarrow }
        \left( \cdots 
        C_{2, \rightarrow }
        \left(
        C_{1, \rightarrow }\left( 0^{ k } ; z_1 \right); z_2 \right)
        \cdots ; z_{\ell} \right)
        \; \Big| \;
        z_{1} \gets \{ 0, 1 \}^{ r_{1} },
        \cdots,
        z_{\ell} \gets \{ 0, 1 \}^{ r_{\ell} }
        \biggr\}
        \enspace .
    $$
    If $\varepsilon = 0$, we simply say that $S$ samples from $D$.
\end{definition}

\begin{definition} [$\left( r, t, \ell, \varepsilon \right)$-sequentially invertible distribution]
    Let $D$ a distribution, let $r, t, \ell \in \Nat$ and let $\varepsilon \in [0, 1]$. We say that $D$ is $\left( r, t, \ell, \varepsilon \right)$-sequentially invertible if there exists an $\left( r, t, \ell \right)$-invertible circuit sequence $S$ that samples $\varepsilon$-close to $D$.
\end{definition}

\begin{definition} [$\left( r(n), t(n), \ell(n), \varepsilon(n) \right)$-sequentially invertible distribution family]
    Let $D = \{ D_{n} \}_{n \in \Nat}$ a family of distributions. Let $r, t, \ell : \Nat \rightarrow \Nat$ functions and let $\varepsilon : \Nat \rightarrow [0, 1]$.
    We say that the family of distributions $D$ is $\left( r(n), t(n), \ell(n), \varepsilon(n) \right)$-sequentially invertible if for every $n \in \Nat$, there exists an $\left( r(n), t(n), \ell(n) \right)$-invertible circuit sequence $S_{n}$ that samples $\varepsilon(n)$-close to $D_{n}$.
\end{definition}

Note that when considering sequential invertibility, we can interchange between considering the sequential invertibility of distributions and circuits.

\begin{remark} [Interchanging between sequentially invertible circuits and distributions]
    While we define sequential invertibility for a distribution $D$ over a set $\{ 0, 1 \}^{ k }$, we sometimes refer to sequential invertibility for circuits $C : \{ 0, 1 \}^{k'} \rightarrow \{ 0, 1 \}^{k}$. This extension is natural, as circuits are in particular distributions -- we can consider the output distributions a circuit $C : \{ 0, 1 \}^{k'} \rightarrow \{ 0, 1 \}^{k}$ given a uniformly random $k'$-bit string. Additionally, a circuit contains not only the information of the distribution it computes, but a way to sample from the distribution, using the standard set of classical gates $\{ NOT, OR, AND \}$.
\end{remark}

\paragraph{The Sequentially Invertible Statistical Difference Problem $\SISD_{a(n), b(n), r(n)}$.}
Our last new definition in this work is for a new computational problem. We combine the notions of sequential invertibility presented here above, with the computational problem of Statistical Difference $\SD$, to define the Sequentially Invertible Statistical Difference ($\SISD$) Problem. We take the same steps as we took when we defined the problem $\SD$.

\begin{definition} [The Sequentially Invertible Statistical Difference Problem] \label{definition:sisd_problem}
    Let $a, b : \Nat \rightarrow [0, 1]$, $r : \Nat \rightarrow \Nat \cup \{ 0 \}$ functions, such that for every $n \in \Nat$, $a(n) \leq b(n)$ and $r(n) \in \{ 0, 1, \cdots , n \}$. The $\left( a(n), b(n), r(n) \right)$-gap Sequentially Invertible Statistical Difference problem, denoted $\SISD_{a(n), b(n), r(n)}$ is a promise problem, where the input is a pair of sequences,
    $$
    \left(
    \left( C^{0}_{i, \rightarrow}, C^{0}_{i, \leftarrow} \right)_{i \in [\ell]},
    \left( C^{1}_{i, \rightarrow}, C^{1}_{i, \leftarrow} \right)_{i \in [\ell]}
    \right) \enspace ,
    $$
    where for each $b \in \{ 0, 1 \}$, $\left( C^{b}_{i, \rightarrow}, C^{b}_{i, \leftarrow} \right)_{i \in [\ell]}$ is an $r(n)$-invertible circuit sequence on some (identical) number of bits $k$ (as in Definition \ref{definition:invertible_circuit_sequence}), where $n$ is the input size,
    $$
    n :=
    \bigg|
    \left(
    \left( C^{0}_{i, \rightarrow}, C^{0}_{i, \leftarrow} \right)_{i \in [\ell]},
    \left( C^{1}_{i, \rightarrow}, C^{1}_{i, \leftarrow} \right)_{i \in [\ell]}
    \right)
    \bigg|
    \enspace .
    $$

    For a pair of circuit sequences, we define their output distributions:
    $$
    D_{b}
    :=
    \biggl\{
    C^{b}_{\ell, \rightarrow}\left(
    \cdots
    C^{b}_{2, \rightarrow}\left(
    C^{b}_{1, \rightarrow}\left( 0^{k(n)} ; z_1 \right); z_2 \right)
    \cdots
    ; z_{\ell} \right)
    \biggr\}
    \enspace .
    $$
    The problem is defined as follows.
    \begin{itemize}
        \item
        $\YES$ is the set of pairs of sequences such that $\norm{D_{0} - D_{1}}_{TV} \leq a(n)$,

        \item
        $\NO$ is the set of pairs of sequences such that $\norm{D_{0} - D_{1}}_{TV} > b(n)$.
    \end{itemize}
\end{definition}

Note that in the above definition, while we consider sequentially invertible sampler sequences as the input to the problem, we focus only on $r$-invertibility, and ignore the parameters $t \in \Nat$, $\ell \in \Nat$, $\varepsilon \in [0, 1]$ from the original definition. The reason is that all parameters, except $r$, are implicit in the problem input and problem definition. More precisely, the parameters $t, \ell$ are meaningless as the problem input is the sampler sequence (which defines $t, \ell$), and the total variation distance bound parameter $\varepsilon$ is ignored, as the total variation distance demands depend on $a(n)$ and $b(n)$.

Finally, as in the case of the (standard) Statistical Difference problem $\SD$, we define the family of problems where the gap is not too small, and formally, lower-bounded inverse-polynomially. To this end we define the family of problems $\SISD_{\poly}$ as follows.

\begin{definition} [The Polynomial-Gap Sequentially Invertible Statistical Difference Problem]
    Let $r : \Nat \rightarrow \Nat \cup \{ 0 \}$ a function such that for every $n \in \Nat$, $r(n) \in \{ 0, 1, \cdots , n \}$. The polynomially-bounded sequentially invertible statistical difference problem, denoted $\SISD_{\poly, \, r(n)}$ is a set of promise problems. It is defined as follows.
    $$
    \SISD_{\poly, \, r(n)}
    :=
    \bigcup_{
    \substack{
    a(n), b(n) : \Nat \rightarrow [0,1], \\
    \exists \text{ a polynomial } \poly : \Nat \rightarrow \Nat \text{ such that } \forall n \in \Nat : b(n)^{2} - 2a(n) + a(n)^{2} \geq \frac{1}{\poly(n)}
    }
    }
    \{ \SISD_{a(n), b(n), r(n)} \}
    \enspace .
    $$
\end{definition}

\section{$\SISD_{\poly, O\left( \log(n) \right)} \in \BQP^{\OI}$} \label{section:sisd_algorithm}
In this section we prove that a quantum computer with access to an order interference (OI) oracle is able to solve $\SISD_{\poly, O\left( \log(n) \right)}$ in polynomial time. For the result in this section, a uniform computable OI oracle (Definition \ref{definition:computable_uniform_order_interference_oracle}) suffices, and more so, the simplified Corollary \ref{corollary_definition:oi_pure_states} of the definition.

We first prove a general lemma about an ability of an OI oracle. Basically, the lemma says that an OI oracle can allow an interference between unitary execution choices and not only between unitary execution orders.

\begin{definition} [Choice Interference State] \label{definition:choice_interference_state}
    Let $\{ U_{i} \}_{i \in [m]}$ a set of $m \in \Nat$ unitary circuits, all operating on the same number of qubits $n \in \Nat$. Let $\ket{\psi}$ an $n$-qubit quantum state. The choice interference state of $\left( \{ U_{i} \}_{i \in [m]}, \ket{\psi} \right)$ is denoted $\OI_{c}\left( \{ U_{i} \}_{i \in [m]}, \ket{\psi} \right)$ and defined to be
    $$
    \OI_{c}\left( \{ U_{i} \}_{i \in [m]}, \ket{\psi} \right)
    :=
    \sum_{i \in [m]} U_{i} \cdot \ket{ \psi } \enspace .
    $$
\end{definition}

We are using the following standard fact from quantum computing.

\begin{fact}
    For every $c, N \in \Nat$ there exists an $\ceil{ \log\left( N \right) }$-qubit quantum circuit $U_{\left( +c, N \right)}$ that executes in polynomial-time (in its input size) and possibly uses additional ancillary qubits, such that,
    $$
    \forall a \in \{ 0, 1, 2, \cdots, N - 1 \} : 
    U_{\left( +c, N \right)} \cdot \ket{a} = \ket{a + c \mod N}
    \enspace .
    $$
\end{fact}

The next lemma says that an OI oracle can be used in order to simulate oracle access to choice interference with a factorial reduction in needed complexity, for successfully executing the transformation.

\begin{lemma} [Choice interference from order interference] \label{lemma:choice_interference_oracle_from_OI}
    Let $\OOI\left( \cdot, \cdot, \cdot \right)$ a uniform order interference oracle as in Definition \ref{definition:computable_uniform_order_interference_oracle}. We define the procedure $Q_{CI}$, which is a quantum algorithm with oracle access to $\OOI$. The procedure's input is triplet $\left( \{ U_{i} \}_{i \in [m]}, \ket{\psi}, \lambda \right)$ as in the case of the oracle $\OOI$.

    For every $i \in [m]$ define the unitary circuit $\Tilde{U}_{i}$ that acts on $n + \ceil{ \log\left( m \right) }$ qubits, and is defined as follows:
    \begin{itemize}
        \item
        Execute the circuit $U^{c}_{i}$, which executes $U_{i}$ on the left $n$ qubits, conditioned on the state of the $\ceil{ \log(m) }$ qubits on the right being $\ket{ 0^{ \ceil{ \log(m) } } }$.

        \item 
        Execute the unitary $U_{\left( +1, 2^{ \ceil{ \log(m) } } \right)}$ on the right $\ceil{ \log(m) }$ qubits.
    \end{itemize}
    
    $Q_{CI}\left( \{ U_{i} \}_{i \in [m]}, \ket{\psi}, \lambda \right)$ executes the oracle call
    $$
    \OOI
    \left(
    \{ \Tilde{U}_{i} \}_{ i \in [m] },
    \ket{\psi}\ket{ 0^{ \ceil{\log\left( m \right)} } },
    \lambda
    \right) \enspace ,
    $$
    returns the success/fail bit of the oracle's output and the output state, after tracing out the rightmost $\ceil{\log\left( m \right)}$ qubits.
    The procedure $Q_{CI}$ executes in time $O\left( \lambda + \log\left( m \right) \cdot \left( \sum_{i \in [m]} |U_{i}| \right) \right)$ and with probability at least
    $$
    \left(
    \frac
    { \sum_{ x \in \{ 0, 1 \}^{n} } | \sum_{ i \in [m] } \alpha_{x, i} | }
    { \sum_{ x \in \{ 0, 1 \}^{n} }  \sum_{ i \in [m] } | \alpha_{x, i} | }
    \right)
    \frac
    { \frac{ \norm{\OI_{c}\left( \{ U_{i} \}_{i \in [m]}, \ket{ \psi } \right)} }{ m } }
    { \frac{ \norm{\OI_{c}\left( \{ U_{i} \}_{i \in [m]}, \ket{ \psi } \right)} }{ m } + \frac{1}{\lambda} } 
    \enspace ,
    $$
    outputs the normalization of the choice interference state $\OI_{c}\left( \{ U_{i} \}_{i \in [m]}, \ket{ \psi } \right)$, where for $i \in [m]$, we define $U_{i} \cdot \ket{\psi} := \sum_{ x \in \{ 0, 1 \}^{n} } \alpha_{x, i} \cdot \ket{x}$.
\end{lemma}

\begin{proof}
    Let $\sigma \in S_{m}$ an execution order of the unitaries in $\{ \Tilde{U}_{i} \}_{i \in [m]}$ on the state $\ket{\psi}\ket{ 0^{ \ceil{\log\left( m \right)} } }$.
    We consider what happens to the state in each of the unitary execution orders $\sigma \in S_{m}$. Recall that executing the unitary $\Tilde{U}_{i}$ means executing $U^{c}_{i}$ and then executing $U_{\left( +1, 2^{ \ceil{ \log(m) } } \right)}$. Two things can be easily verified.
    \begin{itemize}
    \item
    For $i \in [m]$ and any execution order $\sigma \in S_{m}$ such that $\sigma^{-1}\left( i \right) = 1$, the state in the left $n$ qubits at the end of the computation is $U_{i} \cdot \ket{\psi}$. In simple terms, the only unitary that effectively executes on the side of the input state $\ket{\psi}$ is the first one in the ordering.

    \item 
    When we look at the state of the right $\ceil{ \log(m) }$-qubit counter register, for \emph{any} permutation $\sigma \in S_{m}$, the state is always $m \mod 2^{ \ceil{ \log(m) } }$.
    \end{itemize}

    For $\sigma \in S_{m}$ we write
    $$
    \prod_{i \in [m]} U_{\sigma^{-1}\left( i \right)} \cdot \ket{\psi} := \sum_{ x \in \{ 0, 1 \}^{n} } \alpha_{x, \sigma} \ket{x}
    \enspace ,
    $$
    and for $i \in [m]$, we define $U_{i} \cdot \ket{\psi} := \sum_{ x \in \{ 0, 1 \}^{n} } \alpha_{x, i} \cdot \ket{x}$. The above observations imply
    $$
    \frac
    { \sum_{ x \in \{ 0, 1 \}^{n} } |\sum_{ \sigma \in S_{m} } \alpha_{x, \sigma}| }
    { \sum_{ x \in \{ 0, 1 \}^{n} } \sum_{ \sigma \in S_{m} } |\alpha_{x, \sigma}| }
    =
    \frac
    { \sum_{ x \in \{ 0, 1 \}^{n} } | \sum_{ i \in [m] } \alpha_{x, i} | }
    { \sum_{ x \in \{ 0, 1 \}^{n} }  \sum_{ i \in [m] } | \alpha_{x, i} | }
    \enspace .
    $$
    
    Also, the above observations imply,
    $$
    \OI
    \left(
    \{ \Tilde{U}_{i} \}_{i \in [m]},
    \ket{\psi}\ket{ 0^{ \ceil{ \log(m) } } }
    \right)
    \underset{(\text{as we explained above})}{=}
    \sum_{\sigma \in S_{m}, i \in [m] \text{ such that } \sigma^{-1}\left( i \right) = 1} U_{i} \cdot \ket{\psi} \ket{ 0^{ \ceil{ \log(m) } } }
    $$
    $$
    =
    \sum_{i \in [m]}\left( \sum_{\sigma \in S_{m} \text{ such that } \sigma^{-1}\left( i \right) = 1} U_{i} \cdot \ket{\psi} \ket{ m \text{ mod } 2^{ \ceil{ \log(m) } } } \right)
    $$
    $$
    =
    \left( m - 1 \right)!
    \sum_{i \in [m]} \left( U_{i} \cdot \ket{\psi} \right) \otimes \ket{ m \text{ mod } 2^{ \ceil{ \log(m) } } } 
    $$
    $$
    :=
    \left( m - 1 \right)! \cdot \left( \OI_{c}\left( \{ U_{i} \}_{i \in [m]}, \ket{\psi} \right) \otimes \ket{ m \text{ mod } 2^{ \ceil{ \log(m) } } } \right)
    \enspace ,
    $$
    which conclude our proof.
\end{proof}

\paragraph{Main Theorems.}
At this point we are ready to prove the two main theorems of this section. We will use the following standard fact from quantum computing.

\begin{fact} [Unitary circuits for Classical Bijections] \label{fact:unitaries_for_classical_bijections}
    Let $C_{\rightarrow} : \{ 0, 1 \}^{k} \rightarrow \{ 0, 1 \}^{k}$, $C_{\leftarrow} : \{ 0, 1 \}^{k} \rightarrow \{ 0, 1 \}^{k}$ two classical circuits that are inverses of each other, that is,
    $$
    \forall x \in \{ 0, 1 \}^{k} :
    C_{\leftarrow}\left( C_{\rightarrow}\left( x \right) \right) = x \enspace .
    $$
    Then, there exists a $k$-qubit unitary circuit $U_{C}$ with complexity $O\left( |C_{\rightarrow}| + |C_{\leftarrow}| \right)$ such that,
    $$
    \forall x \in \{ 0, 1 \}^{k} :
    U_{C} \cdot \ket{x} = \ket{ C_{\rightarrow}\left( x \right) } \enspace .
    $$

    Moreover, there is a classical polynomial-time deterministic Turing Machine $M$ that computes the classical description of the unitary: $\left( C_{\rightarrow}, C_{\leftarrow} \right) \rightarrow U_{C}$.
\end{fact}

We will use the following claim in the proof.

\begin{claim} [Minimal norm of choice interference for classical positive operations] \label{claim:minimal_norm_of_positive_states}
    Let $\ket{\psi} = \sum_{ x \in \{ 0, 1 \}^{k} } \alpha_{x} \cdot \ket{x}$ a $k$-qubit state such that for all $x \in \{ 0, 1 \}^{k}$, the amplitude $\alpha_{x}$ is real and non-negative. Let $\{ U_{i} \}_{i \in [m]}$ a set of $m$ unitary transformations, all act on the same number of qubits $k$, and such that all unitaries map from classical to classical states, formally:
    $$
    \forall i \in [m], x \in \{ 0, 1 \}^{k} \; \exists y_{i, x} \in \{ 0, 1 \}^{k} :
    U_{i} \cdot \ket{ x } = \ket{ y_{i, x} } \enspace .
    $$

    Then, we have the following lower bound on the norm of the choice interference state
    $$
    \norm{ \sum_{i \in [m]}U_{i} \cdot \ket{\psi} } \geq \sqrt{m} \enspace .
    $$
\end{claim}

\begin{proof}
    We calculate the squared norm of the choice interference state for the proof.
    $$
    \norm{ \sum_{i \in [m]} U_{i} \cdot \ket{\psi} }^{2}
    =
    \norm{ \sum_{i \in [m]} \left( U_{i} \cdot \sum_{x \in \{ 0, 1 \}^{k} } \alpha_{x} \cdot \ket{x} \right) }^{2}
    $$
    $$
    =
    \norm{ \sum_{i \in [m]} \left( \sum_{ x \in \{ 0, 1 \}^{k} } \alpha_{x} \cdot \ket{ y_{i, x} } \right) }^{2}
    $$
    $$
    =
    \norm{ \sum_{ x \in \{ 0, 1 \}^{k} } \left( \alpha_{x} \sum_{i \in [m]} \ket{ y_{i, x} } \right) }^{2}
    $$
    $$
    =
    \left(
    \sum_{ x \in \{ 0, 1 \}^{k}, i \in [m] } \alpha_{x} \bra{ y_{i, x} }
    \right)
    \cdot
    \left(
    \sum_{ z \in \{ 0, 1 \}^{k}, j \in [m] } \alpha_{z} \ket{ y_{j, z} }
    \right)
    $$
    $$
    =
    \sum_{ x, z \in \{ 0, 1 \}^{k}, i, j \in [m] } \alpha_{x}\alpha_{z} \bra{ y_{i, x} }\ket{ y_{j, z} }
    $$
    $$
    \underset{(\text{all summands are non-negative})}{\geq}
    \sum_{ x \in \{ 0, 1 \}^{k}, i \in [m] } \alpha_{x}\alpha_{x} \bra{ y_{i, x} }\ket{ y_{i, x} }
    $$
    $$
    =
    \sum_{ x \in \{ 0, 1 \}^{k} } \alpha_{x}^{2} \cdot m
    =
    m \enspace .
    $$
\end{proof}

We next prove the first main theorem of this section.

\begin{theorem} [Quantum State Generation for Invertible Circuit Sequences] \label{theorem:quantum_state_generation_of_invertible_circuit_sequence}
    There exists $Q^{\OI}\left( \cdot, \cdot \right)$ a quantum algorithm with oracle access to a uniform OI oracle (as in Definition \ref{definition:computable_uniform_order_interference_oracle}), that gets as input (1) an amplification parameter $\delta \in \Nat$ and (2) an $r$-invertible circuit sequence $C := \left( C_{i, \rightarrow}, C_{i, \leftarrow} \right)_{i \in [\ell]}$ on $k$ bit strings, such that the total description size of the sequence is denoted with $n$.
    
    The algorithm $Q^{\OI}$ executes in time $O\left( \delta \cdot \poly(n) \cdot 2^{r} \right)$ for some polynomial $\poly(\cdot)$, and outputs a success/fail bit $b \in \{ 0, 1 \}$ together with a $k$-qubit quantum state $\ket{\phi}$. The algorithm succeeds (outputs $b = 1$) with probability $\geq 1 - \frac{1}{\delta}$ and in that case the state $\ket{\phi}$ is the (normalization of the) output distribution state of the circuit sequence:
    \begin{equation}
    \ket{ C(R) }
    :=
    \sum_{
    \substack{
    z_{1} \in \{ 0, 1 \}^{ r_{1} }, \\ \vdots \\ z_{\ell} \in \{ 0, 1 \}^{ r_{\ell} }
    }
    }
    \ket{
    C_{\ell, \rightarrow}\left(
    \cdots
    C_{2, \rightarrow}\left(
    C_{1, \rightarrow}\left( 0^{k} ; z_1 \right);
    z_2 \right)
    \cdots ;
    z_{\ell} \right)
    }
    \enspace .
    \end{equation}
\end{theorem}

\begin{proof}
    Let $k \in \Nat$ the output size of the invertible circuit sequence. For every $i \in [\ell]$, for every $z_{i} \in \{ 0, 1 \}^{r_{i}}$, since this is an invertible circuit sequence, the two circuits $C_{i, \rightarrow}\left( \cdot ; z_{i} \right)$, $C_{i, \leftarrow}\left( \cdot ; z_{i} \right)$ are inverses of each other. It follows from Fact \ref{fact:unitaries_for_classical_bijections} that we can efficiently compute a description of the unitary $U_{i, z_{i}}$ that acts on $k$ qubits, has complexity $O\left( |C_{i, \rightarrow}\left( \cdot ; z_{i} \right)| + |C_{i, \leftarrow}\left( \cdot ; z_{i} \right)| \right)$ and maps,
    $$
    \forall x \in \{ 0, 1 \}^{k} :
    U_{i, z_{i}} \cdot \ket{x} = \ket{ C\left( x ; z_{i} \right) } \enspace .
    $$

    Next, observe that for every $i \in [\ell]$, for every $k$-qubit state $\ket{\phi}$ such that all of its amplitudes are real and non-negative, the vector $\sum_{ z_{i} \in \{ 0, 1 \}^{r_{i}} } U_{i, z_{i}} \cdot \ket{\phi}$ has only real and non-negative amplitudes. This follows because for every $i \in [\ell]$, the unitaries in $\{ U_{i, z_{i}} \}_{ z_{i} \in \{ 0, 1 \}^{r_{i}} }$ all map classical to classical states. It follows that the normalization of the choice interference state $\sum_{ z_{i} \in \{ 0, 1 \}^{r_{i}} } U_{i, z_{i}} \cdot \ket{\phi}$ is a quantum state with only real and non-negative amplitudes.

    According to Claim \ref{claim:minimal_norm_of_positive_states}, for every $i \in [\ell]$ and every such $\ket{\phi}$ as above, 
    $$
    \norm{ \sum_{ z_{i} \in \{ 0, 1 \}^{r_{i}} }U_{i} \cdot \ket{\phi} } \geq \sqrt{ 2^{r_{i}} } \enspace .
    $$
    It thus follows that for every $k$-qubit state $\ket{\phi}$ such that all of its amplitudes are real and non-negative, we have,
    \begin{equation} \label{equation:invertible_circuit_state_generation_step_i}
        \frac{ \norm{ \OI_{c}\left( \{ U_{i, z_{i}} \}_{ z_{i} \in \{ 0, 1 \}^{r_{i}} }, \ket{ \phi } \right) } }{ 2^{r_{i}} }
        \geq 
        \frac{1}{ \sqrt{ 2^{r_{i}} } }
        \geq 
        \frac{1}{ \sqrt{ 2^{r} } }
        \enspace .
    \end{equation}
    By induction on $i$, one can verify that for every $i \in [\ell]$, executing the choice interference oracle (the algorithm $Q_{CI}$ from Lemma \ref{lemma:choice_interference_oracle_from_OI}) on the triplet $\left( \{ U_{i, z_{i}} \}_{ z_{i} \in \{ 0, 1 \}^{r_{i}} }, \ket{ C_{1, \cdots, i - 1}(R_{1, \cdots, i - 1}) }, \delta \cdot \ell \cdot \sqrt{ 2^{r} } \right)$, has the following properties.
    \begin{itemize}
        \item
        The oracle call takes time complexity $O\left( \delta \cdot \ell \cdot \sqrt{ 2^{r} } + r_{i} \cdot \sum_{ z_{i} \in \{ 0, 1 \}^{ r_{i} } } |U_{i, z_{i}}| \right)$, which is bounded by $O\left( \delta \cdot n^{2} \cdot 2^{ r } \right)$, because $\ell \leq n$ and $r_{i} \cdot \sum_{ z_{i} \in \{ 0, 1 \}^{ r_{i} } } |U_{i, z_{i}}| \leq n \cdot 2^{r_{i}} \leq n \cdot 2^{r}$.

        \item 
        The choice interference transformation succeeds with probability $\geq 1 - \frac{1}{\delta \cdot \ell}$, due to Lemma \ref{lemma:OI_oracle_amplification_success_probability} together with Equation \ref{equation:invertible_circuit_state_generation_step_i}.

        \item 
        After a successful transformation, the output state is
        $$
        \ket{ C_{1, \cdots, i}(R_{1, \cdots, i}) }
        :=
        \sum_{
        \substack{
        z_{1} \in \{ 0, 1 \}^{ r_{1} }, \\ \vdots \\ z_{i} \in \{ 0, 1 \}^{ r_{i} }
        }
        }
            \ket{
            C_{i, \rightarrow}\left(
            \cdots
            C_{1, \rightarrow}\left( 0^{k} ; z_1 \right)
            \cdots ;
            z_{i} \right)
        }
        \enspace ,
        $$
        which follows based on the equality,
        $$
        \OI_{c}\left( \{ U_{i, z_{i}} \}_{ z_{i} \in \{ 0, 1 \}^{r_{i}} }, \ket{ C_{1, \cdots, i - 1}(R_{1, \cdots, i - 1}) } \right)
        =
        \ket{ C_{1, \cdots, i}(R_{1, \cdots, i}) }
        \enspace .
        $$
    \end{itemize}

    It can be verified that if all $\ell$ executions of the algorithm $Q_{CI}$ succeed, then by the correctness of the induction, the state we get at the end is the desired normalization of $\ket{ C(R) }$. Finally, all $\ell$ calls together take time complexity $\ell \cdot O\left( \delta \cdot n^{2} \cdot 2^{r} \right) = O\left( \delta \cdot n^{3} \cdot 2^{r} \right)$, and the probability that all calls succeed is $1 - p$ where $p$ is the probability that one of the calls fail. As we saw, for each call, the probability to fail is bounded by $\frac{1}{\delta \cdot \ell}$. By union bound, we have $p \leq \ell \cdot \frac{1}{\delta \cdot \ell} = \frac{1}{\delta}$, which finishes the proof.
\end{proof}

Finally, we conclude with the second main theorem, that says that a quantum computer with computable order interference can be used to solve the sequentially invertible statistical difference problem in quantum polynomial time, whenever the sequential invertibility parameter $r$ is at most logarithmic in the input size.

\begin{theorem} [$\SISD_{\poly, O\left( \log(n) \right)} \in \BQP^{\OI}$] \label{theorem:SISD_in_BQP_OI}
    Let $\SISD_{a, b, r} := \left( \YES, \NO \right)$ an $\SISD$ promise problem, such that $\SISD_{a, b, r} \in \SISD_{\poly, \, O\left( \log(n) \right)}$. Then, there is $M^{\OI}$ a $\QPT$ algorithm with oracle access to a computable uniform order interference oracle $\OOI$ (Definition \ref{definition:computable_uniform_order_interference_oracle}), that for every $x \in \YES$, $M^{\OI}(x) = 1$ with probability at least $p(x)$ and for every $x \in \NO$, $M^{\OI}(x) = 0$ with probability at least $p(x)$, such that
    $$
    \forall x \in \left( \YES \cup \NO \right) : p(x) \geq 1 - 2^{-\poly\left(|x|\right)} \enspace ,
    $$
    for some polynomial $\poly(\cdot)$.
\end{theorem}

\begin{proof}
    Let the following pair of sequences an input for the $\SISD_{a, b, r}$ problem:
    $$
    \left(
    \left( C^{0}_{i, \rightarrow}, C^{0}_{i, \leftarrow} \right)_{i \in [\ell]},
    \left( C^{1}_{i, \rightarrow}, C^{1}_{i, \leftarrow} \right)_{i \in [\ell]}
    \right) \enspace .
    $$
    Denote by $n$ the description size of the pair of sequences, and since $\SISD_{a, b, r} \in \SISD_{\poly, \, O\left( \log(n) \right)}$ it follows that
    \begin{itemize}
        \item
        There exists a polynomial $\poly : \Nat \rightarrow \Nat$ such that for every $n \in \Nat$, $b(n)^{2} - 2a(n) + a(n)^{2} \geq \frac{1}{\poly(n)}$.

        \item 
        There exists a constant $c \in \Nat$ such that for every $n \in \Nat$, $r(n) \leq c \cdot \log\left( 
n \right)$.
    \end{itemize}

    \paragraph{Generating many copies of the output distribution states of the two circuits.}
    Note that for $b \in \{ 0, 1 \}$, if we execute the algorithm $Q^{\OI}\left( \left( C^{b}_{i, \rightarrow}, C^{b}_{i, \leftarrow} \right)_{i \in [\ell]}, 2 \right)$ from Theorem \ref{theorem:quantum_state_generation_of_invertible_circuit_sequence}, the call takes complexity $O\left( 2 \cdot \poly'(n) \cdot 2^{ r } \right)$ and succeeds with probability $\geq 1 - \frac{1}{2} = \frac{1}{2}$. If we execute $3n$ tries, then we have at least a single successful state generation, with probability $\geq 1 - 2^{-3n}$. For each $b \in \{ 0, 1 \}$ we execute $N := 16 \cdot \poly(n)^{2} \cdot 12 \cdot n$ times the previous procedure (where you execute $3n$ tries, and if one of them succeeds, you take the state), and if one of the $2 \cdot N$ attempts fails, the algorithm halts and outputs $\bot$.
    
    If all attempts succeed, then for each $b \in \{ 0, 1 \}$, we now have $N$ copies of $\ket{ C^{b}\left( R \right) }$. Note that all attempts together take complexity $N \cdot O\left( 2 \cdot \poly'(n) \cdot 2^{ r } \right)$ which is $\leq \Tilde{\poly}\left( n \right)$ for some polynomial $\Tilde{\poly}$ (this follows because $r \leq c \cdot \log(n)$ for a constant $c \in \Nat$, and thus $2^{r} \leq n^{c}$ is bounded by a polynomial in $n$ as well). Also, by union bound, all attempts together succeed with probability $\geq 1 - \frac{ 2 \cdot N }{2^{3n}} \geq 1 - \frac{1}{2^{2n}}$.
    This step of the algorithm is the only one that uses computable order interference and the OI oracle. The next steps use only standard quantum computation.

    \paragraph{Repeated swap tests and decision.}
    Let us list the $N$ copies of each of the quantum states as
    $$
    \ket{ C^{0}\left( R \right) }_{Reg_{0, 1}}, \ket{ C^{0}\left( R \right) }_{Reg_{0, 2}}, \cdots , \ket{ C^{0}\left( R \right) }_{Reg_{ 0, N }}
    \enspace ,
    $$
    $$
    \ket{ C^{1}\left( R \right) }_{Reg_{1, 1}}, \ket{ C^{1}\left( R \right) }_{Reg_{1, 2}}, \cdots , \ket{ C^{1}\left( R \right) }_{Reg_{ 1, N }}
    \enspace .
    $$
    For each $j \in [ N ]$, we now execute a swap test (the algorithm $Q_{ST}$ from Theorem \ref{theorem:swap_test}) on the pair of $k$-qubit states $\ket{ C^{0}\left( R \right) }_{Reg_{0, j}}$, $\ket{ C^{1}\left( R \right) }_{Reg_{1, j}}$. For every execution of the swap test the complexity is $O(k) \leq O(n)$, and over all $N$ executions the complexity is polynomial in $n$.
    
    Let $\{ b_{i} \}_{i \in [ N ]}$ the $N$ results of the swap test executions. We take an average of the results, that is
    $$
    A
    :=
    \frac{ \sum_{i \in [N]} b_{i} }{ N }
    \enspace .
    $$
    The algorithm decides where the input belongs by a threshold: We take $T := 1 - \frac{ b(n)^{2} + 2\cdot a(n) - a(n)^{2} }{ 4 }$, and if $A > T$ we output $1$, and otherwise (i.e., $A \leq T$), we output $0$.
    
    \paragraph{Final analysis.}
    The first thing to verify is that the total execution time of the algorithm is polynomial in $n$, as both steps are polynomial in $n$. We next explain why the probability to output the correct answer (i.e., output $1$ when the input is in $\YES$ and output $0$ when the input is in $\NO$) is exponentially (in the input size $n$) close to $1$. 

    By the correctness of the swap test, for every $i \in [N]$ of the swap test executions, the probability for $b_{i} = 1$ is $p := \frac{1 + |\bra{C^{0}\left( R \right)}\ket{C^{1}\left( R \right)}|^{2}}{2}$. Observe that the inner product $\bra{C^{0}\left( R \right)}\ket{C^{1}\left( R \right)}$ equals the fidelity between the output distributions of $C^{0}$, $C^{1}$, thus $|\bra{C^{0}\left( R \right)}\ket{C^{1}\left( R \right)}|^{2} = F\left( C^{0}\left( R \right), C^{1}\left( R \right) \right)^{2}$. Recall the following known relation between total variation distance and fidelity of classical distributions:
    $$
    1 - F\left( D_{0}, D_{1} \right) \leq
    \norm{D_{0} - D_{1}}_{TV}
    \leq
    \sqrt{ 1 - F\left( D_{0}, D_{1} \right)^{2} }
    \enspace .
    $$
    We have the following conclusions.
    \begin{itemize}
        \item
        If the input pair of sequences is in $\YES$, then $\norm{ C^{0}\left( R \right) - C^{1}\left( R \right) }_{TV} \leq a(n) \leq 1$ and,
        $$
        p
        =
        \frac{ 1 + F\left( C^{0}\left( R \right), C^{1}\left( R \right) \right)^{2} }{ 2 }
        \geq
        \frac{ 1 + \left( 1 - \norm{ C^{0}\left( R \right) - C^{1}\left( R \right) }_{TV} \right)^{2} }{ 2 }
        \geq
        \frac{ 1 + \left( 1 - a(n) \right)^{2} }{ 2 }
        $$
        $$
        =
        \frac{ 1 + 1 - 2a(n) + a(n)^{2} }{ 2 }
        =
        1 - \frac{ 2a(n) - a(n)^{2} }{ 2 } \enspace .
        $$

        \item
        If the input pair of sequences is in $\NO$, then $\norm{ C^{0}\left( R \right), C^{1}\left( R \right) }_{TV} > b(n) \geq 0$ and,
        $$
        p
        =
        \frac{ 1 + F\left( C^{0}\left( R \right), C^{1}\left( R \right) \right)^{2} }{ 2 }
        \leq
        \frac{ 1 + 1 - \norm{ C^{0}\left( R \right) - C^{1}\left( R \right) }_{TV}^{2} }{ 2 }
        <
        \frac{ 1 + 1 - b(n)^{2} }{ 2 }
        $$
        $$
        =
        1 - \frac{ b(n)^{2} }{ 2 } \enspace .
        $$
    \end{itemize}

    Next, for each of the $N$ executions of the swap test, note we have a binary random variable $b_{i}$ that is $1$ with probability $p = \frac{1 + |\bra{C^{0}\left( R \right)}\ket{C^{1}\left( R \right)}|^{2}}{2}$. It follows that the expectation of the average of executions of the swap test (that is, the expectation of $A$) is the same, $p$. By Chernoff's bound \ref{theorem:chernoff_bound} it follows that for every $\varepsilon \in (0, 1)$,
    $$
    \Pr_{b_{1}, \cdots, b_{N}}\left[ \big| A - p \big| \geq \varepsilon \cdot p  \right]
    \leq
    2 \cdot e^{ - N \cdot \frac{p \cdot \varepsilon^{2}}{ 3 } }
    \enspace .
    $$
    Note that the threshold $T := 1 - \frac{ b(n)^{2} + 2\cdot a(n) - a(n)^{2} }{ 4 }$ is exactly the average between two numbers: One is the lower bound $B_{\YES} := 1 - \frac{ 2a(n) - a(n)^{2} }{ 2 }$ on $p$ in case the input is in $\YES$, and the second number is the upper bound $B_{\NO} := 1 - \frac{ b(n)^{2} }{ 2 }$ on $p$ in case the input is in $\NO$. This means that $T$ sits on the real number line exactly in the middle between $B_{\YES}$ and $B_{\NO}$. In order to complete the proof, it will be sufficient to (1) understand a lower bound $L$ on the distance $|B_{\YES} - B_{\NO}|$ (half of this lower bound acts as a lower bound for $|B_{\YES} - T|$ and $|B_{\NO} - T|$), and (2) show that the probability for $A$ to deviate from the average $p$ by more than $L/2$ is exponentially small, in each of the cases.

    The first part is easy: We know
    $$
    |B_{\YES} - B_{\NO}|
    =
    1 - \frac{ 2a(n) - a(n)^{2} }{ 2 } - 1 + \frac{ b(n)^{2} }{ 2 }
    =
    \frac{ b(n)^{2} - 2a(n) + a(n)^{2} }{ 2 }
    \geq
    \frac{1}{ 2 \cdot \poly(n) }
    \enspace .
    $$
    This means that by taking half of the above, which is $\frac{1}{ 4 \cdot \poly(n) }$, we have a lower bound from both $|B_{\YES} - T|$ and $|B_{\NO} - T|$.

    Finally, by the above stated instance of the Chernoff bound, for $\varepsilon := \frac{1}{ 4 \cdot \poly(n) }$, the probability for $A$ to deviate from $p$ with an amount $\varepsilon \cdot p \leq \varepsilon = \frac{1}{ 4 \cdot \poly(n) }$, is bounded by 
    $$
    2 \cdot e^{ - N \cdot \frac{ p \cdot \varepsilon^{2} }{ 3 } }
    \underset{ (p \geq \frac{1}{2}) }{ \leq }
    2 \cdot e^{ - N \cdot \frac{\varepsilon^{2}}{ 6 } }
    =
    2 \cdot e^{ - \frac{ N }{ 6\cdot 16 \cdot \poly\left( n \right)^{2} } }
    $$
    $$
    =
    2 \cdot e^{ - \frac{ 16 \cdot \poly(n)^{2} \cdot 12 \cdot n }{ 6\cdot 16 \cdot \poly\left( n \right)^{2} } }
    =
    2 \cdot e^{ -2n } \enspace .
    $$

    To conclude, recall that the probability for the first part of the algorithm to succeed is at least $1 - 2^{-2n}$, and the probability for the second part of the algorithm to be correct, conditioned on the first part being successful, is at least $1 - 2\cdot e^{-2n}$. Overall, if both parts succeed, which happens with probability at least $\left( 1 - 2^{-2n} \right)\cdot \left( 1 - 2\cdot e^{-2n} \right) \geq 1 - 2^{-n} = 1 - 2^{-|x|}$, the algorithm outputs the correct answer for the input, as desired.
\end{proof}

\section{$\GI \leq_{p} \SISD_{\poly, O(\log(n))}$} \label{section:gi_reduction}
In this section we prove our main theorem regarding the computational problem of Graph Isomorphism $\GI$. 

\begin{theorem} [Reduction from $\GI$ to $\SISD_{\frac{n}{2^{n}}, \; 1 - \frac{n}{2^{n}}, \; \log\left( n \right)}$]
    There exists a classical deterministic polynomial time Turing machine $M$ that computes the reduction,
    $$
    \GI \leq_{p} \SISD_{a(n), b(n), r(n)} \enspace ,
    $$
    for $a(n) := \frac{n}{2^{n}}$, $b(n) := 1 - \frac{n}{2^{n}}$, $r(n) := \log\left( n \right)$.
\end{theorem}

Note that for the values of $a(n)$, $b(n)$ above we have $b(n)^{2} - 2a(n) + a(n)^{2}$ is at least inverse polynomially large, and also $r(n) = \log\left( n \right) = O\left( \log(n) \right)$, and thus $\SISD_{a(n), b(n), r(n)} \in \SISD_{ \poly, O\left( \log(n) \right) }$. We proceed to the proof of the theorem.

\begin{proof}
    Let $\left( G_{0}, G_{1} \right)$ an input for the $\GI$ problem, and let $n \in \Nat$ denote the number of vertices in each of the graphs, and assume that the input graphs are represented by adjacency matrices. So, the description size of the input is
    $$
    |\left( G_{0}, G_{1} \right)| = 2 \cdot n^{2} \enspace .
    $$
    The reduction executes in classical deterministic polynomial time and will output a pair $\left( C^{0}, C^{1} \right)$ of invertible circuit sequences (as in Definition \ref{definition:invertible_circuit_sequence}). Let $\ell := \left( n - 1 \right) \cdot n$.

    \paragraph{Definition of the invertible circuit sequence $C$ and then the pair of sequences $\left( C^{0}, C^{1} \right)$.}
    We define an invertible circuit sequence $C := \left( C_{i, \rightarrow}, C_{i, \leftarrow} \right)_{i \in [\ell]}$ below, where all $\ell$ circuit pairs have input/output size $n^{2}$ (which we think of as an adjacency matrix for a graph with $n$ vertices) and randomness of exactly $\log\left( n \right)$ bits (formally, we use $\lceil \log\left( n \right) \rceil$ bits). We will later define the circuits $C^{0}$, $C^{1}$, as a function of the circuit sequence $C$.
    \begin{enumerate}
        \item 
        For $i = n, n - 1, \cdots, 3, 2$ (that is, starting from $i = n$ and moving down to $i = 2$), we define,
        \begin{enumerate}
            \item
            For $j \in [n]$, the circuit $C_{ \left( i, j \right), \rightarrow }$ acts as follows. The circuit's main input is a graph $G = \left( [n], E_{G} \right)$ (in the data structure of adjacency matrix of $n^{2}$ bits) and auxiliary input is the randomness $z_{i, j} \in \{ 0, 1 \}^{\lceil \log\left( n \right) \rceil}$.
            \begin{itemize}
                \item
                Define $d_{i} := \lceil \log\left( i \right) \rceil$, and the circuit computes only on the first $d_{i}$ bits of randomness in $z_{i, j}$, and ignores the rest. Denote by $z'_{i, j} \in \{ 0, 1 \}^{d_{i}}$ the first $d_{i}$ digits of the randomness.

                \item 
                The circuit interprets the string $z'_{i, j} \in \{ 0, 1 \}^{d_{i}}$ as a uniformly random number in $[ 2^{d_{i}} ]$ (by adding $1$ to the number derived from the binary representation of the string), denote this number by $k_{i, j} \in [2^{d_{i}}]$. In case $k_{i, j} > i$ the circuit does nothing and its function is the identity.
                
                \item
                In case $k_{i, j} \leq i$ the fixed-randomness circuit $C_{ \left( i, j \right), \rightarrow }\left( \cdot ; z_{i, j} \right)$ computes the graph $\sigma_{\left( i, k_{i, j} \right)}\left( G \right)$ as its output, where $\sigma_{\left( i, k_{i, j} \right)} \in S_{i}$ is a swap between elements $i$ and $k_{i, j}$, that is, it is a permutation on the elements $[i]$, that acts trivially on all elements that are in $[i] \setminus \{ i, k_{i, j} \}$, and swaps between $i$ and $k_{i, j}$.
            \end{itemize}
        \end{enumerate}
    \end{enumerate}

    We defined what are the forward computations, and the definition of the inverse circuits follows: For every $i \in \{ n, n - 1, \cdots, 3, 2 \}$, $j \in [n]$, the circuit $C_{\left( i, j \right), \leftarrow }$ is exactly identical to the circuit $C_{\left( i, j \right), \rightarrow }$. It can be verified by the reader that indeed for every $i \in \{ n, n - 1, \cdots, 3, 2 \}$, $j \in [n]$, the circuits $C_{\left( i, j \right), \rightarrow }$, $C_{\left( i, j \right), \leftarrow }$ are inverses of each other for every fixed randomness $z_{i, j} \in \{ 0, 1 \}^{ \ceil{ \log(n) } }$, because making the same swap twice leads to the identity permutation.
    
    Regarding the output of the reduction, which should be a pair of invertible circuits sequences $\left( C^{0}, C^{1} \right)$: For $b \in \{ 0, 1 \}$, the sequence $C^{b}$ is the same as the above sequence $C$, with one adding: In the beginning, there is a deterministic (randomness-free) circuit $C^{b}_{0, \rightarrow}$ that adds the description of the graph $G_{b}$, by simply adding the bits (modulo $2$, also known as XOR-ing) of the adjacency matrix of $G_{b}$. Accordingly, the inverse $C^{b}_{0, \leftarrow}$ of that first circuit $C^{b}_{0, \rightarrow}$, is subtracting the description of $G_{b}$ (which makes the circuits equal, because subtracting and adding modulo $2$ is identical).
    To conclude, it can be easily verified that both $C^{0}$ and $C^{1}$ are legal $\left( \log(n) \right)$-invertible circuit sequences.

    \paragraph{Properties of the output distributions of the circuits $C^{0}$, $C^{1}$.}
    For $b \in \{ 0, 1 \}$, one can observe that the output distribution of $C^{b}$ (when starting from a string of $n^2$ zeros) is as follows.
    \begin{itemize}
        \item
        For every $i \in \{ n, n - 1, \cdots, 3, 2 \}$, if for all $j \in [n]$, the randomness $z_{i, j} \in \{ 0, 1 \}^{\log(n)}$ was such that the number $k_{i, j}$ (which is derived from $z'_{i, j} \in \{ 0, 1 \}^{d_{i}}$, the first $d_{i} := \lceil \log\left( i \right) \rceil$ bits of $z_{i, j}$) is $> i$, then the graph stays the same after passing the entire outer iteration $i$. However, note that $2^{d_{i}} \leq 2 \cdot i$, which means that a uniformly random $k_{i, j} \in [2^{d_{i}}]$ is bounded by $i$ with probability at least $\frac{1}{2}$. Thus, the probability that for all $j \in [n]$, the sample $k_{i, j} > i$, is bounded by $2^{-n}$. By union bound, the probability that there exists an index $i \in \{ n, n - 1, \cdots, 3, 2 \}$ such that for every $j \in [n]$, $k_{i, j}$ caused no change, is bounded by $\frac{n}{2^{n}}$.

        \item 
        In the other case, that happens with probability $\geq 1 - \frac{n}{2^{n}}$, for every $i \in \{ n, n - 1, \cdots, 3, 2 \}$, there is at least a single try $j \in [n]$ such that $k_{i, j} \leq i$. Observe that in this second case the output distribution of $C^{b}$ is as follows. For an iteratively decreasing $i = n, n - 1, \cdots, 3, 2$, apply $\sigma_{\left( i, r_{i} \right)}$ to the graph, for a uniformly random $r_{i} \in [i]$. The above distribution is exactly the description of the Fisher-Yates algorithm, and thus by Fact \ref{fact:fisher_yates_algorithm}, generates a uniformly random permutation $\sigma \in S_{n}$, and applies it to the input graph $G_{b}$.
    \end{itemize}
    The above means that the output distribution of $C^{b}$ is statistically close (has total variation distance bounded by $\frac{n}{2^{n}}$) to $\sigma\left( G_{b} \right)$, for a uniformly random permutation $\sigma \in S_{n}$.

\paragraph{Soundness and completeness of the reduction.}
The last step of our proof is to show that the reduction is correct, that is, that when $\left( G_{0}, G_{1} \right) \in \YES$ then the output distributions of the circuits, denoted $\mathcal{D}\left( C^{0} \right)$, $\mathcal{D}\left( C^{1} \right)$ respectively, are statistically close, and when the input is in $\NO$, then the output distributions are statistically far.

Both soundness and completeness follow rather easily at this point. For soundness, in case $\left( G_{0}, G_{1} \right) \in \NO$, it means the graphs are non-isomorphic, which, by Fact \ref{fact:random_permutations_of_isomorphic_graphs}, means that the supports of the distributions $\sigma\left( G_{0} \right)$, $\sigma\left( G_{1} \right)$ are disjoint, and thus $\norm{ C^{0} - C^{1} }_{TV} \geq 1 - \frac{n}{2^{n}}$.
For completeness, in case $\left( G_{0}, G_{1} \right) \in \YES$, it means the graphs are isomorphic, which again by Fact \ref{fact:random_permutations_of_isomorphic_graphs}, means that the distributions $\sigma\left( G_{0} \right)$, $\sigma\left( G_{1} \right)$ are identical, and thus $\norm{ C^{0} - C^{1} }_{TV} \leq \frac{n}{2^{n}}$.
\end{proof}

\section{$\GCVP_{O\left( n\sqrt{n} \right)} \leq_{p} \SISD_{\poly, O(\log(n))}$} \label{section:gcvp_reduction}
In this section we prove our main theorem regarding $\GCVP$. 

\begin{theorem} [Reduction from $\GCVP_{O\left( n\sqrt{n} \right)}$ to $\SISD_{\left( \frac{1}{4} + O\left( n^{ -\frac{1}{4} } \right) , \; \frac{3}{4} - e^{ -\Omega\left( n \right) }, \; 1 \right)}$]
    There exists a positive absolute constant $c \in \bbR_{> 0}$ and a classical deterministic polynomial time Turing machine $M$ that computes the reduction,
    $$
    \GCVP_{c \cdot n\sqrt{n} } \leq_{p} \SISD_{\left( a(n), b(n), r(n) \right)} \enspace ,
    $$
    for $a(n) := \frac{1}{4} + O\left( n^{ -\frac{1}{4} } \right)$, $b(n) := \frac{3}{4} - e^{ -\Omega\left( n \right) }$, $r(n) := 1$.
\end{theorem}

Note that for the values of $a(n)$, $b(n)$ above we have $b(n)^{2} - 2a(n) + a(n)^{2}$ is at least inverse polynomially large, and also $r(n) = 1 = O\left( \log(n) \right)$, and thus $\SISD_{a(n), b(n), r(n)} \in \SISD_{ \poly, O\left( \log(n) \right) }$. We proceed to the proof of the theorem.

\begin{proof}
    Let $\left( \Basis, \tVector, d \right)$ an input for the $\GCVP_{g(n)}$ problem, for $g(n) = c_{g} \cdot n\sqrt{n}$ where $c_{g} := 64 \cdot c_{0}$, where $c_{0} \in \bbR_{> 0}$ is a positive absolute constant we will specify later, $n \in \Nat$ is the dimension of the lattice basis $\Basis$ and the vector $\tVector$, and $b \in \Nat$ is the size of the binary representation of the largest (in absolute value) number in $\Basis$ and the vector $\tVector$. So, the description size of the input is
    $$
    |\left( \Basis, \tVector, d \right)| = n^{2}\cdot b + n \cdot b + \log_{2}(d) \enspace .
    $$
    The reduction executes in classical deterministic polynomial time and will output a pair $\left( C^{0}, C^{1} \right)$ of $1$-invertible circuit sequences (as in Definition \ref{definition:invertible_circuit_sequence}). Let $\ell := \ell_{\Lattice} + 1 + \ell_{G} + 1$, where $\ell_{\Lattice} := \left( m + 1 \right) \cdot n$ and $\ell_{G} := \kappa \cdot \left( \beta + 1 \right) \cdot n$, such that,
    \begin{itemize}
        \item
        $m := \lceil \log_{2}\left( M \right) \rceil$ for $M := 2^n \cdot \left( \sum_{i \in [n]} \norm{ \bVector_{i} } + \norm{ \tVector } + d \right)$,

        \item 
        $\kappa := c_{\kappa} \cdot n^{2}$ for $c_{\kappa} := c_{0}^{2} \cdot 64$ (as in the case of $c_{g}$, the constant $c_{0}$ will be specified later).
        
        \item 
        $\beta := \lfloor \log_{2}\left( B \right) \rfloor$, for $B := \frac{ g(n) \cdot d }{ c_{\mathcal{N}} \cdot \sqrt{n} \cdot \sqrt{\kappa} }$, where $c_{\mathcal{N}} := 8$. For the algorithm to work we will need $\beta \geq 0$, which in turn happens if $B \geq 1$, and indeed:
        $$
        g(n)
        :=
        64 \cdot c_{0} \cdot n \sqrt{n}
        :=
        8 \cdot \sqrt{\kappa} \sqrt{n}
        :=
        c_{\mathcal{N}} \cdot \sqrt{\kappa} \sqrt{n}
        \enspace ,
        $$
        which, since $d \in \Nat$, implies $B \geq 1$.
    \end{itemize}
    
    Let $s := \lceil b + m + \log\left( n \cdot g(n) \cdot d \right) \rceil$ and note that a binary string of size $s + 1$ can be used to represent any integer of absolute value bounded by
    $$
    2^{s} - 1 \geq 2^{b} \cdot M \cdot n \cdot g(n) \cdot d - 1 \enspace ,
    $$
    by using the first bit as a decider for the sign of the number, and the rest of the $s$ bits as the binary representation of the absolute value of the number (the only number that has two such binary representations, and not one representation, is $0$). In the following proof we thus use binary strings of size $s + 1$ to represent numbers in the set $\bbZ_{ \left( -\left( 2^{s} - 1\right), 2^{s} - 1 \right) }$.

    \paragraph{Definition of the circuit sequence $C$ and an pair $\left( C^{0}, C^{1} \right)$.}
    We define an invertible circuit sequence $C := \left( C_{i, \rightarrow}, C_{i, \leftarrow} \right)_{i \in [\ell]}$ below, where all $\ell$ circuit pairs have input/output size $n \cdot \left( s + 1 \right)$ (which we think of as an integer vector $\vVector \in \bbZ^{n}_{ \left( -\left( 2^{s} - 1\right), 2^{s} - 1 \right) }$) and randomness bounded by $1$ (that is, each circuit either uses a single bit of randomness, or it is deterministic and does not use randomness at all). In the below circuits, all arithmetic (which comes down to additions only) is modulo in absolute value of $2^{s} - 1$, that is, once it gets out of bounds of one side (right, $2^{s} - 1$ or left, $-\left( 2^{s} - 1\right)$), it pops up, continuing in the same direction, at the opposite side (left or right, respectively).
    \begin{enumerate}
        \item \label{gcvp_reduction_step_1}
        \textit{First $\ell_{\Lattice}$ circuits: Constructing a random lattice vector with a uniform coordinates vector in $\bbZ^{n}_{ \left( 0, 2^{m + 1} - 1 \right) }$.}
        For $i \in [m + 1]$, $j \in [n]$, the circuit $C_{ \left( 1, i, j \right), \rightarrow }$, given main input $\vVector \in \bbZ^{n}_{ \left( -\left( 2^{s} - 1\right), 2^{s} - 1 \right) }$ and auxiliary input randomness $z_{1, i, j} \in \{ 0, 1 \}$, adds the vector $z_{1, i, j} \cdot 2^{i - 1} \cdot \Basis_{j}$, where $\Basis_{j}$ is the $j$-th column of the basis $\Basis$.

        \item \label{gcvp_reduction_step_2}
        \textit{Circuit $\ell_{\Lattice} + 1$: Stabilizing the coordinate vectors to $\bbZ^{n}_{ \left( -2^{m}, 2^{m} - 1 \right) }$, by subtraction.}
        The circuit $C_{ \left( 2, 1 \right), \rightarrow }$, which is deterministic and does not use randomness, given input $\vVector \in \bbZ^{n}_{ \left( -\left( 2^{s} - 1\right), 2^{s} - 1 \right) }$, outputs $\vVector - 2^{m} \cdot \sum_{j \in [n]} \Basis_{j}$. This step uses a single deterministic circuit.

        \item \label{gcvp_reduction_step_3}
        \textit{Circuits $\ell_{\Lattice} + 1 + i$ for $i \in [\ell_{G}]$: Adding approximated Gaussian noise.}
        For $i_{1} \in [\kappa], i_{2} \in [\beta + 1], i_{3} \in [n]$, the circuit $C_{ \left( 3, i_{1}, i_{2}, i_{3} \right), \rightarrow }$, given input $\vVector \in \bbZ^{n}_{ \left( -\left( 2^{s} - 1\right), 2^{s} - 1 \right) }$ and randomness $z_{\left( 3, i_{1}, i_{2}, i_{3} \right)} \in \{ 0, 1 \}$, adds the vector
        $$
        z_{\left( 3, i_{1}, i_{2}, i_{3} \right)}
        \cdot
        2^{i_{2} - 1}
        \cdot
        e_{i_{3}}
        \enspace ,
        $$
        where for $j \in [n]$, the vector $e_{j}$ is the $j$-th standard basis vector, i.e., $e_{j} := \left( 0_{1}, \cdots, 0_{j - 1}, 1_{j}, 0_{j + 1}, \cdots, 0_{n} \right)^{T}$.

        \item \label{gcvp_reduction_step_4}
        \textit{Circuit $\ell_{\Lattice} + 1 + \ell_{G} + 1$: Stabilizing the Gaussian noise by subtraction.}
        The circuit $C_{ \left( 4, 1 \right), \rightarrow }$, which is deterministic and does not use randomness, given input $\vVector \in \bbZ^{n}_{ \left( -\left( 2^{s} - 1\right), 2^{s} - 1 \right) }$, outputs
        $$
        \vVector
        -
        \left(
        \frac{\kappa}{2}\cdot
        \left( 2^{\beta + 1} - 1 \right)
        \left( 1_{1}, 1_{2}, \cdots, 1_{n} \right)^{T}
        \right)
        \enspace .
        $$
    \end{enumerate}
    This step uses a single deterministic circuit. Note that $\kappa$ is even and $\kappa/2$ is an integer.

    We defined what are the forward computations (that is, for every $i \in [\ell]$, the circuit $C_{i, \rightarrow}$), and the definition of the circuits $C_{i, \leftarrow}$ follows: 
    \begin{enumerate}
        \item 
        \textit{Inverses of first $\ell_{\Lattice}$ circuits.}
        For $i \in [m + 1]$, $j \in [n]$, the circuit $C_{ \left( 1, i, j \right), \leftarrow }$, given main input $\vVector$ and auxiliary input randomness $z_{1, i, j} \in \{ 0, 1 \}$, subtracts the vector $z_{1, i, j} \cdot 2^{i - 1} \cdot \Basis_{j}$ instead of adding it.

        \item 
        \textit{Inverse of circuit $\ell_{\Lattice} + 1$.}
        The circuit $C_{ \left( 2, 1 \right), \leftarrow }$, given input $\vVector$, adds the vector $2^{m} \cdot \sum_{j \in [n]} \Basis_{j}$ instead of subtracting it.

        \item 
        \textit{Inverses of the circuits $\ell_{\Lattice} + 1 + i$, for $i \in [\ell_{G}]$.}
        For $i_{1} \in [\kappa], i_{2} \in [\beta + 1], i_{3} \in [n]$, the circuit $C_{ \left( 3, i_{1}, i_{2}, i_{3} \right), \leftarrow }$, given main input $\vVector$ and auxiliary input randomness $z_{\left( 3, i_{1}, i_{2}, i_{3} \right)} \in \{ 0, 1 \}$, subtracts the vector $z_{\left( 3, i_{1}, i_{2}, i_{3} \right)} \cdot 2^{i_{2} - 1} \cdot e_{i_{3}}$, instead of adding it.        

        \item 
        \textit{Inverse of circuit $\ell_{\Lattice} + 1 + \ell_{G} + 1$.}
        The circuit $C_{ \left( 4, 1 \right), \leftarrow }$, given input $\vVector$, adds the vector
        $$
        \left(
        \frac{\kappa}{2}\cdot
        \left( 2^{\beta + 1} - 1 \right)
        \left( 1_{1}, 1_{2}, \cdots, 1_{n} \right)^{T}
        \right)
        \enspace ,
        $$
        instead of subtracting it.
    \end{enumerate}

    Regarding the output of the reduction, which is a pair of invertible circuits sequences $\left( C^{0}, C^{1} \right)$: The first circuit sequence $C^{0}$ is exactly the circuit sequence $C$. As for the second circuit sequence $C^{1}$, it is the same as $C = C^{0}$ only that we add the vector $\tVector$ in the end, and formally, we make the following single change to $C = C^{0}$: The last circuit pair $C_{\left( 4, 1 \right), \rightarrow}$, $C_{\left( 4, 1 \right), \leftarrow}$ is changed as follows to the pair $C'_{\left( 4, 1 \right), \rightarrow}$, $C'_{\left( 4, 1 \right), \leftarrow}$.
    \begin{itemize}
        \item
        $C'_{\left( 4, 1 \right), \rightarrow}$ is the same as $C_{\left( 4, 1 \right), \rightarrow}$, only that it further adds $\tVector$ in the end of its computation.

        \item 
        The inverse $C'_{\left( 4, 1 \right), \leftarrow}$ acts the same as $C_{\left( 4, 1 \right), \leftarrow}$, but further subtracts $\tVector$ in the end of its computation.
    \end{itemize}
    It can be easily verified that both $C^{0}$ and $C^{1}$ are legal $1$-invertible circuit sequences.

    \paragraph{Properties of the output distributions of the circuits $C^{0}$, $C^{1}$.}
    We next observe the following properties of the output distribution of the circuit sequence $C$, and then for the circuits $C^{0}$, $C^{1}$.
    \begin{itemize}
        \item
        Note that in step \ref{gcvp_reduction_step_1} of the circuit $C$, for every $j \in [n]$, it samples a uniformly random number $M_{j} \in \{ 0, 1, \cdots, 2^{m + 1} - 1 \}$, and adds $M_{j} \cdot \Basis_{j}$ to the input. This means that at the end of step \ref{gcvp_reduction_step_1}, we have the sum $\sum_{j \in [n]} M_{j} \cdot \Basis_{j}$.

        \item 
        The next step \ref{gcvp_reduction_step_2} of the algorithm is intended to move the distribution of the coordinates vector $\left( M_1, \cdots, M_{n} \right)$ be uniform with expectation $\approx 0$, and to this end the circuit $C_{\ell_{\Lattice} + 1, \rightarrow}$ subtracts $2^{m}\cdot \sum_{j \in [n]} \Basis_{j}$. At the end of this step the output distribution is accordingly $\sum_{j \in [n]} \left( M_{j} - 2^{m} \right)\cdot \Basis_{j}$, where for every $j \in [n]$, the variables $M_{j} - 2^{m}$ are i.i.d. samples from the uniform distribution over $\{ -2^{m}, \cdots, -1, 0, 1, \cdots, 2^{m} - 1 \}$. In other words, at the end of step \ref{gcvp_reduction_step_2} we have a random lattice vector, over a uniformly random coordinates vector in the restricted set $\bbZ^{n}_{\left( -2^{m}, 2^{m} - 1 \right)}$.

        \item 
        The next step \ref{gcvp_reduction_step_3} of the algorithm is intended to add positive approximated Gaussian noise. Observe that at the end of this step, for each $j \in [n]$ we add $\kappa$ uniformly random i.i.d. samples $\Bar{\beta}_{j, i_{1}} \in \{ 0, 1, \cdots, 2^{\beta + 1} - 1 \}$ (for $i_{1} \in [\kappa]$) to the distribution from before.

        \item 
        The next step is intended to stabilize the Gaussian noise to center around $0$ (in expectation). To this end, step \ref{gcvp_reduction_step_4} subtracts $\frac{\kappa}{2}\left( 2^{\beta + 1} - 1 \right)$ ($\kappa/2$ is an integer) from each of the $n$ coordinates of $\vVector$. Since this is exactly the expectation of the random variable we add to each coordinate, the average of Gaussian noise becomes zero. 
    \end{itemize}
    
    To conclude what we saw so far,
    \begin{itemize}
        \item
        The output distribution of $C^{0}$ is the variable $\left( \vVector + \eVector \right) \in \bbZ^{n}_{\left( -\left( 2^{m + b}\cdot n + \kappa \cdot 2^{\beta} \right), 2^{m + b}\cdot n + \kappa \cdot 2^{\beta} \right)}$, where $\vVector \in \bbZ^{n}_{\left( -\left( 2^{m + b}\cdot n \right), 2^{m + b}\cdot n \right)}$ is a random lattice vector with a uniformly random coordinates vector $\aVector \in \bbZ^{n}_{\left( -2^{m}, 2^{m} - 1 \right)}$, and $\eVector \in \bbZ^{n}_{\left( -\kappa\cdot 2^{\beta}, \kappa\cdot 2^{\beta} \right)}$ such that for each $j \in [n]$, the coordinate $\eVector_{j}$ is a random variable defined as follows: We sample $\kappa$ i.i.d. samples from the distribution $\mathcal{U}_{\left( 0, 2^{\beta + 1} - 1 \right)}$ (the discrete uniform distribution on the set $\{ 0, 1, \cdots, 2^{\beta + 1} - 1 \}$), add all $\kappa$ samples together, and then subtract $\frac{\kappa}{2}\cdot \left( 2^{\beta + 1} - 1 \right)$, which is exactly the expectation of the sum, and is also an integer since $\kappa$ is an even natural number.

        \item 
        The output distribution of $C^{1}$ is the same as that of $C^{0}$ only that the vector $\tVector$ is added in the end, that is, the output distribution of $C^{1}$ is $\vVector + \eVector + \tVector$.
    \end{itemize}

    \noindent
    \textbf{To conclude} the part where we describe the reduction and its properties, we have the equality
    $$
    \norm{ \mathcal{D}\left( C^{0} \right) - \mathcal{D}\left( C^{1} \right) }_{TV}
    $$
    $$
    :=
    \norm{ \mathcal{D}\left( \vVector + \eVector \right) - \mathcal{D}\left( \vVector + \eVector + \tVector \right) }_{TV}
    \enspace .
    $$

    Next, let $\Tilde{B} := \ceil{ \sqrt{ \kappa } \cdot \Tilde{\sigma} }$, $\Tilde{\sigma} := \sqrt{ \frac{ 2^{2\beta + 2} + 2^{\beta + 2} }{ 12 } }$, and let $D^{n}_{ \Tilde{B} }$ the $n$-dimensional discrete truncated Gaussian distribution with standard deviation $\Tilde{B}$, as in Definition \ref{definition:discrete_truncated_gaussian}. By Fact \ref{fact:tv_distance_switched_rv}, the above total variation distance equals
    \begin{equation} \label{equation:tv_distance_switched_rv}
    =
    \norm{ \mathcal{D}\left( \vVector + D^{n}_{ \Tilde{B} } \right) - \mathcal{D}\left( \vVector + D^{n}_{ \Tilde{B} } + \tVector \right) }_{TV}
    +
    2\cdot d_{ \left( \eVector, D^{n}_{ \Tilde{B} } \right) }
    \enspace ,
    \end{equation}
    for some $d_{ \left( \eVector, D^{n}_{ \Tilde{B} } \right) } \in \bbR$ inside the range
    $$
    \bigg[
    -\norm{ \mathcal{D}\left( D^{n}_{ \Tilde{B} } \right) - \mathcal{D}\left( \eVector \right) }_{TV}
    ,
    \norm{ \mathcal{D}\left( D^{n}_{ \Tilde{B} } \right) - \mathcal{D}\left( \eVector \right) }_{TV}
    \bigg]
    \enspace ,
    $$
    We next observe that the complexity of the reduction is polynomial in the input size, and then turn to proving correctness, that is, that when $\left( \Basis, \tVector, d \right) \in \YES$ then the output distributions of the circuits are statistically close (completeness), and when the input is in $\NO$, then the output distributions are statistically far (soundness). To this end, we will first bound $\norm{ \mathcal{D}\left( D^{n}_{ \Tilde{B} } \right) - \mathcal{D}\left( \eVector \right) }_{TV}$ (which in turn will bound the absolute value of $d_{\left( \eVector, D^{n}_{ \Tilde{B} } \right)}$), and then analyze the soundness and completeness separately.

    \paragraph{Computational complexity of the reduction.}
    For each $b' \in \{ 0, 1 \}$, $i \in [\ell]$, $Y \in \{ \rightarrow, \leftarrow \}$, it is easy to verify that the circuit $C^{b'}_{i, Y}$ executes in polynomial time in its input size, which is in turn $n \cdot \left( s + 1 \right)$. The construction of the circuit $C^{b'}_{i, Y}$ given the input $\left( \Basis, \tVector, d \right)$ also takes polynomial time in the input size. Finally, the amount of circuits is $\ell$ and each circuits size is polynomial in $n \cdot \left( s + 1 \right)$, which are both polynomial in the input size. It follows that the reduction is implemented by a classical deterministic polynomial time Turing machine.

    \paragraph{Sequentially invertible approximation of the discrete Gaussian distribution.} 
    Recall the parameter $d_{\left( \eVector, D_{\Tilde{B}}^{n} \right)}$ from Equation \ref{equation:tv_distance_switched_rv}: For both, the proofs of soundness and completeness of the reduction, we need to bound the absolute value of $d_{\left( \eVector, D_{\Tilde{B}}^{n} \right)}$. By the definition of $d_{\left( \eVector, D_{\Tilde{B}}^{n} \right)}$ it will be sufficient to bound $\norm{ \mathcal{D}\left( D^{n}_{ \Tilde{B} } \right) - \mathcal{D}\left( \eVector \right) }_{TV}$. Observe that the distance $\norm{ \mathcal{D}\left( D^{n}_{ \Tilde{B} } \right) - \mathcal{D}\left( \eVector \right) }_{TV}$ is essentially the quality of our (sequentially invertible) approximation of the truncated discrete Gaussian distribution.
    
    Note that each of the distributions $\mathcal{D}\left( \eVector \right)$, $\mathcal{D}\left( D^{n}_{ \Tilde{B} } \right)$ is a list of $n$ i.i.d. samples of some distribution, and thus,
    $$
    \norm{ \mathcal{D}\left( \eVector \right) - \mathcal{D}\left( D^{n}_{ \Tilde{B} } \right) }_{TV}
    \leq
    n \cdot \norm{ \mathcal{D}\left( \eVector_{1} \right) - \mathcal{D}\left( D_{ \Tilde{B} } \right) }_{TV}
    \enspace ,
    $$
    where $\eVector_{1}$ is the distribution of the first coordinate of the vector $\eVector$ (which distributes identically for all of its coordinates $i \in [n]$). By Lemma \ref{lemma:distance_between_discrete_uniform_sum_and_discrete_gaussian}, there exists a positive absolute constant $c_{0} \in \Nat$ such that $\norm{ \mathcal{D}\left( \eVector_{1} \right) - \mathcal{D}\left( D_{ \Tilde{B} } \right) }_{TV} \leq c_{0} \cdot \frac{1}{\sqrt{\kappa}}$. In conclusion, $|d_{\left( \eVector, D_{\Tilde{B}}^{n} \right)}| \leq \frac{c_{0} \cdot n}{\sqrt{\kappa}} \leq \frac{1}{8}$, thus $|2 \cdot d_{\left( \eVector, D_{\Tilde{B}}^{n} \right)}| \leq \frac{1}{4}$.

    \paragraph{Soundness of the reduction.} 
    In case $\left( \Basis, \tVector, d \right) \in \NO$, then $\Delta\left( \Lattice_{\Basis}, \tVector \right) > d\cdot g(n)$, and we should give a lower bound for 
    $$
    \norm{ \mathcal{D}\left( C^{0} \right) - \mathcal{D}\left( C^{1} \right) }_{TV}
    $$
    $$
    =
    \norm{ \mathcal{D}\left( \vVector + D^{n}_{ \Tilde{B} } \right) - \mathcal{D}\left( \vVector + D^{n}_{ \Tilde{B} } + \tVector \right) }_{TV}
    +
    2\cdot d_{ \left( \eVector, D^{n}_{ \Tilde{B} } \right) }
    $$
    $$
    \geq
    \norm{ \mathcal{D}\left( \vVector + D^{n}_{ \Tilde{B} } \right) - \mathcal{D}\left( \vVector + D^{n}_{ \Tilde{B} } + \tVector \right) }_{TV}
    -
    \frac{1}{4}
    \enspace .
    $$
    
    Recall that,
    $$
    \Tilde{\sigma}
    := \sqrt{ \frac{ 2^{2\beta + 2} + 2^{\beta + 2} }{ 12 } }
    < \sqrt{ \frac{ 8 \cdot 2^{2\beta} }{ 12 } }
    < 2^{\beta}
    \underset{(\beta := \lfloor \log_{2}\left( B \right) \rfloor)}{\leq}
    B \enspace .
    $$
    Also,
    $$
    \Tilde{B}
    :=
    \lceil
    \sqrt{\kappa} \cdot \Tilde{\sigma}
    \rceil
    <
    \lceil
    \sqrt{\kappa} \cdot B
    \rceil
    :=
    \lceil
    \sqrt{\kappa} \cdot \frac{g(n) \cdot d}{c_{\mathcal{N}} \cdot \sqrt{n} \cdot \sqrt{\kappa} }
    \rceil
    =
    \lceil
    \frac{ g(n) \cdot d }{ c_{\mathcal{N}} \cdot \sqrt{n} }
    \rceil
    \underset{\left( \frac{ g(n) \cdot d }{ c_{\mathcal{N}} \cdot \sqrt{n} } > 1 \right)}{<}
    \frac{ 2 \cdot g(n) \cdot d }{ c_{\mathcal{N}} \cdot \sqrt{n} } \enspace .
    $$
    Now, Fact \ref{fact:gaussian_vector_norm_concentration} tells us that with probability at least $1 - e^{ -c_{D} \cdot n }$ (for some absolute constant $c_{D} \in \bbR_{> 0}$), the norm of the variable $\uVector \gets D_{ \Tilde{B} }^{n}$ is bounded by $2 \cdot \sqrt{n} \cdot \Tilde{B}$, which in turn is bounded by $\frac{ g(n) \cdot d }{ 2 }$ due to $c_{\mathcal{N}} := 8$. It follows that for a sample $\uVector \gets D_{ \Tilde{B} }^{n}$, with probability $\geq 1 - e^{ -c_{D} \cdot n }$ we have $\norm{\uVector} \leq \frac{ g(n) \cdot d }{ 2 }$.

    Now, consider the output distributions $\mathcal{D}\left( \vVector + \uVector \right)$, $\mathcal{D}\left( \vVector + \uVector + \tVector \right)$, for the case $\norm{\uVector} \leq \frac{ g(n) \cdot d }{ 2 }$.
    For any $\vVector \in \Lattice_{\Basis}$, $\Delta\left( \Lattice_{\Basis}, \vVector + \uVector \right) \leq \frac{g(n) \cdot d}{2}$. Also, since $\Delta\left( \Lattice_{\Basis}, \tVector \right) > d\cdot g(n)$ and $\Delta\left( \Lattice_{\Basis}, \vVector \right) = 0$ then
    $$
    \Delta\left( \Lattice_{\Basis}, \vVector + \uVector + \tVector \right)
    =
    \Delta\left( \Lattice_{\Basis}, \uVector + \tVector \right)
    \geq
    \Delta\left( \Lattice_{\Basis}, \tVector \right) - \norm{ \uVector }
    > d\cdot g(n) - \frac{g(n) \cdot d}{2} = \frac{g(n) \cdot d}{2}
    \enspace .
    $$
    It follows that if $\norm{\uVector} \leq \frac{ g(n) \cdot d }{ 2 }$ then the distributions $\mathcal{D}\left( \vVector + \uVector \right)$, $\mathcal{D}\left( \vVector + \uVector + \tVector \right)$ don't intersect, and since this happens with probability $\geq 1 - e^{ -c_{D} \cdot n }$, then 
    $$
    \norm{ \mathcal{D}\left( \vVector + D^{n}_{ \Tilde{B} } \right) - \mathcal{D}\left( \vVector + D^{n}_{ \Tilde{B} } + \tVector \right) }_{TV}
    \geq
    1 - e^{ -c_{D} \cdot n }
    \enspace ,
    $$
    which finally implies the reduction's soundness, i.e., if $\left( \Basis, \tVector, d \right) \in \NO$ then
    $$
    \norm{ \mathcal{D}\left( C^{0} \right) - \mathcal{D}\left( C^{1} \right) }_{TV}
    \geq
    \frac{3}{4} - e^{ -c_{D} \cdot n }
    \geq
    \frac{3}{4} - e^{ - \Omega \left( n \right) }
    \enspace .
    $$

    \paragraph{Completeness of the reduction, first part (reducing the analysis to analyzing total variation distance of the noise vector $\uVector$).}
    In case $\left( \Basis, \tVector, d \right) \in \YES$, then $\Delta\left( \Lattice_{\Basis}, \tVector \right) \leq d$, and we should give an upper bound for 
    $$
    \norm{ \mathcal{D}\left( C^{0} \right) - \mathcal{D}\left( C^{1} \right) }_{TV}
    =
    \norm{ \mathcal{D}\left( \vVector + D^{n}_{ \Tilde{B} } \right) - \mathcal{D}\left( \vVector + D^{n}_{ \Tilde{B} } + \tVector \right) }_{TV}
    +
    2\cdot d_{ \left( \eVector, D^{n}_{ \Tilde{B} } \right) }
    $$
    $$
    \leq
    \norm{ \mathcal{D}\left( \vVector + D^{n}_{ \Tilde{B} } \right) - \mathcal{D}\left( \vVector + D^{n}_{ \Tilde{B} } + \tVector \right) }_{TV}
    +
    \frac{1}{4}
    \enspace .
    $$

    Define $\uVector \gets D^{n}_{ \Tilde{B} }$ as a sample from the distribution $D^{n}_{ \Tilde{B} }$.
    We would like to use the fact that adding a lattice vector to a uniform distribution over that lattice, does not change the distribution. However, the lattice is infinite and we cannot sample a uniformly random vector over the entire lattice (in finite time). Formally, since $\Delta\left( \Lattice_{\Basis}, \tVector \right) \leq d$, thus there exist $\sVector \in \Lattice_{\Basis}$ and $\eVector_{0} \in \bbZ^{n}$ such that $\tVector = \sVector + \eVector_{0}$ and $\norm{ \eVector_{0} } \leq d$. It follows that,
    $$
    \norm{
    \mathcal{D}\left( \vVector + \uVector \right)
    -
    \mathcal{D}\left( \vVector + \uVector + \tVector \right)
    }_{TV}
    $$
    $$
    \underset{(\text{triangle inequality})}{\leq}
    \norm{
    \mathcal{D}\left( \vVector + \uVector \right)
    -
    \mathcal{D}\left( \vVector - \sVector + \uVector + \tVector \right)
    }_{TV}
    +
    \norm{
    \mathcal{D}\left( \vVector + \uVector + \tVector \right)
    -
    \mathcal{D}\left( \vVector - \sVector + \uVector + \tVector \right)
    }_{TV}
    $$
    $$
    =
    \norm{
    \mathcal{D}\left( \left( \vVector \right) + \uVector \right)
    -
    \mathcal{D}\left( \left( \vVector \right) + \uVector + \eVector_{0} \right)
    }_{TV}
    +
    \norm{
    \mathcal{D}\left( \vVector + \left( \uVector + \tVector \right) \right)
    -
    \mathcal{D}\left( \vVector - \sVector + \left( \uVector + \tVector \right) \right)
    }_{TV}
    $$
    $$
    \underset{(\text{summing more independent variables can only reduce TV distance})}{\leq}
    \norm{
    \mathcal{D}\left( \uVector \right)
    -
    \mathcal{D}\left( \uVector + \eVector_{0} \right)
    }_{TV}
    +
    \norm{
    \mathcal{D}\left( \vVector \right)
    -
    \mathcal{D}\left( \vVector - \sVector \right)
    }_{TV}
    \enspace .
    $$
    We now use the fact that we took the bound $M$ (which is effectively the sample space of our coordinates for the vector $\vVector$) to be exponentially (in the input dimension $n$) larger than the vectors in $\Basis$ and $\tVector$. Formally, note that since the norm of $\sVector$ is $\leq \norm{ \tVector } + \norm{ \eVector_{0} } \leq \norm{ \tVector } + d$, then $\norm{ \sVector } \cdot 2^{n} \leq M$. Since the variable $\vVector$ has a uniformly random coordinates vector in $\bbZ^{n}_{\left( -2^{m}, 2^{m} - 1 \right)}$, the length $\sVector$ is tiny compared to $\vVector$ with overwhelming probability. It follows that the distribution $\vVector - \sVector$ has total variation distance bounded by $2^{-n}$ to $\vVector$, and the above sum of total variation distances is bounded by
    $$
    \norm{
    \mathcal{D}\left( \uVector \right)
    -
    \mathcal{D}\left( \uVector + \eVector_{0} \right)
    }_{TV}
    +
    2^{-n}
    \enspace .
    $$
    It remains to bound the total variation distance between $\mathcal{D}\left( \uVector \right)$ and $\mathcal{D}\left( \uVector + \eVector_{0} \right)$ for any $\eVector_{0}$ such that $\norm{ \eVector_{0} } \leq d$.

    \paragraph{Completeness of the reduction, second part (Discrete Gaussians with close centers are close in total variation distance).}
    So far in the proof of completeness, we saw that the total variation distance between the output distributions of the circuits $C^{0}$ and $C^{1}$ is bounded by 
    $$
    \norm{
    \mathcal{D} \left( \uVector \right)
    -
    \mathcal{D}\left( \uVector + \eVector_{0} \right)
    }_{TV}
    +
    \frac{ 1 }{ 4 }
    +
    2^{ -n }
    \enspace .
    $$
    Recall that $\uVector \gets D^{n}_{ \Tilde{B} }$, and Lemma \ref{lemma:wide_discrete_gaussians_are_close} tells us that
    $$
    \norm{ \mathcal{D}\left( \uVector \right) - \mathcal{D}\left( \uVector + \eVector_{0} \right) }_{TV}
    \leq
    \sqrt{
    1
    -
    e^{ -\frac{ \pi \cdot \left( \norm{ \eVector_{0} }^{2} + 2\cdot \norm{ \eVector_{0} }\cdot \sqrt{n}\cdot \Tilde{B} \right) }{ \Tilde{B}^{2} } }
    }
    \enspace .
    $$
    We have $\norm{ \eVector_{0} } \leq d$, and we previously set,
    \begin{itemize}
        \item 
        $g(n) := c_{g} \cdot n \sqrt{n}$.
    
        \item
        $B := \frac{ g(n) \cdot d }{ c_{\mathcal{N}} \cdot \sqrt{n} \cdot \sqrt{\kappa} }$, where $c_{\mathcal{N}} := 8$.

        \item 
        $\beta := \lfloor \log_{2}\left( B \right) \rfloor$.

        \item 
        $\Tilde{\sigma} := \sqrt{ \frac{ 2^{2\beta + 2} + 2^{\beta + 2} }{ 12 } }$.

        \item 
        $\Tilde{B} := \ceil{ \sqrt{\kappa} \cdot \Tilde{\sigma} }$.
    \end{itemize}
    The above implies 
    $$
    \Tilde{\sigma}
    :=
    \sqrt{ \frac{ 2^{2\beta + 2} + 2^{\beta + 2} }{ 12 } }
    >
    \sqrt{ \frac{ 4 \cdot 2^{2\beta} }{ 12 } }
    =
    2^{\beta} \cdot \frac{ 1 }{ \sqrt{3} }
    >
    \frac{2^{\beta}}{2}
    :=
    \frac{ 2^{ \lfloor \log_{2}\left( B \right) \rfloor } }{2}
    \geq
    \frac{ B }{ 4 }
    :=
    \frac{ g(n) \cdot d }{ 32 \cdot \sqrt{n} \cdot \sqrt{\kappa} }
    \enspace ,
    $$
    which in turn tells us that $\Tilde{B} > \frac{ g(n) \cdot d }{ 32 \cdot \sqrt{n} }$. We now use this lower bound for $\Tilde{B}$ for giving a lower bound for the exponent above:
    $$
    e^{ -\frac{ \pi \cdot \left( \norm{ \eVector_{0} }^{2} + 2\cdot \norm{ \eVector_{0} }\cdot \sqrt{n}\cdot \Tilde{B} \right) }{ \Tilde{B}^{2} } }
    $$
    $$
    =
    e^{
    -\frac{ \pi \cdot \norm{ \eVector_{0} }^{2} }{ \Tilde{B}^{2} }
    }
    \cdot
    e^{
    -\frac{ \pi \cdot 2 \cdot \norm{ \eVector_{0} }\cdot \sqrt{n} }{ \Tilde{B} }
    }
    $$
    $$
    >
    e^{
    -\frac{ 32^{2} \cdot n \cdot \pi \cdot d^{2} }{ g(n)^{2} \cdot d^{2} }
    }
    \cdot
    e^{
    -\frac{ 32 \cdot \sqrt{n} \cdot \pi \cdot 2 \cdot d \cdot \sqrt{n} }{ g(n) \cdot d }
    }
    $$
    $$
    =
    e^{
    -\frac{ 32^{2} \cdot n \cdot \pi }{ c_{g}^{2} \cdot n^{3} }
    }
    \cdot
    e^{
    -\frac{ 32 \cdot \sqrt{n} \cdot \pi \cdot 2 \cdot \sqrt{n} }{ c_{g} \cdot n \sqrt{n} }
    }
    $$
    $$
    =
    e^{
    -\frac{ 32^{2} \cdot \pi }{ c_{g}^{2} \cdot n^{2} }
    }
    \cdot
    e^{
    -\frac{ 32 \cdot \pi \cdot 2 }{ c_{g} \sqrt{n} }
    }
    \enspace ,
    $$
    which implies 
    $$
    \left(
    e^{ -\frac{ \pi \cdot \left( \norm{ \eVector_{0} }^{2} + 2\cdot \norm{ \eVector_{0} }\cdot \sqrt{n}\cdot \Tilde{B} \right) }{ \Tilde{B}^{2} } }
    \right)
    \geq
    \left(
    e^{
    -\frac{ 1 }{ \Omega \left( \sqrt{n} \right) }
    }
    \right)
    \enspace .
    $$
    Finally, note that for every positive constant $c \in \bbR_{> 0}$ we have the limit 
    $$
    \lim_{n \rightarrow \infty}
    \frac{
    \left(
    1
    -
    e^{
    -\frac{ c }{ \sqrt{n} }
    }
    \right)
    }
    {
    \frac{1}{\sqrt{n}}
    }
    =
    c
    \enspace .
    $$
    It follows that 
    $$
    \sqrt{
    1
    -
    e^{ -\frac{ \pi \cdot \left( \norm{ \eVector_{0} }^{2} + 2\cdot \norm{ \eVector_{0} }\cdot \sqrt{n}\cdot \Tilde{B} \right) }{ \Tilde{B}^{2} } }
    }
    \underset{(\exists c \in \bbR_{> 0})}{\leq}
    \sqrt{
    1
    -
    e^{
    -\frac{ c }{ \sqrt{n} }
    }
    }
    \leq
    \frac{ \sqrt{c} }{ n^{\frac{1}{4}} }
    \enspace .
    $$

    To summarize, in case $\left( \Basis, \tVector, d \right) \in \YES$, we have the upper bound 
    $$
    \norm{ \mathcal{D}\left( C^{0} \right) - \mathcal{D}\left( C^{1} \right) }_{TV}
    \leq
    \frac{1}{4}
    +
    2^{-n}
    +
    \frac{ \sqrt{c} }{ n^{\frac{1}{4}} }
    \leq
    \frac{1}{4}
    +
    O\left( n^{-\frac{1}{4}} \right)
    \enspace .
    $$
\end{proof}

\begin{lemma} [Discrete uniform sum is indistinguishable from discrete Gaussian] \label{lemma:distance_between_discrete_uniform_sum_and_discrete_gaussian}
    For $\kappa \in \Nat$, $\beta \in \Nat \cup \{ 0 \}$ such that $\kappa$ is an even number, let $S$ the random variable defined as the sum of $\kappa + 1$ random variables: $\kappa$ i.i.d. samples from $\mathcal{U}_{ \left( 0, 2^{\beta + 1} - 1 \right) }$ (the discrete uniform distribution on $\{ 0, 1, \cdots, 2^{\beta + 1} - 1 \}$), followed by subtracting (the zero-entropy "variable") $\frac{\kappa}{2}\cdot \left( 2^{\beta + 1} - 1 \right)$:
    $$
    S
    :=
    \sum_{i \in [\kappa]} \mathcal{U}^{\left( i \right)}_{ \left( 0, 2^{\beta + 1} - 1 \right) }
    -
    \frac{\kappa}{2}\cdot \left( 2^{\beta + 1} - 1 \right)
    \enspace .
    $$
    Then, 
    $$
    \norm{
    \mathcal{D}\left( S \right)
    -
    \mathcal{D}\left( D_{\Tilde{B}} \right)
    }_{TV}
    \leq
    O\left( \frac{1}{\sqrt{ \kappa }} \right)
    \enspace ,
    $$
    for $\Tilde{\sigma} := \sqrt{ \frac{ 2^{2\beta + 2} + 2^{\beta + 2} }{ 12 } }$, $\Tilde{B} := \ceil{ \sqrt{\kappa} \cdot \Tilde{\sigma} }$.
\end{lemma}

\begin{proof}
    We use Theorem \ref{theorem:discrete_TV_distance_berry_esseen}. We have the following parameters:
    \begin{itemize}
        \item
        The expectation of $\mathcal{U}^{(i)}_{ \left( 0, 2^{\beta + 1} - 1 \right) }$ is $\mu_{i} = \frac{2^{\beta + 1} - 1}{2}$.

        \item 
        The variance of $\mathcal{U}^{(i)}_{ \left( 0, 2^{\beta + 1} - 1 \right) }$ is: $\sigma_{i}^{2} = \frac{ \left( 2^{\beta + 1} \right)^{2} - 1 }{ 12 } = \frac{ 2^{2\beta + 2} - 1 }{ 12 } = \Theta\left( 2^{2\beta} \right)$.

        \item 
        The variance of the sum $S$ is $\sigma^{2} = \sum_{i \in [\kappa]} \sigma_{i}^{2} = \kappa \cdot \sigma_{i}^{2} = \Theta\left( \kappa \cdot 2^{2\beta} \right)$, since we are dealing with i.i.d. random variables.

        \item 
        The third moment of $\mathcal{U}^{(i)}_{ \left( 0, 2^{\beta + 1} - 1 \right) }$ is of the order:
        $$
        \gamma_{i} :=
        \frac{ \bbE_{ u \gets \mathcal{U}_{ \left( 0, 2^{\beta + 1} - 1 \right) } }
        \left(
        | u - \mu_{ i } |^{3} \right) }{ \sigma^{3} } = \Theta\left( \frac{1}{ \kappa \sqrt{\kappa} }
        \right)
        \enspace .
        $$

        \item 
        The expectation of the sum of the $\kappa$ variables is $\frac{\kappa}{2}\cdot \left( 2^{\beta + 1} - 1 \right)$, and since in the full sum $S$ we subtracted it, we nullified the expectation of $S$, i.e., $\bbE\left( S \right) := \mu := \sum_{ i \in [\kappa] } \mu_{i} - \frac{\kappa}{2}\cdot \left( 2^{\beta + 1} - 1 \right) = 0$.

        \item 
        The sum of third moments is $\gamma := \sum_{ i \in [\kappa] }\gamma_{i} = \Theta\left( \frac{1}{ \sqrt{\kappa} } \right)$.

        \item 
        For every $i \in [\kappa]$, the random variable $S^{(i)} := S - \mathcal{U}^{\left( i \right)}_{\left( 0, 2^{\beta + 1} - 1 \right)}$ is summing all $\kappa$ variables, except variable $i$.
    \end{itemize}
    The theorem's upper bound on the total variation distance implies,
    $$
    \norm{
    \mathcal{D}\left( S \right)
    -
    \mathcal{D}\left( \mathcal{N}^{ \bbZ }\left( 0, \sqrt{\kappa} \cdot \Tilde{\sigma} \right) \right)
    }_{TV}
    $$
    $$
    \leq
    \frac{3}{2} \cdot \sigma \cdot \sum_{ i \in [\kappa] }
    \left(
    \left( \gamma_{i} + \frac{ 2 \cdot \sigma_{i}^{2} }{ 3 \cdot \sigma^{3} } \right)
    \cdot
    \norm{ \mathcal{D}\left( S^{(i)} \right) - \mathcal{D}\left( S^{(i)} + 1 \right) }_{TV}
    \right)
    $$
    $$
    +
    \left( 5 + 3\sqrt{\frac{\pi}{8}} \right)\cdot \gamma
    +
    \frac{1}{\sigma \cdot 2 \sqrt{2\pi}}
    \enspace , 
    $$
    where $\Tilde{\sigma} := \sqrt{ \frac{ 2^{2\beta + 2} + 2^{\beta + 2} }{ 12 } }$.
    We analyze each of the three summands separately, from last to first. As for the third (and last) summand, $\frac{1}{\sigma \cdot 2 \sqrt{2\pi}} \leq O\left( \frac{1}{ \sqrt{\kappa} \cdot 2^{ \beta } } \right)$ because $\sigma = \Theta \left( \sqrt{\kappa} \cdot 2^{ \beta } \right)$. The second summand is $\left( 5 + 3\sqrt{\frac{\pi}{8}} \right)\cdot \gamma \leq O\left( \frac{1}{ \sqrt{\kappa} } \right)$ because $\gamma = \Theta\left( \frac{1}{ \sqrt{ \kappa } } \right)$. Overall the sum of the third and second summands together is $O\left( \frac{1}{ \sqrt{\kappa} } \right)$.

    The first summand is a sum by itself, and we have
    $$
    \frac{3}{2} \cdot \sigma \cdot \sum_{ i \in [\kappa] }
    \left(
    \left( \gamma_{i} + \frac{ 2 \cdot \sigma_{i}^{2} }{ 3 \cdot \sigma^{3} } \right)
    \cdot
    \norm{ \mathcal{D}\left( S^{(i)} \right) - \mathcal{D}\left( S^{(i)} + 1 \right) }_{TV}
    \right)
    $$
    $$
    =
    \sum_{ i \in [\kappa] }
    \left(
    \left( \frac{3}{2} \cdot \sigma \cdot \gamma_{i} + \frac{ \sigma_{i}^{2} }{ \sigma^{2} } \right)
    \cdot
    \norm{ \mathcal{D}\left( S^{(i)} \right) - \mathcal{D}\left( S^{(i)} + 1 \right) }_{TV}
    \right)
    $$
    $$
    \underset{(*)}{=}
    \kappa \cdot
    \left(
    \left( \frac{3}{2} \cdot \sigma \cdot \gamma_{1} + \frac{ \sigma_{1}^{2} }{ \sigma^{2} } \right)
    \cdot
    \norm{ \mathcal{D}\left( S^{(1)} \right) - \mathcal{D}\left( S^{(1)} + 1 \right) }_{TV}
    \right)
    $$
    $$
    =
    \kappa
    \cdot \frac{3}{2} \cdot \sigma \cdot \gamma_{1} 
    \cdot \norm{ \mathcal{D}\left( S^{(1)} \right) - \mathcal{D}\left( S^{(1)} + 1 \right) }_{TV}
    $$
    $$
    +
    \kappa 
    \cdot \frac{ \sigma_{1}^{2} }{ \sigma^{2} }
    \cdot \norm{ \mathcal{D}\left( S^{(1)} \right) - \mathcal{D}\left( S^{(1)} + 1 \right) }_{TV}
    $$
    $$
    \leq
    O\left( 
    2^{\beta}
    \cdot
    \norm{ \mathcal{D}\left( S^{(1)} \right) - \mathcal{D}\left( S^{(1)} + 1 \right) }_{TV}
    \right)
    +
    O\left( 
    \norm{ \mathcal{D}\left( S^{(1)} \right) - \mathcal{D}\left( S^{(1)} + 1 \right) }_{TV}
    \right)
    $$
    $$
    \leq
    O\left( 
    2^{\beta} \cdot \norm{ \mathcal{D}\left( S^{(1)} \right) - \mathcal{D}\left( S^{(1)} + 1 \right) }_{TV}
    \right)
    \enspace ,
    $$
    where the equality $(*)$ follows because for every $i \in [\kappa]$, $\gamma_{i} = \gamma_{1}$, $\sigma_{i} = \sigma_{1}$, and
    $$
    \norm{ \mathcal{D}\left( S^{(i)} \right) - \mathcal{D}\left( S^{(i)} + 1 \right) }_{TV}
    =
    \norm{ \mathcal{D}\left( S^{(1)} \right) - \mathcal{D}\left( S^{(1)} + 1 \right) }_{TV}
    \enspace .
    $$
    Finally, by Fact \ref{fact:moving_sum_by_1},
    $$
    2^{\beta} \cdot \norm{ \mathcal{D}\left( S^{(1)} \right) - \mathcal{D}\left( S^{(1)} + 1 \right) }_{TV}
    \leq
    O\left( \frac{1}{ \sqrt{ \kappa } } \right)
    \enspace ,
    $$
    and overall we showed so far,
    $$
    \norm{
    \mathcal{D}\left( S \right)
    -
    \mathcal{D}\left( \mathcal{N}^{ \bbZ }\left( 0, \sqrt{\kappa} \cdot \Tilde{\sigma} \right) \right)
    }_{TV}
    \leq 
    O\left( \frac{1}{ \sqrt{ \kappa } } \right)
    \enspace .
    $$
    Note that in order to finish our proof, by the triangle inequality, it will be sufficient to give an upper bound of $O\left( \frac{1}{ \sqrt{ \kappa } } \right)$ on $\norm{ \mathcal{D}\left( \mathcal{N}^{ \bbZ }\left( 0, \sqrt{\kappa} \cdot \Tilde{\sigma} \right) \right)-\mathcal{D}\left( D_{ \Tilde{B} } \right) }_{TV}$.

    Using the triangle inequality twice we get,
    $$
    \norm{
    \mathcal{D}\left( \mathcal{N}^{ \bbZ }\left( 0, \sqrt{\kappa} \cdot \Tilde{\sigma} \right) \right)
    -
    \mathcal{D}\left( D_{ \Tilde{B} } \right)
    }_{TV}
    $$
    $$
    \leq
    \norm{
    \mathcal{D}\left( \mathcal{N}^{ \bbZ }\left( 0, \sqrt{\kappa} \cdot \Tilde{\sigma} \right) \right)
    -
    \mathcal{D}\left( \mathcal{N}^{ \bbZ }\left( 0, \Tilde{B} \right) \right)
    }_{TV}
    $$
    $$
    +
    \norm{
    \mathcal{D}\left( \mathcal{N}^{ \bbZ }\left( 0, \Tilde{B} \right) \right)
    -
    \mathcal{D}\left( \mathcal{N}^{ \Tilde{B} }\left( 0, \Tilde{B} \right) \right)
    }_{TV}
    $$
    $$
    +
    \norm{
    \mathcal{D}\left( \mathcal{N}^{ \Tilde{B} }\left( 0, \Tilde{B} \right) \right)
    -
    \mathcal{D}\left( D_{ \Tilde{B} } \right)
    }_{TV}
    \enspace .
    $$
    Since $\Tilde{B} := \ceil{ \sqrt{\kappa} \cdot \Tilde{\sigma} }$ then $\Tilde{B} - 1 < \sqrt{\kappa} \cdot \Tilde{\sigma} \leq \Tilde{B}$, and thus 
    $$
    \norm{
    \mathcal{D}\left( \mathcal{N}^{ \bbZ }\left( 0, \sqrt{\kappa} \cdot \Tilde{\sigma} \right) \right)
    -
    \mathcal{D}\left( \mathcal{N}^{ \bbZ }\left( 0, \Tilde{B} \right) \right)
    }_{TV}
    <
    \norm{
    \mathcal{D}\left( \mathcal{N}^{ \bbZ }\left( 0, \Tilde{B} - 1 \right) \right)
    -
    \mathcal{D}\left( \mathcal{N}^{ \bbZ }\left( 0, \Tilde{B} \right) \right)
    }_{TV}
    \enspace .
    $$
    Finally, by Fact \ref{fact:rounded_gaussian_properties}, all of the above 3 summands are bounded by $O\left( \frac{1}{\Tilde{B}} \right) \leq O\left( \frac{1}{\sqrt{\kappa}} \right)$, as needed.
\end{proof}

In the proof above we used the following known fact about the distance between uniform sums.

\begin{fact} [Moving a large uniform sum has little difference] \label{fact:moving_sum_by_1}
    Let $S$ and $S^{\left( 1 \right)}$ as defined in the statement and proof of Lemma \ref{lemma:distance_between_discrete_uniform_sum_and_discrete_gaussian}. Then,
    $$
    \norm{
    \mathcal{D}\left( S^{\left( 1 \right)} \right)
    -
    \mathcal{D}\left( S^{\left( 1 \right)} + 1 \right) }_{TV}
    \leq 
    O\left( \frac{1}{ 2^{\beta} \cdot \sqrt{\kappa} } \right)
    \enspace .
    $$
\end{fact}

As part of this work we prove a lemma analogous to that about spheres from \cite{goldreich1998limits}. That is, discrete Gaussian distributions with (1) sufficiently close centers of mass, and (2) sufficiently large standard deviations, have the following upper bound on their total variation distance. 

\begin{lemma} [Close discrete Gaussians have bounded total variation distance] \label{lemma:wide_discrete_gaussians_are_close}
    Let $n \in \Nat$, let $B > 0$ and let $\eVector_{0} \in \bbZ^{n}$, then,
    $$
    \norm{ \mathcal{D}\left( D^{n}_{ B } \right) - \mathcal{D}\left( D^{n}_{ B } + \eVector_{0} \right) }_{TV}
    \leq
    \sqrt{
    1
    -
    e^{ -\frac{ \pi \cdot \left( \norm{ \eVector_{0} }^{2} + 2\cdot \norm{ \eVector_{0} }\cdot \sqrt{n}\cdot B \right) }{ B^{2} } }
    }
    \enspace .
    $$
\end{lemma}

\begin{proof}
    We calculate a lower bound for the fidelity between the distributions, and then draw our conclusions for the total variation distance. Let us denote by
    $$
    P
    :=
    \sum_{ \eVector \in \bbZ^{n}_{\left( -B, \cdots, -1, 0, 1, \cdots, B \right)} }
    e^{ -\frac{ \pi \cdot \norm{ \eVector }^{2} }{ B^{2} } }
    \enspace ,
    $$
    the normalization factor of the $n$-dimensional discrete truncated Gaussian distribution. With accordance to Definition \ref{definition:discrete_truncated_gaussian},
    \begin{itemize}
        \item
        The probability density function of $\mathcal{D}\left( D^{n}_{ B } \right)$ is
        $$
        \forall \eVector \in \bbZ^{n}_{\left( -B, \cdots, -1, 0, 1, \cdots, B \right)} :
        D^{n}_{ B }\left( \eVector \right)
        :=
        \frac{ e^{ -\frac{ \pi \cdot \norm{ \eVector }^{2} }{ B^{2} } } }{ P }
        \enspace .
        $$

        \item
        The probability density function of $\mathcal{D}\left( D^{n}_{ B } + \eVector_{0} \right)$ is
        $$
        \forall \eVector \in \bbZ^{n}_{\left( -B, \cdots, -1, 0, 1, \cdots, B \right)} :
        \left( D^{n}_{ B } + \eVector_{0} \right)\left( \eVector \right)
        :=
        \frac{ e^{ -\frac{ \pi \cdot \norm{ \eVector - \eVector_{0} }^{2} }{ B^{2} } } }{ P }
        =
        D^{n}_{B}\left( \eVector - \eVector_{0} \right)
        \enspace .
        $$
    \end{itemize}
    The fidelity between $\mathcal{D}\left( D^{n}_{ B } \right)$ and $\mathcal{D}\left( D^{n}_{ B } + \eVector_{0} \right)$ is as follows.
    $$
    F\left( \mathcal{D}\left( D^{n}_{ B } \right), \mathcal{D}\left( D^{n}_{ B } + \eVector_{0} \right) \right)
    :=
    \sum_{\eVector \in \bbZ^{n}_{\left( -B, \cdots, -1, 0, 1, \cdots, B \right)}}
    \sqrt{
    D^{n}_{B}\left( \eVector \right) \cdot D^{n}_{B}\left( \eVector - \eVector_{0} \right)
    }
    $$
    $$
    =
    \sum_{\eVector \in \bbZ^{n}_{\left( -B, \cdots, -1, 0, 1, \cdots, B \right)}}
    \sqrt{
    \frac{ e^{ -\frac{ \pi \cdot \norm{ \eVector }^{2} }{ B^{2} } } }{ P }
    \cdot 
    \frac{ e^{ -\frac{ \pi \cdot \norm{ \eVector - \eVector_{0} }^{2} }{ B^{2} } } }{ P }
    }
    $$
    $$
    =
    \sum_{\eVector \in \bbZ^{n}_{\left( -B, \cdots, -1, 0, 1, \cdots, B \right)}}
    \frac{1}{P}
    \cdot
    e^{ -\frac{ \pi \cdot \left( \norm{ \eVector }^{2} + \norm{ \eVector - \eVector_{0} }^{2} \right) }{ 2 \cdot B^{2} } }
    $$
    $$
    \underset{(\text{triangle inequality})}{\geq}
    \sum_{\eVector \in \bbZ^{n}_{\left( -B, \cdots, -1, 0, 1, \cdots, B \right)}}
    \frac{1}{P}
    \cdot
    e^{ -\frac{ \pi \cdot \left( \norm{ \eVector }^{2} + \left( \norm{ \eVector } + \norm{ \eVector_{0} } \right)^{2} \right) }{ 2 \cdot B^{2} } }
    $$
    $$
    =
    \sum_{\eVector \in \bbZ^{n}_{\left( -B, \cdots, -1, 0, 1, \cdots, B \right)}}
    \frac{1}{P}
    \cdot
    e^{ -\frac{ \pi \cdot \left( \norm{ \eVector }^{2} + \norm{ \eVector }^{2} + \norm{ \eVector_{0} }^{2} + 2\cdot \norm{ \eVector_{0} }\norm{ \eVector } \right) }{ 2 \cdot B^{2} } }
    $$
    $$
    \underset{\left( \forall \eVector \in \bbZ^{n}_{\left( -B, \cdots, -1, 0, 1, \cdots, B \right)} : \norm{ \eVector } \leq \sqrt{n}\cdot B \right)}{\geq}
    \sum_{\eVector \in \bbZ^{n}_{\left( -B, \cdots, -1, 0, 1, \cdots, B \right)}}
    \frac{1}{P}
    \cdot
    e^{ -\frac{ \pi \cdot \left( 2\cdot \norm{ \eVector }^{2} + \norm{ \eVector_{0} }^{2} + 2\cdot \norm{ \eVector_{0} }\cdot \sqrt{n}\cdot B \right) }{ 2 \cdot B^{2} } }
    $$
    $$
    =
    e^{ -\frac{ \pi \cdot \left( \norm{ \eVector_{0} }^{2} + 2\cdot \norm{ \eVector_{0} }\cdot \sqrt{n}\cdot B \right) }{ 2 \cdot B^{2} } }
    \cdot
    \sum_{\eVector \in \bbZ^{n}_{\left( -B, \cdots, -1, 0, 1, \cdots, B \right)}}
    \frac{1}{P}
    \cdot
    e^{ -\frac{ \pi \cdot 2 \cdot \norm{ \eVector }^{2} }{ 2 \cdot B^{2} } }
    $$
    $$
    \underset{(\text{by definition of }P)}{=}
    e^{ -\frac{ \pi \cdot \left( \norm{ \eVector_{0} }^{2} + 2\cdot \norm{ \eVector_{0} }\cdot \sqrt{n}\cdot B \right) }{ 2 \cdot B^{2} } } \enspace .
    $$
    Our proof ends by using the known relation between fidelity and total variation distance,
    $$
    \norm{ \mathcal{D}\left( D^{n}_{ B } \right) - \mathcal{D}\left( D^{n}_{ B } + \eVector_{0} \right) }_{TV}
    \leq
    \sqrt{ 1 - F\left( \mathcal{D}\left( D^{n}_{ B } \right), \mathcal{D}\left( D^{n}_{ B } + \eVector_{0} \right) \right)^{2} }
    \enspace .
    $$
\end{proof}

\paragraph{Acknowledgements}

\vspace{2mm}
\noindent
We are grateful to Scott Aaronson and to Joseph Carolan for insightful discussions on the computational power of the models suggested in this work. 

\ifllncs
\bibliographystyle{plain}
\else
\bibliographystyle{alpha}
\fi

\bibliography{bibliography,abbrev0}

\end{document}